\documentclass[12pt,a4paper]{article}

\usepackage{amsmath,amssymb,amsthm,natbib,graphicx,mathrsfs,tikz,amsfonts}
\usepackage{multirow}
\usepackage[colorlinks,citecolor=blue,urlcolor=blue]{hyperref}

\usetikzlibrary[decorations.pathreplacing]
\usetikzlibrary{patterns}

\begin{document}

\title{The Spectral Approach to Linear Rational Expectations Models}
\author{Majid M.\ Al-Sadoon\footnote{Thanks are due to Todd Walker, Bernd Funovits, Mauro Bambi, Piotr Zwiernik, Abderrahim Taamouti, Benedikt P\"{o}tscher, three anonymous referees, and seminar participants at Heriot-Watt University, Universitat Pompeu Fabra, University of Bologna, and Aarhus University.}\\Durham University Business School}

\newtheorem{lem}{Lemma}
\newtheorem{thm}{Theorem}
\newtheorem{cor}{Corollary}
\newtheorem{prop}{Proposition}
\theoremstyle{definition}
\newtheorem{exmp}{Example}
\newtheorem{defn}{Definition}
\newtheorem{alg}{Algorithm}

\renewcommand{\labelenumi}{(\roman{enumi})}

\maketitle

\abstract{
This paper considers linear rational expectations models in the frequency domain. The paper characterizes existence and uniqueness of solutions to particular as well as generic systems. The set of all solutions to a given system is shown to be a finite dimensional affine space in the frequency domain. It is demonstrated that solutions can be discontinuous with respect to the parameters of the models in the context of non-uniqueness, invalidating mainstream frequentist and Bayesian methods. The ill-posedness of the problem motivates regularized solutions with theoretically guaranteed uniqueness, continuity, and even differentiability properties.

\bigskip

\noindent JEL Classification: C10, C32, C62, E32.

\bigskip

\noindent Keywords: Linear rational expectations models, frequency domain, spectral representation, Wiener-Hopf factorization, regularization, Gaussian likelihood function.}

\newpage

\section{Introduction}
Spectral (or frequency domain) analysis of covariance stationary processes is a cornerstone of time series analysis. Since its beginnings in the late 1930s, it has benefited from being at the intersection of a number of fundamental mathematical subjects including probability, functional analysis, and complex analysis \citep{rozanov,nikolski,bingham2,bingham1}. Almost concurrently, economic theory began to focus on expectations of future earnings, prices, interest rates etc.\ as determinants of present economic activity \citep{knight,keynes,cagan}. This led to the pioneering work of \cite{muth} who proposed that expectations, rather than being arbitrary inputs into models or arbitrarily determined within models, could be made both endogenous and model-consistent, hence ``rational'' (see \cite{pesaran1987} for further historical context). Today, rational expectations models are the mainstay of business cycle research \citep{canova,dd,hs}. This paper attempts to bridge the gap between the two strands of literature.

Classical linear systems, such as vector autoregressive moving average (VARMA) models, are linear transformations from an input process to an output process, with present values of the output depending linearly on present and past values of the input as well as past values of the output in every time period. The linear rational expectations model (LREM) class extends classical linear systems by allowing linear dependence on expectations of future values of the output as well. Such models arise naturally from the inter-temporal optimization problems of households and firms in economic modeling. Spectral analysis has focused almost entirely on classical linear systems \citep{bd,pourahmadi,lp}. Although a number of works have considered LREMs in the frequency domain \citep{whiteman,onatski,tanwalker,tan,meyer}, none have attempted a general account that parallels the aforementioned textbook treatments. Thus, the first aim of this paper is to situate LREMs in the frequency domain literature at a high level of generality that extends the aforementioned textbook treatments.

Using the spectral representation of covariance stationary processes due to \cite{kolmogorov1,kolmogorov2,kolmogorov3} and \cite{cramer1,cramer2}, this paper recasts the LREM problem in the classical Hilbert space of the frequency domain literature. This is demonstrated concretely on simple scalar models before generalizing to multivariate models. Whereas classical spectral analysis focused entirely on the backward shift operator, the new spectral analysis of LREMs requires the introduction of a new operator associated with expectations. The LREM problem then reduces to a linear system in Hilbert space. However, unlike the special case of VARMA models, LREMs cannot be solved by simply inverting a polynomial matrix due to the presence of the new expectation operators. Solving the LREM problem requires using the method of Wiener-Hopf factorization \citep{wh,gf,cg}. This paper characterizes existence and uniqueness of solutions to particular as well as generic LREMs, generalizing results by \cite{onatski}. The set of all solutions to a given LREM is shown to be a finite dimensional affine space in the frequency domain. The dimension of this space is expressed much more simply than in \cite{funovits,funovits2020}. It is important to note that the underlying assumptions in this paper are weaker than in the previous literature \citep{whiteman,onatski,tanwalker,tan,meyer}, which requires the exogenous process to have a purely non deterministic \cite{wold} representation, an assumption that is demonstrably unnecessary. The weaker assumptions of this paper also permit a clear answer as to why unit roots must be excluded, an aspect of the theory absent from the previous literature.

The main results of the paper concern the ill-posedness of the LREM problem in macroeconometrics. \cite{hadamard} defines a problem to be well-posed if its solutions satisfy the conditions of existence, uniqueness, and continuous dependence on its parameters. The LREM problem is ill-posed because it violates not just the second condition but also the third. Indeed, it has long been accepted that non-uniqueness is a feature of the LREM problem and many techniques have been developed to select a solution for any given LREM exhibiting non-uniqueness (e.g.\  \cite{taylor}, \cite{msv}, \cite{sunspots}, \cite{fanelli}, \cite{farmeretal}, \cite{bianchinicolo}). This paper highlights the fact that selections that have been proposed in the literature are not guaranteed to be continuous with respect to parameters of the LREM in the context of non-uniqueness. This problem seems to be either not well understood or not fully appreciated so far; to the author's knowledge, \cite{generic} is the only acknowledgement of this problem.

The problem of discontinuity is quite serious because it invalidates mainstream econometric methodology as reviewed, for example, in \cite{canova}, \cite{dd}, or \cite{hs}. This methodology takes a black-box approach to estimation and inference, whereby data is used to construct certain objective functions (likelihood functions, GMM criterion functions, or posterior distributions), which are then either numerically optimized by the Newton-Raphson algorithm or explored by the random walk Metropolis-Hastings algorithm (or the many variations of such algorithms) to produce estimates of parameters as well as confidence or credible intervals. When the discontinuous selected solutions are fed into this methodology, the aforementioned objective functions become discontinuous, invalidating assumptions required for the aforementioned methodologies to work. This paper illustrates concretely how things can go wrong with a simple example in Section \ref{sec:application}.

Fortunately, the literature on ill-posed problems offers an immediate solution to the problem: regularization. The idea here is that when theory is insufficient to pin down a unique solution, other information can be brought to bear. The method can be interpreted in at least two ways: (i) penalizing economically unreasonable solutions or shrinking towards economically reasonable ones or (ii) imposing prior beliefs on the frequencies of fluctuations that solutions ought to exhibit. For example, we may like to avoid solutions where certain variables vary too wildly or impose the prior that solutions to a business cycle model should exhibit fluctuations of period between 4 and 32 quarters in quarterly data. The paper provides conditions for existence and uniqueness of regularized solutions and proves that they are continuously (even differentiably) dependent on their parameters. Thus, under non-uniqueness, regularization selects solutions that can be used in any mainstream econometric method, frequentist or Bayesian.

This work is related to several recent strands in the literature. \cite{kn}, \cite{tq17}, \cite{kk18}, \cite{ident}, and \cite{kk23} study the identification of LREMs based on the spectral density of observables. \cite{cv}, \cite{tq12qe}, \cite{sala15} utilize spectral domain methods for estimating LREMs using ideas that go back to \cite{hs80}. \cite{recoverability} study the problem of subordination (what they call ``recoverability'') in the context of macroeconometric models. \cite{linsys} utilizes a generalization of Wiener-Hopf factorization in order to study unstable and non-stationary solutions ot LREMs. \cite*{ephermidze} provide recent results on the continuity of spectral factorization. Finally, \cite{regular} provides numerical algorithms for computing regularized solutions.

This paper is organized as follows. Section \ref{sec:notation} sets up the notation and reviews the fundamental concepts of spectral analysis of time series. Section \ref{sec:examples} considers the solution of simple LREMs with elementary frequency domain methods. Section \ref{sec:WHF} introduces Wiener-Hopf factorization. Section \ref{sec:eu} sets up the LREM problem, its existence and uniqueness properties, and establishes its ill-posedness. Section \ref{sec:regular} introduces regularized solutions. Section \ref{sec:application} provides an application of the theory. Section \ref{sec:conclusion} concludes. Sections \ref{sec:setsol}-\ref{sec:proofs} comprise the Appendix.

\section{Notation and Review}\label{sec:notation}
We will denote by $\mathbb{Z}$, $\mathbb{R}$, and $\mathbb{C}$ the sets of integers, real numbers, and complex numbers respectively. We will denote by $\mathbb{T}=\{z\in\mathbb{C}:|z|=1\}$ the unit circle in $\mathbb{C}$.  The letters $i$, $j$, $n$, and $m$ will always stand for natural numbers. By $\mathbb{C}^{m\times n}$ we will denote the set of $m\times n$ matrices of complex numbers. We will use $I_n$ to denote the identity $n\times n$ matrix. For $M\in\mathbb{C}^{m\times m}$, $\mathrm{tr}(M)=\sum_{i=1}^m M_{ii}$. For $M\in\mathbb{C}^{m\times n}$, $M^\ast\in\mathbb{C}^{n\times m}$ is the conjugate transpose of $M$, and $\|M\|_{\mathbb{C}^{m\times n}}=\left(\mathrm{tr}(MM^\ast)\right)^{1/2}=\left(\sum_{i=1}^m\sum_{j=1}^n|M_{ij}|^2\right)^{1/2}$.

All random variables in this paper are defined over a single probability space $(\Omega,\mathscr{F},P)$. For a complex random variable $x$, the expectation is denoted by $\boldsymbol{E}x=\int_\Omega x(\omega)dP(\omega)$. The space $\mathscr{L}_2$ is defined as the Hilbert space of complex valued random variables $x$, modulo $P$-almost sure equality, such that $\boldsymbol{E}|x|^2<\infty$ with inner product and norm
\begin{align*}
\langle x,y\rangle=\boldsymbol{E}x\overline{y},\qquad \|x\|_{\mathscr{L}_2}^2=\langle x,x\rangle,\qquad x,y\in \mathscr{L}_2.
\end{align*}
Similar to other Hilbert spaces we will consider in this paper, $\mathscr{L}_2$ is a set of equivalence classes of functions and not a set of functions. However, mathematical convention ``relegate[s] this distinction to the status of a tacit understanding'' \cite[p.\ 67]{xrudin}. Thus, we will write ``$x=y$'' instead of ``$x(\omega)=y(\omega)$ for $P$-almost all $\omega\in\Omega$'' and similarly for elements of all of the other Hilbert spaces we consider in this paper. The Hilbert space $\mathscr{L}_2^n$ is defined as the $n$-fold Cartesian product, $\mathscr{L}_2\times\cdots\times \mathscr{L}_2$ consisting of column vectors $x=\left[\begin{smallmatrix} x_1\\ \vdots\\ x_n\end{smallmatrix}\right]$, $x_j\in \mathscr{L}_2$, $j=1,\ldots, n$, with the inner product and norm
\begin{align*}
\langle\langle x,y\rangle\rangle=\sum_{j=1}^n\langle x_j,y_j\rangle,\qquad\|x\|_{\mathscr{L}_2^n}^2=\langle\langle x,x\rangle\rangle,\qquad x,y\in \mathscr{L}_2^n.
\end{align*}
If $\mathscr{S}\subset\mathscr{L}_2$ is a closed subspace and $x\in\mathscr{L}_2^n$, then the minimum of
\begin{align*}
\|x-y\|_{\mathscr{L}_2^n}^2=\sum_{j=1}^n\|x_j-y_j\|_{\mathscr{L}_2}^2,
\end{align*}
with respect to $y\in\mathscr{S}^n=\mathscr{S}\times\cdots\times\mathscr{S}$ is attained by minimizing each term on the right hand side independently with respect to $y_j\in\mathscr{S}$, $j=1,\ldots, n$. It follows that the orthogonal projection of $x\in\mathscr{L}_2^n$ onto $\mathscr{S}^n$ in $\mathscr{L}_2^n$, denoted by $\boldsymbol{P}(x|\mathscr{S}^n)$, is given by $\left[\begin{smallmatrix} \boldsymbol{P}(x_1|\mathscr{S})\\ \vdots\\ \boldsymbol{P}(x_n|\mathscr{S})\end{smallmatrix}\right]$, where $\boldsymbol{P}(x_j|\mathscr{S})$ is the orthogonal projection of $x_j$ onto $\mathscr{S}$ in $\mathscr{L}_2$, $j=1,\ldots, n$.

Let $\ldots,\xi_{-1},\xi_0,\xi_1,\ldots$ be a doubly infinite sequence in $\mathscr{L}_2^n$. We will refer to this stochastic process simply as $\xi$. The $j$-th element of $\xi_t$ will be denoted by $\xi_{jt}$. If for all $t,s\in\mathbb{Z}$, $\boldsymbol{E}\xi_t=\boldsymbol{E}\xi_s$ and $\boldsymbol{E}\xi_t\xi_s^\ast$ depends on $t$ and $s$ only through $t-s$, we say that $\xi$ is a covariance stationary process. Given $\xi$, we may define a number of useful objects.

Let $\mathscr{H}$ be the closure in $\mathscr{L}_2$ of the set of all finite complex linear combinations of the set $\{\xi_{jt}: j=1,\ldots,n, t\in\mathbb{Z}\}$. We define $\mathscr{H}^n\subset \mathscr{L}_2^n$ to be the $n$-fold Cartesian product of $\mathscr{H}$.

Traditionally, spectral analysis utilizes the forward (rather than backward) shift operator,
\begin{align*}
\boldsymbol{U}:\xi_{jt}\mapsto \xi_{jt+1},\qquad j=1,\ldots, n,\quad t\in\mathbb{Z}.
\end{align*}
This operator extends uniquely to a unitary operator on $\mathscr{H}$ \cite[Theorem I.4.1]{rozanov}.

We may next define the unique spectral measure $F$ on the unit circle, $\mathbb{T}$, which satisfies
\begin{align*}
\boldsymbol{E}\xi_t\xi_s^\ast=\int z^{t-s}dF,\qquad t,s\in\mathbb{Z},
\end{align*}
\cite[Theorems I.5.1 and I.5.2]{rozanov}. Many textbooks express the integral above as\linebreak $\int_{-\pi}^\pi e^{\mathrm{i}(t-s)\lambda}dF(\lambda)$; the notation here is substantially more compact. When each element of $F$ is absolutely continuous with respect to normalized Lebesgue measure on $\mathbb{T}$, defined by
\begin{align*}
\mu\left(\{e^{\mathrm{i}\lambda}:a\leq\lambda\leq b\}\right)=\frac{1}{2\pi}(b-a),\qquad 0\leq b-a\leq 2\pi,
\end{align*}
then the Radon-Nikod\'{y}m derivative $dF/d\mu$ is the spectral density matrix of $\xi$. Note that the spectral density of $n$-dimensional standardized white noise is $I_n$. Another important example of a spectral measure, is the Dirac measure at $w\in\mathbb{T}$ defined as
\begin{align*}
\delta_w(\Lambda)=\begin{cases}
1, &\;w\in\Lambda,\\
0, &\;w\not\in\Lambda,
\end{cases}
\end{align*}
for Borel subsets $\Lambda\subset\mathbb{T}$. Note that $F=\delta_w I_n$ is the spectral measure of the purely deterministic process $\xi_t=\xi_0 w^t$ for $t\in\mathbb{Z}$, with $\boldsymbol{E}\xi_0\xi_0^\ast=I_n$.

We may also define the Hilbert space $H$ of all Borel-measurable mappings $\phi:\mathbb{T}\rightarrow\mathbb{C}^{1\times n}$ (i.e.\ $\phi(z)$ is a row vector for $z\in\mathbb{T}$) such that $\int \phi dF\phi^\ast<\infty$ with inner product and norm
\begin{align*}
(\phi,\varphi)=\int \phi dF\varphi^\ast,\qquad\|\phi\|^2_H=(\phi,\phi),\qquad \phi,\varphi\in H,
\end{align*}
where we identify $\phi=\varphi$ if $\|\phi-\varphi\|_H^2=\int(\phi-\varphi)dF(\phi-\varphi)^\ast=0$ \citep[Lemma I.7.1]{rozanov}. We define the Hilbert space $H^m=H\times\cdots\times H$ to be the $m$-fold Cartesian product of $H$ consisting of the $m\times n$ matrices $\phi=\left[\begin{smallmatrix} \phi_1\\ \vdots\\ \phi_m\end{smallmatrix}\right]$, $\phi_j\in H$, $j=1,\ldots, m$ endowed with the inner product and norm
\begin{align*}
((\phi,\varphi))=\sum_{j=1}^m(\phi_j,\varphi_j),\qquad\|\phi\|^2_{H^m}=((\phi,\phi)),\qquad \phi,\varphi\in H^m.
\end{align*}
By the same argument as used earlier, if $S\subset H$ is a closed subspace, the orthogonal projection of $\phi\in H^m$ onto $S^m$ in $H^m$, denoted by $\boldsymbol{P}(\phi|S^m)$, is given by $\left[\begin{smallmatrix} \boldsymbol{P}(\phi_1|S)\\ \vdots\\ \boldsymbol{P}(\phi_m|S)\end{smallmatrix}\right]$, where $\boldsymbol{P}(\phi_j|S)$ is the orthogonal projection of $\phi_j$ onto $S$ in $H$, $j=1,\ldots, m$.

The spectral representation theorem states that
\begin{align*}
\xi_t=\int z^td\Phi,\qquad t\in\mathbb{Z},
\end{align*}
for a column vector of random measures $\Phi$ with $\Phi(\Lambda)\in \mathscr{H}^n$ and $\boldsymbol{E}\Phi(\Lambda)\Phi^\ast(\Lambda)=F(\Lambda)$ for Borel subsets $\Lambda\subset\mathbb{T}$ \cite[Theorem I.4.2]{rozanov}. Many textbooks express the integral above equivalently as $\int_{-\pi}^\pi e^{\mathrm{i}t\lambda}dZ(\lambda)$, where $Z(\lambda)=\Phi(\{e^{\mathrm{i}\tau}:\tau\in(-\pi,\lambda]\})$ for $\lambda\in(-\pi,\pi]$; again, the notation here is substantially more compact. The spectral representation theorem establishes a unitary mapping $H^m\rightarrow \mathscr{H}^m$ defined by $\phi\mapsto h=\int \phi d\Phi$ \citep[p.\ 32]{rozanov}. We call $\phi$ the spectral characteristic of $h$. Thus,
\begin{align*}
\left\langle\left\langle \int \phi d\Phi,\int \varphi d\Phi\right\rangle\right\rangle=((\phi,\varphi)),\qquad \left\|\int \phi d\Phi\right\|^2_{\mathscr{L}_2^m}=\|\phi\|^2_{H^m},\qquad \phi,\varphi\in H^m.
\end{align*}

Denote by $e_j\in\mathbb{C}^{1\times n}$ the row vector whose elements are all equal to zero except for the $j$-th which is equal to one. The spectral representation theorem implies that $\mathscr{H}_t$, the closure in $\mathscr{L}_2$ of the set of all finite complex linear combinations of $\{\xi_{js}: j=1,\ldots,n, s\leq t\}$, is in correspondence with the closure of the set of all finite complex linear combinations of $\{z^se_j: j=1,\ldots,n, s\leq t\}$ in $H$, denoted $H_t$. Let $\mathscr{H}^m_t$ and $H_t^m$ be $m$-fold Cartesian products of $\mathscr{H}_t$ and $H_t$ respectively. This implies that
\begin{align*}
\boldsymbol{P}\left(\int\phi d\Phi \bigg|\mathscr{H}^m_t\right)=\int \boldsymbol{P}(\phi|H_t^m)d\Phi,\qquad t\in\mathbb{Z}.
\end{align*}
Thus, the best linear prediction of $\xi_{t+s}$ in terms of $\xi_t, \xi_{t-1},\ldots$ is given by
\begin{align*}
\boldsymbol{P}(\xi_{t+s}|\mathscr{H}_t^n)=\int \boldsymbol{P}(z^{t+s}I_n|H_t^n)d\Phi,\qquad s,t\in\mathbb{Z}.
\end{align*}
It is easily established that $\mathscr{H}_t=\boldsymbol{U}^t\mathscr{H}_0$ and $H_t=z^tH_0$ for all $t\in\mathbb{Z}$ so that $\boldsymbol{P}(\boldsymbol{U}^t\xi|\mathscr{H}_t)=\boldsymbol{U}^t\boldsymbol{P}(h|\mathscr{H}_0)$ for all $t\in\mathbb{Z}$ and $h\in \mathscr{H}$ and likewise $\boldsymbol{P}(z^t\phi|H_t)=z^t\boldsymbol{P}(\phi|H_0)$ for all $t\in\mathbb{Z}$ and $\phi\in H$ \citep[p.\ 52]{rozanov}. Thus
\begin{align*}
\boldsymbol{P}(\xi_{t+s}|\mathscr{H}^n_t)=\int z^t\boldsymbol{P}\left(z^s I_n|H_0^n\right)d\Phi,\qquad s,t\in\mathbb{Z}.
\end{align*}

If $\zeta$ is an additional covariance stationary processes such that for every $t,s\in\mathbb{Z}$, $\boldsymbol{E}\zeta_s\xi_t^\ast$ depends on $t$ and $s$ only through $t-s$, we will say that $\zeta$ is causal in $\xi$ if $\zeta_0\in \mathscr{H}_0^m$. Finally, if $\nu$ is causal in $\xi$ and satisfies $\boldsymbol{P}(\nu_{t+1}|\mathscr{H}_t)=0$ for all $t\in\mathbb{Z}$, we call $\nu$ an innovation process.

\section{Examples}\label{sec:examples}
Armed with the basic machinery above, we now make a first attempt at solving LREMs in the frequency domain. We will see that solving simple univariate LREMs involves only elementary spectral domain techniques as discussed in textbook treatments of spectral analysis  such as \cite{bd} or \cite{pourahmadi}. The methods also provide strong hints to the general approach to solving LREMs. In this section, $\xi$ is a scalar covariance stationary process with $\boldsymbol{U}$, $\Phi$, $F$, $\mathscr{H}$, and $H$ defined as in the previous section.

\subsection{The Autoregressive Model}
We begin on familiar territory with the stationary autoregression
\begin{align}
X_t-\alpha X_{t-1}=\xi_t,\qquad t\in\mathbb{Z},\label{eq:ar1}
\end{align}
where $|\alpha|<1$. The frequency-domain analysis of this model is available in many textbooks. For completeness, we provide a treatment here that is geared towards understanding the more general cases to come.

We require a covariance stationary solution causal in $\xi$ (further motivation of this restriction can be found in Section \ref{sec:eu}). Thus, we require
\begin{align*}
X_t=\int z^t\phi d\Phi,\qquad t\in\mathbb{Z},
\end{align*}
for some spectral characteristic $\phi\in H_0$.

Notice that we may restrict attention to the equation
\begin{align*}
(1-\alpha \boldsymbol{U}^{-1})X_0=\xi_0.
\end{align*}
If a solution $X_0$ exists, then the rest of the process can be generated as
\begin{align*}
X_t=\boldsymbol{U}^tX_0,\qquad t\in\mathbb{Z},
\end{align*}
and clearly satisfies \eqref{eq:ar1}.

Thus, we must solve for $\phi$ in
\begin{align*}
\int (1-\alpha z^{-1})\phi d\Phi=\int d\Phi.
\end{align*}
Since integration with respect to $\Phi$ is a unitary mapping $H\rightarrow\mathscr{H}$, it has an inverse and the equation above is equivalent to
\begin{align}
(1-\alpha z^{-1})\phi=1\label{eq:ar1freq}
\end{align}
in the frequency domain.

The linear mapping $\phi\mapsto \alpha z^{-1} \phi$ on $H_0$ is bounded in norm by $|\alpha|<1$, thus the mapping $\phi\mapsto (1-\alpha z^{-1})\phi$ is invertible and
\begin{align*}
\phi=\sum_{s=0}^\infty \alpha^s z^{-s}
\end{align*}
is the unique solution to \eqref{eq:ar1freq}, where the summation is understood to converge with respect to $H$ norm \cite[Theorem 2.8.1]{basic}. Thus, we arrive at the unique solution to $X_0$,
\begin{align*}
X_0&=\int \sum_{s=0}^\infty \alpha^s z^{-s} d\Phi=\sum_{s=0}^\infty \alpha^s \int z^{-s}d\Phi=\sum_{s=0}^\infty \alpha^s \xi_{-s}.
\end{align*}
The interchange of the summation and the stochastic integral is admissible because the stochastic integral is a bounded linear operator mapping from $H$ to $\mathscr{H}$ and the inner summation converges in $H$. It follows that the unique covariance stationary solution is
\begin{align*}
X_t=\boldsymbol{U}^t X_0=\sum_{s=0}^\infty \alpha^s \boldsymbol{U}^t\xi_{-s}=\sum_{s=0}^\infty \alpha^s \xi_{t-s},\qquad t\in\mathbb{Z}.
\end{align*}
The operator $\boldsymbol{U}^t$ is interchangeable with the summation because it is a bounded linear operator on $\mathscr{H}$ and the summation converges in $\mathscr{L}_2$.

\subsection{The Cagan Model}
The Cagan model is given as,
\begin{align}
X_t-\beta \boldsymbol{P}(X_{t+1}|\mathscr{H}_t)=\xi_t,\qquad t\in\mathbb{Z},\label{eq:cagan}
\end{align}
with $|\beta|<1$. Again, we look for a covariance stationary solution causal in $\xi$ and we restrict attention to the equation
\begin{align*}
X_0-\beta \boldsymbol{P}(UX_0|\mathscr{H}_0)=\xi_0,
\end{align*}
and, following a similar argument to that used above, we arrive at the underlying frequency domain problem,
\begin{align}
\phi-\beta \boldsymbol{P}(z\phi|H_0)=1.\label{eq:0func}
\end{align}

Since the linear mapping $\phi\mapsto \beta \boldsymbol{P}(z\phi|H_0)$ on $H_0$ is bounded in norm by $|\beta|<1$, we see that the mapping $\phi\mapsto \phi-\beta \boldsymbol{P}(z\phi|H_0)$ is invertible so that
\begin{align*}
\phi=\sum_{s=0}^\infty \beta^s \boldsymbol{P}(z^s|H_0)
\end{align*}
is the unique solution to \eqref{eq:0func}, where the summation is understood to converge with respect to $H$ norm \cite[Theorem 2.8.1]{basic}. Thus, we arrive at the unique solution to $X_0$,
\begin{align*}
X_0&=\int \sum_{s=0}^\infty \beta^s \boldsymbol{P}(z^s|H_0)d\Phi=\sum_{s=0}^\infty \beta^s \int \boldsymbol{P}(z^s|H_0)d\Phi=\sum_{s=0}^\infty \beta^s \boldsymbol{P}(\xi_s|\mathscr{H}_0).
\end{align*}
The interchange of the summation and the stochastic integral is admissible because the stochastic integral is a bounded linear operator mapping $H$ to $\mathscr{H}$ and the inner summation converges in $H$. It follows that the unique covariance stationary solution is
\begin{align*}
X_t=\boldsymbol{U}^t X_0=\sum_{s=0}^\infty \beta^s \boldsymbol{U}^t\boldsymbol{P}(\xi_s|\mathscr{H}_0)=\sum_{s=0}^\infty \beta^s \boldsymbol{P}(\xi_{t+s}|\mathscr{H}_t),\qquad t\in\mathbb{Z}.
\end{align*}
The operator $\boldsymbol{U}^t$ is interchangeable with the summation because it is a bounded linear operator on $\mathscr{H}$ and the summation converges in $\mathscr{L}_2$.

\subsection{The Mixed Model}
Now suppose we have the more general model
\begin{align}
a \boldsymbol{P}(X_{t+1}|\mathscr{H}_t)+bX_t+c X_{t-1}=\xi_t,\qquad t\in\mathbb{Z},\label{eq:lremsimple}
\end{align}
where $a,b,c\in\mathbb{C}$ and $ac\neq0$. This leads to the frequency-domain equation
\begin{align}
\boldsymbol{P}((az+b+cz^{-1})\phi|H_0)=1.\label{eq:lremsimplefreq}
\end{align}
As noted by \cite{sargent79}, the solution to this system depends on the factorization of
\begin{align*}
M(z)=az+b+cz^{-1}=a z^{-1}(z-\delta)(z-\gamma).
\end{align*}
We assume, without loss of generality, that $|\gamma|\leq|\delta|$. There are four cases to consider.

Suppose $|\gamma|<1<|\delta|$. We can then write $M(z)=a(z-\delta)(1-\gamma z^{-1})$ and express the system as
\begin{align*}
\boldsymbol{P}((1-\delta^{-1}z)\varphi|H_0)=1,\qquad -a\delta(1-\gamma z^{-1})\phi=\varphi.
\end{align*}
The first equation can be solved as in the Cagan model (since $|\delta^{-1}|<1$), while the second can be solved as in the autoregressive model (since $|\gamma|<1$). This procedure leads us to the unique solution
\begin{align*}
\phi=-\frac{1}{a\delta}\sum_{u=0}^\infty\sum_{s=0}^\infty \gamma^u \delta^{-s} z^{-u}\boldsymbol{P}(z^s|H_0)=-\frac{1}{a\delta}\sum_{u=0}^\infty\sum_{s=0}^\infty \gamma^u \delta^{-s} \boldsymbol{P}(z^{s-u}|H_{-u}).
\end{align*}
In the time domain, we obtain the following solution
\begin{align*}
X_0=-\frac{1}{a\delta}\sum_{u=0}^\infty\sum_{s=0}^\infty \gamma^u \delta^{-s} \boldsymbol{P}(\xi_{s-u}|\mathscr{H}_{-u}),
\end{align*}
which then gives us the general solution,
\begin{align*}
X_t=-\frac{1}{a\delta}\sum_{u=0}^\infty\sum_{s=0}^\infty \gamma^u \delta^{-s} \boldsymbol{P}(\xi_{s+t-u}|\mathscr{H}_{t-u}),\qquad t\in\mathbb{Z}.
\end{align*}
Thus, when $|\gamma|<1<|\delta|$, there exists a unique solution.

Next, suppose that $|\delta|<1$. Then we may write $M(z)=az(1-\delta z^{-1})(1-\gamma z^{-1})$ and express our system as
\begin{align*}
\boldsymbol{P}(z\varphi|H_0)=1,\qquad a(1-\delta z^{-1})(1-\gamma z^{-1})\phi=\varphi.
\end{align*}
The first equation does not have a unique solution in general. For example, when $F = \mu$, then $\varphi=z^{-1}+\psi$ solves the equation for any $\psi\in\mathbb{C}$. More generally, every solution is
of the form
\begin{align*}
\varphi=z^{-1}+\psi,
\end{align*}
where $\psi\in H_0$ with $z\psi$ orthogonal to $H_0$. We can then solve for $\phi$ as
\begin{align*}
\phi=\frac{1}{a}\sum_{u=0}^\infty\sum_{s=0}^\infty\gamma^u\delta^sz^{-s-u}(z^{-1}+\psi).
\end{align*}
In the time domain, this leads to the general solution
\begin{align*}
X_t=\frac{1}{a}\sum_{u=0}^\infty\sum_{s=0}^\infty \gamma^u\delta^s(\xi_{t-s-u-1}+\nu_{t-s-u}),\qquad t\in\mathbb{Z},
\end{align*}
where $\nu_t=\int z^t\psi d\Phi\in \mathscr{H}_t$ for $t\in\mathbb{Z}$ is an arbitrary innovation process. Indeed, $\nu_t \in \mathscr{H}_t$ for all $t \in \mathbb{Z}$ because $\psi \in H_0$ and $\nu_{t+1}$ is orthogonal to $\mathscr{H}_t$ for all $t \in \mathbb{Z}$ because $\boldsymbol{P} (\nu_{t+1}|\mathscr{H}_t) = \int\boldsymbol{P} (z^{t+1}\psi|H_t)d\Phi = \int z^t\boldsymbol{P} (z\psi|H_0)d\Phi = 0$ as $z\psi$ is orthogonal to $H_0$. Thus, when $|\delta|<1$, there exist potentially infinitely many solutions.

Now suppose $|\gamma|>1$. Then we may write $M(z)=a\delta\gamma z^{-1}(1-\delta^{-1}z)(1-\gamma^{-1}z)$ and
\begin{align*}
\boldsymbol{P}(a\delta\gamma(1-\delta^{-1}z)(1-\gamma^{-1}z)\varphi|H_0)=1,\qquad z^{-1}\phi=\varphi.
\end{align*}
Clearly $\varphi$ can be solved as in the Cagan model to produce
\begin{align*}
\varphi=\frac{1}{a\delta\gamma}\sum_{u=0}^\infty\sum_{s=0}^\infty\gamma^{-u}\delta^{-s} \boldsymbol{P}(z^{s+u}|H_0).
\end{align*}
This implies that
\begin{align*}
z^{-1}\phi=\frac{1}{a\delta\gamma}\sum_{u=0}^\infty\sum_{s=0}^\infty\gamma^{-u}\delta^{-s} \boldsymbol{P}(z^{s+u}|H_0).
\end{align*}
However, this equation cannot hold generally. To see this, let $\xi$ be a standard white noise so that $F=\mu$. Then $\{z^s: s \in \mathbb{Z}\}$ is an orthonormal set and the equation above reduces to
\begin{align*}
z^{-1}\phi=\frac{1}{a\delta\gamma}.
\end{align*}
Since $\phi$ is in the closure of the linear span of $\{z^s:s\leq0\}$, the zero-th Fourier coefficient of the left hand side is equal to zero, while that of the right hand side is equal to $\frac{1}{a\delta\gamma}\neq0$. So there is no solution in general when $|\gamma|>1$.

Finally, suppose either $|\gamma|=1$ or $|\delta|=1$ so that $M(w)=0$ for some $w\in\mathbb{T}$. Then there is no solution in general, in the sense that there exist processes $\xi$ for which no covariance stationary solution $X$ can be found. To see this, let $F=\delta_w$, the Dirac measure at $w$. If a solution $\phi\in H_0$ to \eqref{eq:lremsimplefreq} exists, then it must satisfy
\begin{align*}
\|\phi\|_{H_0}^2=\int |\phi|^2dF=|\phi(w)|^2<\infty.
\end{align*}
This then implies that $M\phi=M(w)\phi(w)=0$ in $H$, which implies that $\boldsymbol{P}(M\phi|H_0)=0$, contradicting the fact that $\boldsymbol{P}(M\phi|H_0)=1$. Thus, when $|\gamma|=1$ or $|\delta|=1$ there exists no solution in general.

We will see that the analysis of the simple models above contains most of the elements necessary for the general multivariate case with more than one lead and/or lag and arbitrary covariance stationary $\xi$. The reader interested in a more complete elementary analysis of the simple models above (e.g.\ considering the case $|\beta|>1$ of the Cagan model) is directed to previous versions of this paper available on arXiv.

\section{Wiener-Hopf Factorization}\label{sec:WHF}
The approach we have taken in the last section is well understood in the theory of convolution equations \citep{gf}. The requisite factorization of $M(z)$ into two parts, a part to solve like the Cagan model and a part to solve like the autoregressive model, is known as a Wiener-Hopf factorization \citep{wh}. In this section, we state the basic concepts and properties of Wiener-Hopf factorization that we will need.

\begin{defn}\label{defn:w}
Let $\mathcal{W}$ be the class of functions $M:\mathbb{T}\rightarrow\mathbb{C}$ defined by
\begin{align}
M(z)=\sum_{s=-\infty}^\infty M_sz^s,\label{eq:W}
\end{align}
where $M_s\in\mathbb{C}$ for $s\in\mathbb{Z}$ and $\sum_{s=-\infty}^\infty |M_s|<\infty$. Define $\mathcal{W}_\pm\subset\mathcal{W}$ to be the class of functions \eqref{eq:W} with $M_s=0$ for $s\lessgtr0$. The sets $\mathcal{W}^{m\times n}$, $\mathcal{W}^{m\times n}_\pm$ are defined as the sets of matrices of size $m\times n$ populated by elements of $\mathcal{W}$ and $\mathcal{W}_\pm$ respectively.
\end{defn}

The class of functions $\mathcal{W}$ is known as the Wiener algebra in the functional analysis literature \cite[Section XXIX.2]{classes2}. Note that every function analytic in a neighborhood of $\mathbb{T}$ defines an element of $\mathcal{W}$ but the opposite inclusion does not hold (e.g.\ $\sum_{s=0}^\infty s^{-2}z^s\in\mathcal{W}_+$ diverges outside $\mathbb{T}$). For $M\in\mathcal{W}^{m\times n}$, we also have that
\begin{align*}
\|M\|_\infty\leq\sum_{s=-\infty}^\infty \|M_s\|_{\mathbb{C}^{m\times n}}<\infty,
\end{align*}
where $\|M\|_\infty$ is the $\mu$--essential supremum of $\|M(z)\|_{\mathbb{C}^{m\times n}}$. We will also need the operator
\begin{align*}
\left[\sum_{s=-\infty}^\infty M_sz^s\right]_-=\sum_{s=-\infty}^0 M_sz^s
\end{align*}
defined over $\mathcal{W}^{m\times n}$ and, more generally, over the class of square summable series satisfying $\sum_{s=-\infty}^\infty \|M_s\|_{\mathbb{C}^{n\times m}}^2<\infty$ \citep[Corollary XXIII.3.3]{classes2}.

\begin{defn}\label{defn:whf}
A Wiener-Hopf factorization of $M\in\mathcal{W}^{m\times m}$ is a factorization
\begin{align}
M=M_+M_0M_-,\label{eq:whf}
\end{align}
where $M_+\in\mathcal{W}^{m\times m}_+$, $M_+^{-1}\in\mathcal{W}^{m\times m}_+$, $M_-\in\mathcal{W}^{m\times m}_-$, $M_-^{-1}\in\mathcal{W}^{m\times m}_-$, $M_0$ is a diagonal matrix with diagonal elements $z^{\kappa_1},\ldots,z^{\kappa_m}$, and $\kappa_1\geq\cdots\geq\kappa_m$ are integers, called partial indices.
\end{defn}

We remark that the factorization given in Definition \ref{defn:whf} is termed a left factorization in the Wiener-Hopf factorization literature \cite[p.\ 185]{gf}. It differs from the right factorization utilized in \cite{onatski} and \cite{linsys}, where the roles of $M_\pm(z)$ are reversed. The difference is due to the fact that the present analysis works with the forward shift operator, whereas \cite{onatski} and \cite{linsys} work with the backward shift operator. A left factorization of $M$ is obtained from a right factorization of $M^\ast$.

\begin{thm}\label{thm:whf}
Let $M\in\mathcal{W}^{m\times m}$ and suppose $\det(M(z))\neq0$ for all $z\in\mathbb{T}$. Then $M$ has a Wiener-Hopf factorization and its partial indices are unique.
\end{thm}
\begin{proof}
See Theorems VIII.1.1 and VIII.2.2 of \cite{gf}.
\end{proof}

The condition in Theorem \ref{thm:whf} is minimal for existence of a Wiener-Hopf factorization and uniqueness of the $M_0$ part. The $M_\pm$ parts are not unique but their general form is well understood \citep[Theorem VIII.1.2]{gf}. The non-uniqueness of $M_\pm$ has no bearing on any of our results. Wiener-Hopf factorizations can be computed in a variety of ways (see \cite{constructive} for a recent survey).

The partial indices allow us to identify an important subset of $\mathcal{W}^{m\times m}$.

\begin{defn}
Let $\mathcal{W}^{m\times m}_\circ$ be the subset of $M\in\mathcal{W}^{m\times m}$ such that $\det(M(z))\neq0$ for all $z\in\mathbb{T}$ and $\kappa_1-\kappa_m\leq1$.
\end{defn}

When $m=1$, $0=\kappa_1-\kappa_m\leq1$ so that $\mathcal{W}^{1\times 1}_\circ$ is the set of $M\in\mathcal{W}$ such that $M(z)\neq0$ for all $z\in\mathbb{T}$. Notice that for every element of $\mathcal{W}^{m\times m}_\circ$, the partial indices are either all non-negative, all non-positive, or all zero. This implies that for every element of $\mathcal{W}^{m\times m}_\circ$,
\begin{align*}
\mathrm{sign}(\kappa_i)=\mathrm{sign}\left(\sum_{i=1}^m\kappa_i\right),\qquad i=1,\ldots,m,
\end{align*}
where $\mathrm{sign}$ is equal to 1, $-1$, or 0 according to whether the argument is positive, negative, or zero respectively. Since $\sum_{i=1}^m\kappa_i$ is the winding number of $\det(M(z))$ as $z$ traverses the unit circle counter-clockwise \cite[Theorem VIII.3.1 (c)]{gf}, we can easily determine the sign of the partial indices of elements of $\mathcal{W}^{m\times m}_\circ$ by looking at the winding number. The importance of this fact will become clear in the next section when we combine it with the following fact.

\begin{thm}\label{thm:generic}
If $\mathcal{W}^{m\times m}$ is endowed with the $\mu$--essential supremum norm then $\mathcal{W}^{m\times m}_\circ$ is open and dense in $\{M\in\mathcal{W}^{m\times m}: \det(M(z))\neq0\text{ for all }z\in\mathbb{T}\}$.
\end{thm}
\begin{proof}
See Theorems 1.20 and 1.21 of \cite{gks}.
\end{proof}

Theorem \ref{thm:generic} implies that the generic or typical element of $\mathcal{W}^{m\times m}$ that admits a Wiener-Hopf factorization is an element of $\mathcal{W}_\circ^{m\times m}$. Said differently, the elements of $\mathcal{W}^{m\times m}\backslash\mathcal{W}_\circ^{m\times m}$ are non-generic or exceptional in the space of Wiener-Hopf factorizable elements of $\mathcal{W}^{m\times m}$.

To see the role played by Wiener-Hopf factorization, consider system \eqref{eq:lremsimplefreq} again. In the first case, $M$ factorized as
\begin{align*}
M_+(z)&=1-\delta^{-1}z,& M_0(z)&=1& M_-(z)&=-a\delta(1-\gamma z^{-1}).\\
\intertext{In the second case, $M$ factorized as}
M_+(z)&=1,& M_0(z)&=z& M_-(z)&=a(1-\delta z^{-1})(1-\gamma z^{-1}).\\
\intertext{In the third case, $M$ factorized as}
M_+(z)&=a\delta\gamma (1-\delta^{-1}z)(1-\gamma^{-1}z),& M_0(z)&=z^{-1}& M_-(z)&=1.
\end{align*}
In the fourth case, $M(w)=0$ for some $w\in\mathbb{T}$ and there does not exist a solution in general. Notice that the case of existence and uniqueness is associated with a partial index of zero, the case of existence and non-uniqueness is associated with a positive partial index, and the cases of non-existence in general are associated with a negative partial index and/or a zero of $M(z)$. These associations are not accidental as we will see in the next section.

\section{Existence and Uniqueness}\label{sec:eu}
Given the Wiener-Hopf factorization techniques, we can now address the general LREM problem in the frequency domain. Our first task is to define existence and uniqueness of solutions to the LREM problem in the time domain.

\begin{defn}\label{defn:formal}
An LREM is a pair $(M,N)\in\mathcal{W}^{m\times m}\times\mathcal{W}^{m\times n}$, expressed formally as
\begin{align}
\sum_{s=-\infty}^\infty M_sE_tX_{t+s}=\sum_{s=-\infty}^\infty N_sE_t\xi_{t+s},\qquad t\in\mathbb{Z}\label{eq:lrem}
\end{align}
where $\xi$ is exogenous and $X$ is endogenous.
\end{defn}

Expression \eqref{eq:lrem} is understood as a formal set of structural equations whose mathematical meaning we will postpone to Definition \ref{defn:lrem} below. The structural equations relate current, expected, and lagged values of the endogenous process $X$ to current, expected, and lagged values of an exogenous process $\xi$. In particular, $E_tX_{t+s}$ is understood to mean the expectation of $X_{t+s}$ at time $t$ if $s>0$ and $X_{t-|s|}$ if $s\leq0$; $E_t\xi_{t+s}$ is understood similarly. The notion of expectation pertinent to us will be made clear in Definition \ref{defn:lrem}. A textbook example is the New Keynesian LREM,
\begin{align}
\begin{aligned}
X_{1t}&=\theta_1 E_tX_{1,t+1}+\theta_2 X_{2t}\\
X_{2t}&=E_tX_{2,t+1}+\theta_3 X_{3t}\\
X_{3t}&=\theta_4 X_{1,t-1}+\theta_5 X_{2,t-1}+\theta_6\xi_{1t},\label{eq:nk}
\end{aligned}
\end{align}
where $X_1$, $X_2$, and $X_3$ consists of inflation, output, and the interest rate respectively; $\xi$ is interpreted as a policy shock; and $\theta=(\theta_1,\ldots,\theta_6)$ is a parameter. The first equation of \eqref{eq:nk} is the New Keynesian Phillips curve, relating current inflation to one-period-ahead expected inflation and current output, and results from the inter-temporal optimization of firms; the second equation relates current output to one-period-ahead expected output and the interest rate and results from the inter-temporal optimization of households; finally, the third equation is the policy rule by which the monetary authority sets the interest rate (see e.g. \cite{galibook}). LREM \eqref{eq:nk} is represented by the pair,
\begin{align*}
M(z,\theta)=\left[\begin{array}{ccc}
1-\theta_1z & -\theta_2 & 0 \\
0 & 1-z & -\theta_3\\
-\theta_4z^{-1} & -\theta_5z^{-1} & 1
\end{array}\right],\qquad
N(z,\theta)=\left[\begin{array}{c}
0 \\ 0 \\ \theta_6
\end{array}\right].
\end{align*}
Occasionally, researchers also consider solutions driven by sunspots, economic forces that do not appear explicitly in the system (their associated columns of $N_s$ are equal to zero) but are nevertheless considered to influence the behavior of the solution; we discuss these solutions further in Section \ref{sec:regular}.

The class of models considered in Definition \ref{defn:formal} includes the class of structural equation models ($M$ and $N$ are constant), the class of VARMA models ($M$ and $N$ are matrix polynomials in $z^{-1}$), and all of the LREMs considered in \cite{canova}, \cite{dd}, and \cite{hs}. For example, the model of \cite{sw}, studied in Section 6.2 of \cite{hs}, consists of $m=14$ endogenous variables and $n=7$ exogenous shocks.

Now in order to endow \eqref{eq:lrem} with mathematical meaning, we must first note that LREMs take as inputs not only realizations of exogenous inputs $\xi$ but also a spectral measure $F$.

\begin{center}
\begin{tikzpicture}
\node (i) at (-2.5,0) {$(\,\ldots,\xi_{t-1}(\omega),\xi_t(\omega),\xi_{t+1}(\omega),\ldots,F)$};
\node (t) at (2,0) [rectangle,draw,thick,minimum size=10mm] {$(M,N)$};
\node (o) at (6.5,0) {$(\,\ldots,X_{t-1}(\omega),X_t(\omega),X_{t+1}(\omega),\ldots\,)$};
\draw[->,thick] (i) -- (t);
\draw[->,thick] (t) -- (o);
\end{tikzpicture}
\end{center}
That is because the output $X$ of an LREM solves a system of equations in past, present, as well as expected values of $X$. In the time domain perspective on LREMs, the role of $F$ is played by a filtration with respect to which conditional expectations can be computed \cite[Definition 4.2]{linsys}. Here, conditional expectations are substituted by linear projections, the more convenient analogue to expectations in the frequency domain and, in order to compute projections, $F$ needs to be specified as well. We will see that the transfer function of solutions requires the triple $(M,N,F)$, while the realizations of the output require, additionally, the realizations of the inputs. In following this system-theoretic approach to LREMs, therefore, we will often use phrases such as ``for every spectral measure $F$'' or ``for every covariance stationary process $\xi$''.

\begin{defn}\label{defn:lrem}
Let $\xi$ be a zero-mean, $n$-dimensional, covariance stationary process with spectral measure $F$ and let $(M,N)$ be an LREM. A solution to $(M,N)$ (or equation \eqref{eq:lrem}) is an $m$-dimensional covariance stationary process $X$, causal in $\xi$, and satisfying
\begin{align}
\sum_{s=-\infty}^\infty M_s\boldsymbol{P}(X_{t+s}|\mathscr{H}_t^m)=\sum_{s=-\infty}^\infty N_s\boldsymbol{P}(\xi_{t+s}|\mathscr{H}_t^n),\qquad t\in\mathbb{Z},\label{eq:lremsol}
\end{align}
where the series converge in $\mathscr{H}^m$. We say that $(M,N)$ has no solution in general if it is possible to find a $\xi$ such that no solution to $(M,N)$ exists. A solution $X$ is unique if whenever $Y$ is also a solution, then $X_t=Y_t$ for all $t\in\mathbb{Z}$.
\end{defn}

By Definition \ref{defn:lrem}, if a unique solution (or at least a unique selection from among all possible solutions) exists for every covariance stationary $\xi$, then the LREM is a mathematical system that transforms arbitrary covariance stationary inputs into covariance stationary outputs \citep{kalmanetal,kailath,hd,sontag,caines}. If no solution exists for a given $\xi$, then the LREM is said to have no solution in general (it will often be convenient to consider purely deterministic $\xi$ to show that no solution exists in general). The restriction to covariance stationarity involves no loss of generality as $X$ and $\xi$ describe deviations away from steady state in modern LREMs (see e.g. \cite{galibook}; non-stationary LREMs are studied in \cite{linsys}).

Definition \ref{defn:lrem} also restricts attention to causal solutions. This is because the purpose of an LREM is to explain the behavior of economic variables in terms of past and present economic shocks and to obtain impulse responses. It makes little sense to consider models where current inflation is determined by a future shock to technology, for example. Note, however, that we do not impose invertibility of $\xi$ in terms of $X$ as it is not necessary for our purposes, although it does become necessary for estimation purposes (see \cite{ident}).

As in the previous section, it suffices to solve the $t=0$ equation of \eqref{eq:lremsol},
\begin{align*}
\sum_{s=-\infty}^\infty M_s\boldsymbol{P}(X_s|\mathscr{H}_0^m)=\sum_{s=-\infty}^\infty N_s\boldsymbol{P}(\xi_s|\mathscr{H}_0^n).
\end{align*}
For if $X$ is a covariance stationary process causal in $\xi$ and satisfies this system, applying the forward shift operator $t$ times to each equation we obtain \eqref{eq:lremsol}. Of course, the forward shift operator commutes with the summation because the sum converges in $\mathscr{H}^m$. Let $X_0$ have the spectral characteristics $\phi\in H_0^m$. Then the frequency domain equivalent is
\begin{align}
\sum_{s=-\infty}^\infty M_s\boldsymbol{P}(z^s\phi|H_0^m)=\sum_{s=-\infty}^\infty N_s\boldsymbol{P}(z^sI_n|H_0^n).\label{eq:lremfreq}
\end{align}
Classical frequency domain theory is built on the backwards shift operator $\phi\mapsto z^{-1}\phi$. In order to solve for $\phi$ in \eqref{eq:lremfreq}, we will need an additional, closely related, operator.

\begin{defn}
Define $\boldsymbol{V},\boldsymbol{V}^{(-1)}:H_0\rightarrow H_0$ to be the operators
\begin{align*}
\boldsymbol{V}:\phi\mapsto\boldsymbol{P}(z\phi|H_0),& & \boldsymbol{V}^{(-1)}:\phi\mapsto z^{-1}\phi.
\end{align*}
If $\kappa\geq1$ (resp.\ $\kappa\leq-1$), we denote by $\boldsymbol{V}^\kappa$ the operator $\boldsymbol{V}$ (resp.\ $\boldsymbol{V}^{(-1)}$) composed with itself $|\kappa|$ times and we define $\boldsymbol{V}^0=\boldsymbol{I}$, the identity mapping.
\end{defn}
The following lemma (proof omitted) lists the most important properties of $\boldsymbol{V}$ and $\boldsymbol{V}^{(-1)}$.

\begin{lem}\label{lem:V}
The operators $\boldsymbol{V}$ and $\boldsymbol{V}^{(-1)}$ have the following properties:
\begin{enumerate}
\item $\boldsymbol{V}^\ast=\boldsymbol{V}^{(-1)}$.
\item $(\boldsymbol{V}\phi,\boldsymbol{V}\phi)\leq(\phi,\phi)$ for all $\phi\in H_0$.
\item $\left(\boldsymbol{V}^{(-1)}\phi,\boldsymbol{V}^{(-1)}\varphi\right)=(\phi,\varphi)$ for all $\phi,\varphi\in H_0$.
\item $\boldsymbol{V}^{(-1)}$ is a right inverse of $\boldsymbol{V}$.
\item $\ker(\boldsymbol{V})=H_0\ominus z^{-1}H_0$.
\item $\dim\ker(\boldsymbol{V})\leq n$.
\item For all $\kappa\in\mathbb{Z}$ and $\phi\in H_0$, $\boldsymbol{V}^\kappa\phi=\boldsymbol{P}(z^\kappa\phi|H_0)$.
\item For $\kappa>1$, $\ker(\boldsymbol{V}^\kappa)=\ker(\boldsymbol{V})\oplus \boldsymbol{V}^{(-1)}\ker(\boldsymbol{V})\oplus\cdots\oplus \boldsymbol{V}^{1-\kappa}\ker(\boldsymbol{V})$.
\end{enumerate}
\end{lem}

Lemma \ref{lem:V} (i) states that the backwards shift operator $\boldsymbol{V}^{(-1)}$ is the adjoint of $\boldsymbol{V}$. Lemma \ref{lem:V} (ii) implies that the operator norm of $\boldsymbol{V}$ is bounded above by 1. Lemma \ref{lem:V} (iii) implies that $\boldsymbol{V}^{(-1)}$ is an isometry \cite[Theorem X.3.1]{basic}. This implies that the operator norm of $\boldsymbol{V}^{(-1)}$ is equal to 1. Lemma \ref{lem:V} (iv) establishes that $\boldsymbol{V}$ is right-invertible. Lemma \ref{lem:V} (v) clarifies the obstruction to left-invertibility of $\boldsymbol{V}$ as $\ker(\boldsymbol{V})$ may be non-trivial. For example, when $F=\mu$ and $\phi\in\mathbb{C}$, then $\boldsymbol{V}^{(-1)}\boldsymbol{V}(\phi)=\boldsymbol{P}(\,\phi\,|z^{-1},z^{-2},\ldots\,)=0$. Since $H_0\ominus z^{-1}H_0$ is the set of spectral characteristics associated with innovations to $\xi$, $\ker(\boldsymbol{V})=\{0\}$ if and only if $\xi$ is purely deterministic. Lemma \ref{lem:V} (vi) expresses the intuitive fact that the dimension of the innovation space of a covariance stationary process is bounded above by the dimension of the process. Lemma \ref{lem:V} (vii) is a convenient expression for iterates of the $\boldsymbol{V}$ operator. Finally, Lemma \ref{lem:V} (viii) decomposes the kernel of $\boldsymbol{V}^\kappa$ into a direct sum generated by the kernel of $\boldsymbol{V}$. It follows, since $\boldsymbol{V}^{(-1)}$ is an isometry, that $\dim\ker(\boldsymbol{V}^\kappa)=\kappa\dim\ker(\boldsymbol{V})$. The time-domain analogue of the decomposition in Lemma \ref{lem:V} (viii) is the familiar one from \cite{brozeetal85,brozeetal95}, where a process $\nu$ causal in $\xi$ satisfies
\begin{align*}
\boldsymbol{P}(\nu_{t+\kappa}|\mathscr{H}_t^n)=0,\qquad t\in\mathbb{Z},
\end{align*}
if and only if
\begin{align*}
\nu_{t+\kappa}=\sum_{s=1}^\kappa \boldsymbol{P}(\nu_{t+\kappa}|\mathscr{H}_{t+s}^n)-\boldsymbol{P}(\nu_{t+\kappa}|\mathscr{H}_{t+s-1}^n),\qquad t\in\mathbb{Z}.
\end{align*}
That is, if and only if $\nu_{t+\kappa}$ is representable as the sum of the prediction revisions between $t+1$ and $t+\kappa$ for all $t\in\mathbb{Z}$.

The fact that the operators $\boldsymbol{V}^s$ are uniformly bounded in the operator norm by 1 (Lemma \ref{lem:V} (i) and (ii)) ensures that $\sum_{s=-\infty}^\infty M_s\boldsymbol{V}^s$ is a bounded linear operator on $H_0$ whenever $\sum_{s=-\infty}^\infty M_sz^s\in\mathcal{W}$ \cite[Theorem I.3.2]{classes1}. More generally, we have the following definition, adopted from \cite{gf}.

\begin{defn}\label{defn:MV}
For $M\in\mathcal{W}^{m\times n}$ with $ij$-th element $M_{ij}$, define $\boldsymbol{M}_{ij}:H_0\rightarrow H_0$ as
\begin{align*}
\boldsymbol{M}_{ij}&=\sum_{s=-\infty}^\infty M_{sij}\boldsymbol{V}^s,\qquad i=1,\ldots, m,\qquad j=1,\ldots, n,
\end{align*}
where the series converges in the operator norm, and $\boldsymbol{M}:H_0^n\rightarrow H_0^m$ as
\begin{align*}
\boldsymbol{M}\phi&=\left[\begin{array}{ccc}
\boldsymbol{M}_{11}& \cdots & \boldsymbol{M}_{1n}\\
\vdots & & \vdots\\
\boldsymbol{M}_{m1}& \cdots & \boldsymbol{M}_{mn}
\end{array}\right]\left[\begin{array}{c}
\phi_1\\
\vdots\\
\phi_n
\end{array}\right]=\left[\begin{array}{c}
\sum_{j=1}^n \boldsymbol{P}(M_{1j}\phi_j|H_0)\\
\vdots\\
\sum_{j=1}^n \boldsymbol{P}(M_{mj}\phi_j|H_0)
\end{array}\right]=\boldsymbol{P}(M\phi|H_0^m).
\end{align*}
\end{defn}

For $M\in\mathcal{W}^{m\times n}$, $\|\boldsymbol{M}\phi\|_{H^m}=\|\boldsymbol{P}(M\phi|H_0^m)\|_{H^m}\leq\|M\phi\|_{H^m}\leq\|M\|_\infty\|\phi\|_{H^n}$, thus the operator norm of $\boldsymbol{M}$ is bounded by $\|M\|_\infty$. Note that by Lemma \ref{lem:V} (i), the adjoint of $\boldsymbol{M}$ is 
\begin{align*}
\boldsymbol{M}^\ast:\varphi\mapsto\boldsymbol{P}(M^\ast\varphi|H_0^n),
\end{align*}
where $M^\ast\in\mathcal{W}^{n\times m}$ is given by $(M(z))^\ast$ for $z\in\mathbb{T}$.

Definition \ref{defn:MV} allows us to express \eqref{eq:lremfreq} more compactly as,
\begin{align}
\boldsymbol{M}\phi=\boldsymbol{N}I_n.\label{eq:lremfreqcompact}
\end{align}
Recall that the spectral characteristic of $\xi$ is $I_n$. Thus, we have arrived at a linear equation in the Hilbert space $H_0^m$. Equations \eqref{eq:ar1freq}, \eqref{eq:0func}, and \eqref{eq:lremsimplefreq} are clearly special cases of \eqref{eq:lremfreqcompact}. It is also instructive to consider the special case where $(M,N)$ is a VARMA model, so $\boldsymbol{M}\phi=\boldsymbol{P}(M\phi|H_0^m)=M\phi$ and $\boldsymbol{N}I_n=\boldsymbol{P}(N|H_0^m)=N$, and the system reduces to $M\phi=N$, a problem that is very well understood \citep[Sections 1.1-1.2]{hd}.

A more delicate analysis is required for the general case. Luckily Section \ref{sec:examples} hints towards a solution. First, we obtain a Wiener-Hopf factorization,
\begin{align*}
M=M_+M_0M_-.
\end{align*}
By Theorem \ref{thm:whf}, this factorization exists if $\det(M(z))\neq0$ for all $z\in\mathbb{T}$. Then $\boldsymbol{M}_+$, $\boldsymbol{M}_0$, and $\boldsymbol{M}_-$ can be defined as in Definition \ref{defn:MV} and it is easily checked that
\begin{align*}
\boldsymbol{M}=\boldsymbol{M}_+\boldsymbol{M}_0\boldsymbol{M}_-,
\end{align*}
a fact that at first seems trivial until one recalls that $\boldsymbol{V}^i\boldsymbol{V}^j$ is not generally equal to $\boldsymbol{V}^{i+j}$ when $i<0<j$ (Lemma \ref{lem:V} (iv) and (v)). Then \eqref{eq:lremfreqcompact} can be broken up into the system of equations,
\begin{align}
\boldsymbol{M}_+\psi=\boldsymbol{N}I_n, && \boldsymbol{M}_0\varphi=\psi,&& \boldsymbol{M}_-\phi=\varphi.\label{eq:factorization}
\end{align}
Solving each of these systems is completely understood in the Wiener-Hopf factorization literature \citep{gf,cg}. Indeed, the first system of equations involves only the $\boldsymbol{V}$ operator and is solved as in the Cagan model, while the third system of equations involves only the $\boldsymbol{V}^{(-1)}$ operator and is solved as in the autoregressive model. To see this, note that since $M_+^{-1},M_-^{-1}\in\mathcal{W}^{m\times m}$, the operators
\begin{align*}
\boldsymbol{M}_+^{-1}:\phi\mapsto\boldsymbol{P}(M_+^{-1}\phi|H_0^m)& &\boldsymbol{M}_-^{-1}:\phi\mapsto\boldsymbol{P}(M_-^{-1}\phi|H_0^m)
\end{align*}
are bounded and linear on $H_0^m$. By Lemma \ref{lem:V} (vii), since $M_+, M_+^{-1}\in\mathcal{W}_+^{m\times m}$, we have that for all $\phi\in H_0^m$,
\begin{align*}
\boldsymbol{M}_+\boldsymbol{M}_+^{-1}\phi&=\boldsymbol{P}(M_+\boldsymbol{P}(M_+^{-1}\phi|H_0^m)|H_0^m)=\boldsymbol{P}(M_+M_+^{-1}\phi|H_0^m)=\boldsymbol{P}(\phi|H_0^m)=\phi\\
\boldsymbol{M}_+^{-1}\boldsymbol{M}_+\phi&=\boldsymbol{P}(M_+^{-1}\boldsymbol{P}(M_+\phi|H_0^m)|H_0^m)=\boldsymbol{P}(M_+^{-1}M_+\phi|H_0^m)=\boldsymbol{P}(\phi|H_0^m)=\phi.
\end{align*}
On the other hand, since $M_-,M_-^{-1}\in\mathcal{W}_-^{m\times m}$, $\boldsymbol{M}_-:\phi\mapsto M_-\phi$ and $\boldsymbol{M}_-^{-1}:\phi\mapsto M_-^{-1}\phi$ for every $\phi\in H_0^m$ (i.e.\ they are multiplication operators). It follows that
\begin{align*}
\boldsymbol{M}_-\boldsymbol{M}_-^{-1}\phi&=M_-M_-^{-1}\phi=\phi\\
\boldsymbol{M}_-^{-1}\boldsymbol{M}_-\phi&=M_-^{-1}M_-\phi=\phi.
\end{align*}
We have established that
\begin{align*}
\boldsymbol{M}_+ \boldsymbol{M}_+^{-1}=\boldsymbol{M}_+^{-1}\boldsymbol{M}_+=\boldsymbol{M}_- \boldsymbol{M}_-^{-1}=\boldsymbol{M}_-^{-1}\boldsymbol{M}_-=\boldsymbol{I}
\end{align*}
where $\boldsymbol{I}$ is the identity mapping on $H_0^m$. The first and third equations of \eqref{eq:factorization} are therefore uniquely solvable as
\begin{align*}
\psi=\boldsymbol{M}_+^{-1}\boldsymbol{N}I_n,&& \phi=\boldsymbol{M}_-^{-1}\varphi.
\end{align*}
It remains to solve for $\varphi$ in \eqref{eq:factorization}. Since $\boldsymbol{M}_0$ is a diagonal operator matrix with $i$-th diagonal entry equal to $\boldsymbol{V}^{\kappa_i}$, existence and uniqueness of solutions to \eqref{eq:lremfreqcompact} depend on the partial indices of $M$. Notice that if all of the partial indices of $M$ are equal to zero, then $\boldsymbol{M}_0=\boldsymbol{I}$, and \eqref{eq:lremfreqcompact} has the unique solution $\phi=\boldsymbol{M}_-^{-1}\boldsymbol{M}_+^{-1}\boldsymbol{N}I_n$. If all the partial indices of $M$ are non-negative, the operators $\boldsymbol{V}^{\kappa_i}$ above are right invertible, which implies that $\boldsymbol{M}_0$ is right invertible; in that case, a right inverse of $\boldsymbol{M}_0$ is given as
\begin{align}
\boldsymbol{M}_0^{(-1)}=\left[\begin{array}{ccccc}
\boldsymbol{V}^{(-\kappa_1)} & 0 & \cdots & 0\\
0 & \boldsymbol{V}^{(-\kappa_2)} & \cdots & 0\\
\vdots & \vdots & \ddots & \vdots \\
0 & 0 & \cdots & \boldsymbol{V}^{(-\kappa_m)}
\end{array}\right].\label{eq:m0inv}
\end{align}
If $\kappa_1>0$ and $\kappa_m\geq0$, there are generally infinitely many other right inverses of $\boldsymbol{M}_0$ and therefore infinitely many right inverses of $\boldsymbol{M}$; indeed, every right inverse of $\boldsymbol{M}$ is of the form $\boldsymbol{M}_-^{-1}\boldsymbol{M}_0^{(-1)}\boldsymbol{M}_+^{-1}$ for some right inverse $\boldsymbol{M}_0^{(-1)}$ of $\boldsymbol{M}_0$. 

The theoretical foundations for existence and uniqueness are now complete and all that remains is to apply well-known cookie-cutter results from functional analysis along with the simple techniques we employed in Section \ref{sec:examples}.

\begin{lem}\label{lem:existencephi}
If $(M,N)$ is an LREM, $\det(M(z))\neq0$ for all $z\in\mathbb{T}$, and $M$ has a Wiener-Hopf factorization $M_+M_0M_-$, then for every spectral measure $F$, $\boldsymbol{M}$ is one-to-one if the partial indices of $M$ are non-positive and $\boldsymbol{M}$ is onto if the partial indices of $M$ are non-negative. In the latter case, the general form of solutions to \eqref{eq:lremfreqcompact} is
\begin{align}
\phi=\boldsymbol{M}^{(-1)}\boldsymbol{N}I_n+\boldsymbol{M}_-^{-1}\psi,\label{eq:solphi}
\end{align}
where
\begin{align}
\boldsymbol{M}^{(-1)}=\boldsymbol{M}_-^{-1}\boldsymbol{M}_0^{(-1)}\boldsymbol{M}_+^{-1},\label{eq:Minv}
\end{align}
$\boldsymbol{M}_0^{(-1)}$ is a right inverse of $\boldsymbol{M}_0$, 
$\psi\in \ker(\boldsymbol{M}_0)$, and $\dim(\ker(\boldsymbol{M}))=\dim(\ker(\boldsymbol{V}))\sum_{i=1}^m\kappa_i$.
\end{lem}
\begin{proof} 
Proposition $2^\circ$ of Section VIII.4 of \cite{gf} applied to $M^\ast$ together with Lemma \ref{lem:V} imply that
\begin{align*}
\dim(\mathrm{coker}(\boldsymbol{M}))=-\dim(\ker(\boldsymbol{V}))\sum_{\kappa_i<0}\kappa_i,& &\dim(\ker(\boldsymbol{M}))=\dim(\ker(\boldsymbol{V}))\sum_{\kappa_i>0}\kappa_i.
\end{align*}
Thus, $\boldsymbol{M}$ is one-to-one if $\kappa_1\leq 0$ and onto if $\kappa_m\geq0$. In the latter case, \eqref{eq:solphi} is clearly a solution to \eqref{eq:lremfreqcompact}. On the other hand, if $\phi$ is a solution to \eqref{eq:lremfreqcompact}, define
\begin{align*}
\psi=\boldsymbol{M}_-\phi-\boldsymbol{M}_0^{(-1)}\boldsymbol{M}_+^{-1}\boldsymbol{N}I_n.
\end{align*}
Clearly $\psi\in H_0^m$. Finally,
\begin{equation*}
\boldsymbol{M}_0\psi=\boldsymbol{M}_0\boldsymbol{M}_-\phi-\boldsymbol{M}_+^{-1}\boldsymbol{N}I_n=\boldsymbol{M}_+^{-1}(\boldsymbol{M}\phi-\boldsymbol{N}I_n)=0.\qedhere
\end{equation*}
\end{proof}

Basic linear algebra implies that there exists a solution to \eqref{eq:lremfreqcompact} if and only if $\boldsymbol{N}I_n$ is in the image of $\boldsymbol{M}$ and a solution is unique if and only if $\ker(\boldsymbol{M})=\{0\}$. Thus, Lemma \ref{lem:existencephi} provides a sufficient condition for existence and necessary and sufficient conditions for uniqueness of solutions to \eqref{eq:lremfreqcompact} irrespective of the exogenous inputs. Recall that we insist on conditions invariant with respect to $F$ in keeping with our system-theoretic approach to LREMs. Of course, restricting attention to a particular $F$, we can say slightly more. For example, if the partial indices are non-negative and $F$ has the property that $\ker(\boldsymbol{V})=\{0\}$ then by Lemma \ref{lem:V} (viii), $\ker(\boldsymbol{M}_0)=\{0\}$ and there is a unique solution to \eqref{eq:lremfreqcompact}. That is to say, there can be no multiplicity of solutions for a perfectly predictable $\xi$. This point is made in a different context in \cite{linsys}, p.\ 641.

In the special case of a VARMA model, Lemma \ref{lem:existencephi} reduces to the classical result that if $\det(M(z))\neq0$ for all $|z|\geq1$, there exists a unique causal covariance stationary solution for every spectral measure $F$ \citep[Sections 1.1-1.2]{hd}. This is due to the fact that in that case $M_+=M_0=I_m$ and $M_-=M$ is a Wiener-Hopf factorization of $M$. Note that the problem of uniqueness does not arise in the case of VARMA.

Consider the case of existence and non-uniqueness ($\kappa_1>0$ and $\kappa_m\geq0$) more closely. Lemma \ref{lem:V} (v) and (viii) imply that $\ker(\boldsymbol{M}_0)$ is made up of spectral characteristics corresponding to arbitrary innovation processes (see Appendix \ref{sec:setsol} for a detailed analysis). Thus, Lemma \ref{lem:existencephi} states that the set of all solutions to \eqref{eq:lremfreqcompact} is generated by spectral characteristics corresponding to arbitrary innovation processes. The dimension of $\ker(\boldsymbol{M})$ is the dimension of the innovation space of $\xi$ multiplied by the winding number of $\det(M)$. This number is obtained in \cite{funovits,funovits2020} based on the \cite{sims} framework (see also \cite{sorge}).

Lemma \ref{lem:existencephi} begs the question of what happens if zeros are present on the unit circle. This is addressed in the next result.

\begin{lem}\label{lem:nonexistencephi}
If $(M,N)$ is an LREM, $\det(M(w))=0$ for some $w\in\mathbb{T}$, and $\mathrm{rank}[\;M(w)\quad N(w)\;]=m$, there exists a spectral measure $F$ such that \eqref{eq:lremfreqcompact} has no solution in $H_0^m$.
\end{lem}
\begin{proof}
Let $0\neq x\in\mathbb{C}^{1\times m}$ satisfy $x M(w)=0$ and choose $F=\delta_w I_n$. If a solution $\phi\in H_0^m$ to \eqref{eq:lremfreqcompact} exists, it must satisfy $\|\phi\|^2_{H^m}=\sum_{j=1}^m\int \phi_jdF\phi_j^\ast=\|\phi(w)\|^2_{\mathbb{C}^{m\times n}}<\infty$. Since $x M(z)\phi(z)=0$ for $F$--a.e.\ $z\in\mathbb{T}$, it follows that $x \boldsymbol{M}\phi=\boldsymbol{P}(x M\phi|H_0)=0$. This implies that $0=x\boldsymbol{N}I_n=\boldsymbol{P}(x N|H_0)$. Since $x N(z)=x N(w)$ for $F$--a.e.\ $z\in\mathbb{T}$ and $x N(w)\in H_0$, $x N(w)=\boldsymbol{P}(x N|H_0)=0$. This implies that $x [\;M(w)\quad N(w)\;]=0$, a contradiction.
\end{proof}

The basic idea behind Lemma \ref{lem:nonexistencephi} is that when $\det(M(w))=0$ for some $w\in\mathbb{T}$, the system has a form of instability akin to that of a resonance frequency in a mechanical system. When such a system is subjected to an input oscillating at frequency $\arg(w)$, its output cannot be covariance stationary. In particular, a mechanical system will oscillate with increasing amplitude until failure \citep[p.\ 183]{arnold}. In the parlance of system theory, the system is said to have an unstable mode \citep[Chapter 5]{sontag}.

The rank condition on $[\;M(w)\quad N(w)\;]$ in Lemma \ref{lem:nonexistencephi} permits inputs at frequency $\arg(w)$ to excite the instability in the system. In the VARMA literature it is typically assumed that $\mathrm{rank}[\;M(z)\quad N(z)\;]=m$ for all $z\in\mathbb{C}$ \citep[Chapter 2]{hd}. In the systems theory literature, similar conditions characterize controllability of the output in terms of the input \citep[Chapter 6]{kailath}. Without a condition of this sort, there may be no input that can excite the system's instability. For example, the LREM $(M,N)=(1-z^{-1},1-z^{-1})$ has a solution $\phi=1$ for any $F$. Note that this condition permits oscillatory inputs to excite instability but not necessarily white noise inputs. For example, the LREM $(M,N)=(-z+3-2z^{-1},3-2z^{-1})$ has $M(1)=0$ and $N(1)=1$ so that the conditions of Lemma \ref{lem:nonexistencephi} are satisfied but the instability of the system cannot be excited by a white noise input because, as is easily checked, $\phi=1$ is a solution to \eqref{eq:lremfreqcompact} when $F=\mu$. It is possible to formulate a different condition on $N$ that will permit white noise inputs to excite instability in the system and, indeed, much more can be said about stability in LREMs. However, a general analysis of stability of LREMs is outside the scope of this paper and is left for future research.

We are now ready to state the following result.

\begin{thm}[Onatski's First Theorem]\label{thm:onatski1}
If $(M,N)$ is an LREM, $\det(M(z))\neq0$ for all $z\in\mathbb{T}$, and $M$ has a Wiener-Hopf factorization $M_+M_0M_-$, then
\begin{enumerate}
\item If the partial indices of $M$ are non-negative, then for every covariance stationary process $\xi$ there exists a solution $X$ to $(M,N)$. The general form of solutions to $(M,N)$ is
\begin{align*}
X_t=\int z^t \boldsymbol{M}^{(-1)}\boldsymbol{N}I_nd\Phi+\int z^t \boldsymbol{M}_-^{-1}\psi d\Phi,\quad t\in\mathbb{Z},
\end{align*}
where $\boldsymbol{M}^{(-1)}$ is a right inverse of $\boldsymbol{M}$, $\psi\in \ker(\boldsymbol{M}_0)$, and the dimension of the solution space is the dimension of the innovation space of $\xi$ times the winding number of $\det(M)$.
\item If the partial indices of $M$ are all equal to zero, there is a unique solution for every covariance stationary process $\xi$ given by
\begin{align*}
X_t=\int z^t \boldsymbol{M}_-^{-1}\boldsymbol{M}_+^{-1}\boldsymbol{N}I_nd\Phi,\quad t\in\mathbb{Z},
\end{align*}
\item If $M$ has a negative partial index and the zero-th Fourier coefficient of the last row of $[M_+^{-1}N]_-$ is non-zero, then there exists no solution to $(M,N)$ in general.
\end{enumerate}
\end{thm}
\begin{proof}
(i) and (ii) follow immediately from Lemma \ref{lem:existencephi} and the spectral representation theorem. (iii) We claim that for $F=\mu I_n$, there exists no solution to $(M,N)$. For any solution must satisfy $\boldsymbol{M}_0\boldsymbol{M}_-\phi=\boldsymbol{M}_+^{-1}\boldsymbol{N}I_n=[M_+N]_-$ as $\{z^se_j: j=1,\ldots, n,s\in\mathbb{Z}\}$ is orthonormal for our choice of $F$. But the last row of this equation is $z^{\kappa_m}e_mM_-\phi=e_m[M_+^{-1}N]_-$, which evidently cannot hold if $\kappa_m<0$, as the zero-th Fourier coefficient of the left hand side is zero, while that of the right hand side is not by assumption.
\end{proof}

The relationship between partial indices and existence and uniqueness of solutions to LREMs was first discovered by \cite{onatski}. Theorem \ref{thm:onatski1} generalizes \cite{onatski} by allowing for arbitrary covariance stationary input and non-constant $N$. The general expression of solutions in (i) is similar to the one obtained in Theorem 4.1 (ii) of \cite{linsys}; \cite{sunspots}, \cite{farmeretal}, \cite{bianchinicolo}, and \cite{regular} provide algorithms for computing these solutions.

Note that the analogue of (iii) in \cite{onatski} is incorrect: if $N$ ($\Gamma$ in Onatski's notation) is constant and non-zero,  and $M$ ($A$ in Onatski's notation) has a negative (positive in Onatski's framework) partial index, it does not follow that a solution fails to exist. A counterexample is $(M,N)=\left(\left[\begin{smallmatrix} 1 & 0 \\ 0 & z^{-1} \end{smallmatrix}\right],\left[\begin{smallmatrix} 1 \\ 0 \end{smallmatrix}\right]\right)$, which has the unique solution $\left[\begin{smallmatrix} 1 \\ 0 \end{smallmatrix}\right]$ for any choice of $F$. It is easy to check, however, that Onatski's statement is correct in the scalar case $m=1$. Our additional condition on $[M_+^{-1}N]_-$, like the coprimeness condition in Lemma \ref{lem:nonexistencephi}, permits us to use white noise as an input to the system for which there cannot be an admissible output (different conditions can also be given). We will see shortly that this condition is generic.

The unit root case can be stated as

\begin{thm}\label{thm:unitroot}
If $(M,N)$ is an LREM, $\det(M(w))=0$ for some $w\in\mathbb{T}$, and\linebreak $\mathrm{rank}[\;M(w)\quad N(w)\;]=m$, then there exists no solution to $(M,N)$ in general.
\end{thm}
\begin{proof}
Follows from Lemma \ref{lem:nonexistencephi} and the spectral representation theorem.
\end{proof}

We can also state the following result for generic systems.

\begin{thm}[Onatski's Second Theorem]\label{thm:onatski2}
For a generic LREM $(M,N)$ with $\det(M(z))\neq0$ for all $z\in\mathbb{T}$, there exists possibly infinitely many solutions, a unique solution, or no solution in general according to whether $\det(M(z))$ winds around the origin a positive, zero, or negative number of times as $z$ traverses $\mathbb{T}$ counter-clockwise.
\end{thm}
\begin{proof}
By Theorem \ref{thm:generic}, $\mathcal{W}_\circ^{m\times m}$ is generic in the space of Wiener-Hopf factorizable elements of $\mathcal{W}^{m\times m}$. For fixed $M\in\mathcal{W}_\circ^{m\times m}$ with Wiener-Hopf factorization $M=M_+M_0M_-$, the generic element of $\mathcal{W}^{m\times n}$ has a zero-th Fourier coefficient of $e_m[M_+^{-1}N]_-$ that is non-zero (i.e.\ $\int e_m[M_+^{-1}N]_-d\mu\neq0$). Indeed, when $\mathcal{W}^{m\times n}$ is endowed with the $\mu$--essential supremum norm, $N\mapsto\int e_m[M_+^{-1}N]_-d\mu$ is a continuous mapping from $\mathcal{W}^{m\times n}$ to $\mathbb{C}^{1\times n}$. The result then follows from Theorem \ref{thm:onatski1} (iii).
\end{proof}

In closing, it is important to note that the assumptions underlying existence and uniqueness in this paper are weaker than in all of the previous literature on frequency domain solutions of LREMs, namely the work of \cite{whiteman}, \cite{onatski}, \cite{tanwalker}, \cite{tan}, and \cite{meyer}. These works require the exogenous process to have a purely non deterministic \cite{wold} representation, an assumption that is demonstrated to be unnecessary. The weaker assumptions of this paper have also facilitated discussion of zeros of $M$, an aspect of the theory absent from the previous literature. The advantage of these stronger assumptions, however, is that they do permit closed form expressions of solutions. See Appendix \ref{sec:previous} for a more detailed discussion.

\section{Ill-Posedness and Regularization}\label{sec:regular}
We have seen in the previous section that LREMs may have infinitely many solutions. In order to estimate and conduct inference, various selection mechanisms have been proposed in the macroeconometrics literature to associate a unique solution to any given set of parameters. We will review these proposals and show that, unfortunately, continuity of the selected solutions is not guaranteed. This leads to the development of a new regularized solution with guaranteed regularity properties.

\subsection{Non-Uniqueness}
The first approach to non-uniqueness, proposed by \cite{taylor}, selects the solution that minimizes the variance of the price variable in the model if one exists. Taylor motivates this solution by arguing that collective rationality of economic agents will naturally lead them to coordinate their activities to achieve this equilibrium. Unfortunately, this provides no guidance for models in which indeterminacy afflicts non-price variables.

The second approach to non-uniqueness commits to a particular algorithm for obtaining a solution and ignores all other solutions. That is, it commits to a particular choice of right inverse to $\boldsymbol{M}$, say $\boldsymbol{M}^{(-1)}$, and selects the solution $\boldsymbol{M}^{(-1)}\boldsymbol{N}I_n$, ignoring all other solutions $\boldsymbol{M}^{(-1)}\boldsymbol{N}I_n+\ker(\boldsymbol{M})$. This is the ``minimum state variable'' approach of \cite{msv}. Although this may appear to solve the problem of selecting a unique solution, there are generally infinitely many algorithms for obtaining solutions as there are infinitely many right inverses of $\boldsymbol{M}$. Thus, the concept of minimum state variable solution is not well-defined without specifying the particular algorithm or right inverse to be used for the solution. Geometrically, there is a minimum state variable solution at every single point of $\boldsymbol{M}^{(-1)}\boldsymbol{N}I_n+\ker(\boldsymbol{M})$ (see Figure \ref{fig:geom}). Without sound economic reasoning for the choice of algorithm or right inverse, it cannot be argued that one minimum state variable solution obtained by one algorithm is a more appropriate choice than another, obtained by a different algorithm.

\begin{figure}\caption{Solutions to Linear Rational Expectations Models}
\centering
\label{fig:geom}
\smallskip
\begin{tikzpicture}
\draw[->] (0,0) -- (0,6);
\draw[->] (0,0) -- (6,0);
\draw (0,5) -- (6,3);
\draw[gray] (0,0) -- (1.5,4.5);
\draw[gray] (1.8,4.4) -- (1.7,4.1) -- (1.4,4.2);
\draw (4,4.3) node {$\boldsymbol{M}^{(-1)}\boldsymbol{N}I_n$};
\draw (1.9,4.9) node {$\boldsymbol{M}^\dag\boldsymbol{N}I_n$};
\draw (5.8,3.7) node {$\boldsymbol{M}^{[-1]}\boldsymbol{N}I_n$};
\draw (6.2,2.7) node {$\boldsymbol{M}^{(-1)}\boldsymbol{N}I_n+\ker(\boldsymbol{M})$};
\draw (6.5,0) node {$H_0^m$};
\node at (3.26,3.91) [circle,inner sep=0.6mm, fill=blue] {};
\node at (1.5,4.5) [circle,inner sep=0.6mm,fill=red] {};
\node at (5,3.33) [circle,inner sep=0.6mm,fill=green] {};
\end{tikzpicture}
\end{figure}

Finally, the modern approach to non-uniqueness, as exemplified by \cite{sunspots}, \cite{farmeretal}, and \cite{bianchinicolo} parametrizes the full set of solutions. Like the minimum state variable approach, it commits to a particular algorithm or right inverse of $\boldsymbol{M}$, say $\boldsymbol{M}^{(-1)}$, and represents each solution as a sum $\boldsymbol{M}^{(-1)}\boldsymbol{N}I_n+\chi$. The first part, $\boldsymbol{M}^{(-1)}\boldsymbol{N}I_n$, is interpreted as generated by fundamental economic forces (the ones that appear explicitly in the model \eqref{eq:lrem}), while the second part, $\chi\in\ker(\boldsymbol{M})$, generated as we have seen by arbitrary innovation processes, is interpreted as non-fundamental sunspots \citep{farmer}. The coordinates on $\ker(\boldsymbol{M})$ are appended to the parameters of $(M,N)$ and both estimated by either frequentist or Bayesian methods. The geometry of the spectral approach allows us to see very clearly a serious conceptual problem with this approach. Various papers in the literature claim to present evidence of the importance of sunspots as drivers of macroeconomic activity by showing that the contribution of sunspots, as measured by the size of the estimated $\chi$, is not insignificant empirically. However, it is clear from Figure \ref{fig:indet} that the size of $\chi$ very much depends on the choice of algorithm or right inverse of $\boldsymbol{M}$: by one representation, $\boldsymbol{M}^{(-1)}\boldsymbol{N}I_n+\chi_1$, sunspots play a large role, by another representation, $\boldsymbol{M}^{[-1]}\boldsymbol{N}I_n+\chi_2$, sunspots play a small role. Without a sound economic reason for the choice of algorithm or right inverse of $\boldsymbol{M}$, it is unclear that the contribution of sunspots is being measured correctly. Thus, the modern approach suffers from a similar difficulty as the minimum state variable approach.

\begin{figure}\caption{Indeterminacy of Sunspots in Solutions to Linear Rational Expectations Models}
\centering
\label{fig:indet}
\smallskip
\begin{tikzpicture}
\draw[->] (0,0) -- (0,6);
\draw[->] (0,0) -- (6,0);
\draw (0,5) -- (6,3);
\draw (1.2,2.5) node {$\boldsymbol{M}^{(-1)}\boldsymbol{N}I_n$};
\draw (3,1.2) node {$\boldsymbol{M}^{[-1]}\boldsymbol{N}I_n$};
\draw (2.7,4.4) node {$\chi_1$};
\draw (4.9,3.7) node {$\chi_2$};
\draw (6.2,2.7) node {$\boldsymbol{M}^{(-1)}\boldsymbol{N}I_n+\ker(\boldsymbol{M})$};
\draw (6.5,0) node {$H_0^m$};
\draw[->, thick, red] (0,0) -- (0.3,4.9) -- (4.42,3.525);
\draw[->, thick, green] (0,0) -- (5,3.33) -- (4.58,3.475);
\node at (4.5,3.5) [circle,inner sep=0.6mm,fill=blue] {};
\end{tikzpicture}
\end{figure}

\subsection{Discontinuity}
We have established in Lemma \ref{lem:existencephi} that the set of all solutions to \eqref{eq:lremfreqcompact} is of the form $\boldsymbol{M}^{(-1)}\boldsymbol{N}I_n+\ker(\boldsymbol{M})$ for some right inverse $\boldsymbol{M}^{(-1)}$ of $\boldsymbol{M}$. In the course of proving Theorem \ref{thm:tykcont} below we will see that small variations of $M$ in the $\mu$--essential supremum norm lead to small variations in the orthogonal projection onto $\ker(\boldsymbol{M})$ in the operator norm. Unfortunately, small changes in $M$ in the $\mu$--essential supremum norm are not guaranteed to lead to small changes in $\boldsymbol{M}^{(-1)}\boldsymbol{N}I_n$ in the $H^m$ norm. Discontinuity can occur when $M$ falls inside the non-generic set $\mathcal{W}^{m\times m}\backslash\mathcal{W}^{m\times m}_\circ$. We will illustrate this discontinuity with the simplest possible example. It is, of course, possible to illustrate discontinuity using a more realistic example similar to \eqref{eq:nk} but doing so comes at the cost of analytic tractability, as solving even the simplest multivariate LREMs can be quite cumbersome. It is important to keep in mind that LREMs are typically parametrized purely on theoretical considerations, without any regard to statistical properties, so that elements of $\mathcal{W}^{m\times m}\backslash\mathcal{W}^{m\times m}_\circ$ are not expressly excluded from the parametrization \citep{onatski,generic}.

Consider the following example,
\begin{align}
F=\mu I_2,\qquad M(z,\theta)=\left[\begin{array}{cc}
z^2 & 0 \\ \theta z & 1
\end{array}\right],\qquad N(z)=\left[\begin{array}{cc}
1 & 0 \\ 0 & 1
\end{array}\right],\label{eq:nongeneric}
\end{align}
with $\theta\in\mathbb{R}$. Thus, $\xi$ is a standardized white noise process.

In order to compute solutions we will need to obtain a Wiener-Hopf factorization,
\begin{align*}
M(z,\theta)=\begin{cases}
\left[\begin{array}{cc}
1 & 0 \\ 0 & 1
\end{array}\right]\left[\begin{array}{cc}
z^2 & 0 \\ 0 & 1
\end{array}\right]\left[\begin{array}{cc}
1 & 0 \\ 0 & 1
\end{array}\right], & \theta=0,\vspace{0.5cm}\\ 
\left[\begin{array}{cc}
1 & z \\ 0 & \theta
\end{array}\right]\left[\begin{array}{cc}
z & 0 \\ 0 & z
\end{array}\right]\left[\begin{array}{cc}
0 & -\theta^{-1} \\ 1 & \theta^{-1}z^{-1}
\end{array}\right], & \theta\neq0.
\end{cases}
\end{align*}
Thus, there exists a (non-unique) solution for every $\theta\in\mathbb{R}$. Define $\boldsymbol{M}(\theta)$ as in Definition \ref{defn:MV}, $\phi\mapsto\boldsymbol{P}(M(\,\cdot\,,\theta)\phi|H_0^2)$. Define $\boldsymbol{M}(\theta)^{(-1)}$ as in \eqref{eq:Minv} with $\boldsymbol{M}_0(\theta)^{(-1)}$ defined as in \eqref{eq:m0inv}. We will restrict attention to minimum state variable solutions (i.e.\ $\psi=0$ in \eqref{eq:solphi}),
\begin{align*}
\phi(\theta)=\boldsymbol{M}(\theta)^{(-1)}\boldsymbol{N}I_n=\begin{cases}
\left[\begin{array}{cc}
z^{-2} & 0 \\ 0 & 1
\end{array}\right],& \theta=0,\vspace{0.5cm}\\
\left[\begin{array}{cc}
z^{-2} & \theta^{-1}z^{-1} \\ -\theta z^{-1} & 0\end{array}\right],& \theta\neq0.
\end{cases}
\end{align*}
This solution is clearly discontinuous at $\theta=0$. In particular,
\begin{align*}
\|\phi(\theta)-\phi(0)\|^2_{H^2}&=\int |\theta^{-1}z^{-1}|^2d\mu+\int|\theta z^{-1}|^2d\mu+\int|1|^2d\mu=\theta^{-2}+\theta^2+1
\end{align*}
tends to infinity as $\theta\rightarrow0$. Thus, discontinuities can arise in the process of solving an LREM.

While this discontinuity may be quite jarring to readers acquainted with the LREM literature, from the Wiener-Hopf factorization literature point of view there is nothing surprising about this discontinuity. Indeed, \eqref{eq:nongeneric} is a minor modification of the example given in Section 1.5 of \cite{gks}. The fact that only multivariate systems with non-unique solutions can exhibit this discontinuity may explain why discontinuity has not received sufficient attention in the LREM literature, as solving such models is difficult to do analytically. The heart of the problem is that, when $m>1$, there are points in $\mathcal{W}^{m\times m}$ at which partial indices are discontinuous, these are precisely the non-generic set $\mathcal{W}^{m\times m}\backslash\mathcal{W}^{m\times m}_\circ$ \cite[Corollary 1.22]{gks}. That is, if $M\in\mathcal{W}^{m\times m}$ has partial indices satisfying $\kappa_1-\kappa_m>1$, then there are small (in the $\mu$--essential supremum norm) changes in $M$ that lead to a jump in $M_0$ and since $M$ has undergone only a small change, $M_\pm$ must also jump. If we then choose a fixed right inverse $\boldsymbol{M}_0^{(-1)}$ in the solution \eqref{eq:solphi} (e.g.\ choosing $\boldsymbol{M}_0^{(-1)}$ as in \eqref{eq:m0inv}), then the discontinuity can affect the solution. This is, unfortunately, what is done in all existing solution algorithms \citep{linsys,regular}. In our example, the partial indices of $M(z,\theta)$ are $(2,0)$ for $\theta=0$ and $(1,1)$ for $\theta\neq0$; $M(z,0)\in \mathcal{W}^{2\times 2}\backslash\mathcal{W}^{2\times 2}_\circ$ and $M(z,\theta)\in \mathcal{W}^{2\times 2}_\circ$ for $\theta\neq0$. The discontinuity in the partial indices is what generates the discontinuity of the Wiener-Hopf factors, which in turn generates a discontinuity in the solution.

It bears emphasizing that this discontinuity is a feature of the mathematical problem, it is not a feature of the particular algorithm used to solve the problem; \cite{regular} shows how it arises in the \cite{sims} framework, \cite{linsys} shows how it arises in a linear systems framework, and Appendix \ref{sec:nongeneric} shows how it arises when solving the system by hand. It is also important to note that the discontinuity of Wiener-Hopf factorization implies that there can exist no numerically stable way to compute it generally (see Example C.2 of the online supplement to \cite{linsys} for an illustration). While the elements of $\mathcal{W}^{m\times m}_\circ$ can be factorized using finite precision arithmetic, the elements of $\mathcal{W}^{m\times m}\backslash\mathcal{W}^{m\times m}_\circ$ cannot be factorized without infinite precision. Thus, it is a non-starter to verify numerically whether a given system is generic or not in the process of estimation and inference. Note that an analogous problem arises in the computation of the Jordan canonical form, which is discontinuous at a non-generic set of matrices in $\mathbb{C}^{m\times m}$ when $m>1$ \cite[pp.\ 127-128]{hj1}; there, the recommendation is to compute the Schur canonical form instead; and in the next section we will propose, similarly, to compute different solutions to LREMs than the ones proposed in the literature.

\subsection{Regularization}
Having established that the LREM problem in macroeconometrics is ill-posed, we now consider how to obtain economically meaningful solutions amenable to mainstream econometric techniques.

Perhaps the most natural solution to the ill-posedness problem is to avoid it altogether and restrict attention to systems with unique solutions (i.e.\ LREMs $(M,N)$, where the partial indices of $M$ are equal to zero). \cite{ident} have shown that unique solutions are not only continuous in the parameters of an LREM but also analytic (see the proof of Theorem 6.2). A less restrictive solution is to allow for non-uniqueness but restrict attention to generic systems (i.e.\ $\mathcal{W}^{m\times m}_\circ$). However, genericity cannot be taken for granted as both \cite{onatski} and \cite{generic} have warned and some models may be parametrized to always fall inside $\mathcal{W}^{m\times m}\backslash\mathcal{W}^{m\times m}_\circ$ (interestingly, \cite{generic} is widely but erroneously considered to be a critique of \cite{onatski}, see \cite{blog}). Moreover, we need differentiability, not just continuity, in order to ensure asymptotic normality of extremum estimators \citep{pp,hd} as it facilitates the construction of confidence intervals as well as hypothesis testing \citep{nm}. Therefore, we opt for a more straightforward solution, regularization.

With the geometry of the spectral approach in view, one is led inexorably to consider the Tykhonov-regularized solution to \eqref{eq:lremfreqcompact}, which minimizes the total variance, $\|X_0\|^2_{\mathscr{H}^m}=\|\phi\|^2_{H^m}$ among all solutions to \eqref{eq:lremsol},
\begin{align*}
\phi_{\boldsymbol{I}}=\arg\min\left\{\|\phi\|^2_{H^m}: \boldsymbol{M}\phi=\boldsymbol{N}I_n\right\}.
\end{align*}
It can be shown that $\phi_{\boldsymbol{I}}=\boldsymbol{M}^\dag\boldsymbol{N}I_n$, where $\boldsymbol{M}^\dag$ is the Moore-Penrose inverse of $\boldsymbol{M}$ \cite[p.\ 41]{groetsch}. This takes a particularly simple form in our context.

\begin{lem}\label{lem:mp}
If $M\in\mathcal{W}^{m\times m}$, $\det(M(z))\neq0$ for all $z\in\mathbb{T}$, and the partial indices of $M$ are non-negative then,
\begin{align*}
\boldsymbol{M}^\dag=\boldsymbol{M}^\ast(\boldsymbol{M}\boldsymbol{M}^\ast)^{-1}.
\end{align*}
\end{lem}
\begin{proof}
By Lemma \ref{lem:existencephi}, $\boldsymbol{M}$ is onto. It follows that $\boldsymbol{M}^\dag=\boldsymbol{M}^\ast(\boldsymbol{M}\boldsymbol{M}^\ast)^\dag$ \cite[Theorem 2.1.5]{groetsch}. Since $\boldsymbol{M}$ is onto, $\boldsymbol{M}^\ast$ is one-to-one and so $(\boldsymbol{M}\boldsymbol{M}^\ast)^\dag=(\boldsymbol{M}\boldsymbol{M}^\ast)^{-1}$.
\end{proof}

Lemma \ref{lem:mp} implies a very simple technique for computing the Tykhonov-regularized solution. Whenever a solution to \eqref{eq:lremfreqcompact} exists, the Tykhonov-regularized solution is obtained from the unique solution to the auxiliary block triangular frequency domain system,
\begin{align*}
\left[\begin{array}{cc}
\boldsymbol{M}\boldsymbol{M}^\ast & 0 \\ -\boldsymbol{M}^\ast & \boldsymbol{I}
\end{array}\right]\left[\begin{array}{c}
\varphi \\ \phi_{\boldsymbol{I}}
\end{array}\right]=\left[\begin{array}{c}
\boldsymbol{N}I_n\\ 0
\end{array}\right].
\end{align*}
This implies that whenever a solution exists, the regularized solution is obtained uniquely by solving an auxiliary LREM. For example, in the mixed model from Section \ref{sec:examples}, the Tykhonov-regularized solution is the solution $X$ to the auxiliary LREM,
\begin{gather*}
a\overline{c} E_tY_{t+2}+(a\overline{b}+b\overline{c})E_tY_{t+1}+(|a|^2+|b|^2)Y_t+|c|^2E_{t-1}Y_t+(b\overline{a}+c\overline{b})Y_{t-1}+c\overline{a}Y_{t-2}=\xi_t,\\
X_t=\overline{a}Y_{t-1}+\overline{b}Y_t+\overline{c}E_tY_{t+1}.
\end{gather*}
This regularized solution exists and is unique whenever the original system satisfies the conditions for existence.

Note that when the partial indices of $M$ are all equal to zero, $\boldsymbol{M}^\dag=\boldsymbol{M}^\ast(\boldsymbol{M}\boldsymbol{M}^\ast)^{-1}=\boldsymbol{M}^\ast(\boldsymbol{M}^\ast)^{-1}(\boldsymbol{M})^{-1}=(\boldsymbol{M})^{-1}$. In other words, Tykhonov-regularization has no effect if the solution is unique.

\begin{figure}\caption{Regularized Solutions to Linear Rational Expectations Models}
\centering
\label{fig:reg}
\smallskip
\begin{tikzpicture}
\clip (-0.2,-0.2) rectangle (8,6);
\draw[->] (0,0) -- (0,6);
\draw[->] (0,0) -- (6,0);
\draw (0,5) -- (6,3);
\draw[gray] (0,0) -- (1.5,4.5);
\draw[gray] (1.8,4.4) -- (1.7,4.1) -- (1.4,4.2);
\draw[red,dashed] (0,0) circle (4.75);
\draw[red,dashed] (0,0) circle (3.17);
\draw[red,dashed] (0,0) circle (1.58);
\draw[blue,densely dashed] (0,0) circle [x radius=6.25, y radius=2, rotate=28];
\draw[blue,densely dashed] (0,0) circle [x radius=4.17, y radius=1.33, rotate=28];
\draw[blue,densely dashed] (0,0) circle [x radius=2.08, y radius=0.67, rotate=28];
\draw (1.7,4.9) node {$\phi_{\boldsymbol{I}}$};
\draw (5.4,3.7) node {$\phi_{\boldsymbol{L}}$};
\draw (6.5,0) node {$H_0^m$};
\node at (1.5,4.5) [circle,inner sep=0.6mm,fill=red] {};
\node at (5,3.33) [circle,inner sep=0.6mm,fill=blue] {};
\end{tikzpicture}
\end{figure}

More generally, we may consider
\begin{align*}
\phi_{\boldsymbol{L}}\in\arg\min\left\{\|\boldsymbol{L}\phi\|^2_{H^l}: \boldsymbol{M}\phi=\boldsymbol{N}I_n\right\},
\end{align*}
where $\boldsymbol{L}:H^m\rightarrow H^l$ is a bounded linear operator chosen by the researcher. Next, we consider economic motivation for a number of choices of $\boldsymbol{L}$.

If the researcher wishes to shrink the variance of the $i$-th component of $X$, they can set
\begin{align*}
\boldsymbol{L}:\phi\mapsto \phi_i.
\end{align*}
Note that when the $i$-th component is the price variable, we obtain the \cite{taylor} solution.

One can also shrink expected values of $X$. For example, if certain solutions obtained by the methods  reviewed above yield expectations of output that are too variable relative to what one expects empirically, then one can impose this prior by using
\begin{align*}
\boldsymbol{L}:\phi\mapsto \boldsymbol{V}\phi_j,
\end{align*}
where $j$ is the coordinate corresponding to output in $X$.

Linear combinations of lagged, current, and expected values of coordinates of $X$ can also be shrunk similarly. The operator
\begin{align*}
\boldsymbol{L}:\phi\mapsto \left[\begin{array}{c}
(\boldsymbol{V}-2\,\boldsymbol{I}+\boldsymbol{V}^{(-1)})\phi_1\\
\vdots \\
(\boldsymbol{V}-2\,\boldsymbol{I}+\boldsymbol{V}^{(-1)})\phi_m
\end{array}\right]
\end{align*}
shrinks the second difference of all variables, imposing smoothness on solutions, similar to the idea of the \cite{hp} filter.

Note that the operators above belong to the class of operators defined in Definition \ref{defn:MV}. Thus, we can more generally use any $L\in\mathcal{W}^{l\times m}$ to construct the weight
\begin{align*}
\boldsymbol{L}:\phi\mapsto \left[\begin{array}{c}
\sum_{j=1}^m \boldsymbol{L}_{1j}\phi_j\\
\vdots\\
\sum_{j=1}^m \boldsymbol{L}_{lj}\phi_j
\end{array}\right].
\end{align*}

More importantly, regularization can allow the researcher to shrink not just across time but across frequencies. For example, the researcher may wish to shrink the spectrum of the solution towards frequencies of between $2\pi/32$ and $2\pi/4$ corresponding to business cycle fluctuations of period 4-32 quarters in quarterly data, i.e.\ using
\begin{align*}
\boldsymbol{L}:\phi\mapsto L\phi,& &L(z)=\begin{cases}
0,& \pi/16\leq |\arg(z)|\leq\pi/2,\\
I_m,& \text{otherwise}.
\end{cases}
\end{align*}

Finally, we can consider solutions that minimize a finite weighted sum of individual criteria as reviewed above,
\begin{align*}
a_1\|\boldsymbol{L}_1\phi\|_{H^{l_1}}^2+\cdots+a_d\|\boldsymbol{L}_d\phi\|_{H^{l_d}}^2,
\end{align*}
where $\boldsymbol{L}_i:H^m\rightarrow H^{l_i}$ and $a_i>0$ for $i=1,\ldots,d$. This would allow the researcher to impose $d$ different criteria according to the weights $a_1,\ldots,a_d$. This can be achieved by using
\begin{align*}
\boldsymbol{L}:\phi\mapsto \left[\begin{array}{c}
\sqrt{a_1}\boldsymbol{L}_1\\
\vdots\\
\sqrt{a_d}\boldsymbol{L}_d
\end{array}\right]\phi.
\end{align*}

The argument for regularization is the same as employed throughout the inverse and ill-posed problems literatures: if theory is insufficient to pin down a unique continuous solution, other information can be employed. In our case, regularization allows economically meaningful shrinkage criteria to select the best among all possible solutions or, what amounts to the same, it allows the researcher's priors about economic behavior to select the most appropriate solution. Regularization resolves the problem of selecting an economically grounded solution and, as we will soon see, also ameliorates the discontinuity problem.

Existence of regularized solutions is guaranteed whenever the given LREM satisfies the already minimal conditions for existence of solutions. That is because, according to Lemma \ref{lem:existencephi}, the set of all solutions is finite dimensional and so the problem of minimizing $\|\boldsymbol{L}\phi\|^2_{H^l}$ subject to $\boldsymbol{M}\phi=\boldsymbol{N}I_n$ reduces to a problem of minimizing a non-negative quadratic form (see Figure \ref{fig:reg}). Uniqueness of regularized solutions, on the other hand, is the subject of the next result.

\begin{lem}\label{lem:reg}
If $(M,N)$ is an LREM, $\det(M(z))\neq0$ for all $z\in\mathbb{T}$, the partial indices of $M$ are non-negative, and $\boldsymbol{L}:H^m\rightarrow H^l$ is a bounded linear operator, then \begin{align*}
\phi_{\boldsymbol{L}}=(\boldsymbol{I}-(\boldsymbol{L}|_{\ker(\boldsymbol{M})})^\dag \boldsymbol{L})\boldsymbol{M}^\dag \boldsymbol{N}I_n,
\end{align*}
where $\boldsymbol{L}|_{\ker(\boldsymbol{M})}$ is the restriction of $\boldsymbol{L}$ to $\ker(\boldsymbol{M})$, is a regularized solution to \eqref{eq:lremfreqcompact}. $\phi_{\boldsymbol{L}}$ is the only element of $\arg\min\left\{\|\boldsymbol{L}\phi\|^2_{H^l}: \boldsymbol{M}\phi=\boldsymbol{N}I_n\right\}$ in $H_0^m$ if and only if $\ker(\boldsymbol{L})\cap\ker(\boldsymbol{M})=\{0\}$.
\end{lem}
\begin{proof}
See Appendix \ref{sec:proofs}.
\end{proof}

In the course of proving Lemma \ref{lem:reg} we find that the set of regularized solutions is $\phi_{\boldsymbol{L}}+\ker(\boldsymbol{M})\cap\ker(\boldsymbol{L})$. Thus, regularization produces a unique solution if and only if $\ker(\boldsymbol{L})\cap\ker(\boldsymbol{M})=\{0\}$. Geometrically, this condition requires the operator $\boldsymbol{L}$ to put weight on all directions of indeterminacy of the LREM; if $\ker(\boldsymbol{L})\cap\ker(\boldsymbol{M})\neq\{0\}$, it will be possible to perturb $\phi_{\boldsymbol{L}}$ in any direction in $\ker(\boldsymbol{L})\cap\ker(\boldsymbol{M})$ to arrive at another regularized solution and uniqueness will fail. A host of other representations of $\phi_{\boldsymbol{L}}$ are obtained in Appendix \ref{sec:proofs}.

Some special cases of Lemma \ref{lem:reg} are particularly instructive. When $\ker(\boldsymbol{L})=\{0\}$, the regularized solution is unique regardless of $\boldsymbol{M}$. That is, when $\boldsymbol{L}$ puts weight on all directions in the solution space, the regularized solution is unique. An example of this is $\boldsymbol{L}=\boldsymbol{I}$, which produces the Tykhonov-regularized solution $\phi_{\boldsymbol{I}}=\boldsymbol{M}^\dag\boldsymbol{N}I_n$. When $\ker(\boldsymbol{M})=\{0\}$ (i.e.\ the solution to the LREM is unique), the regularized solution is the unique solution regardless of $\boldsymbol{L}$. This is due to the fact that $(\boldsymbol{L}|_{\ker(\boldsymbol{M})})^\dag$ is the zero operator (because $\boldsymbol{L}|_{\ker(\boldsymbol{M})}$ is the zero operator) and $\boldsymbol{M}^\dag=(\boldsymbol{M})^{-1}$ so that $\phi_{\boldsymbol{L}}=(\boldsymbol{M})^{-1}\boldsymbol{N}I_n$, the unique solution.

Recalling that $\phi_{\boldsymbol{L}}$ is a solution to the LREM, note how the mapping from parameters to the set of all solutions, $(M,N)\mapsto\phi_{\boldsymbol{L}}+\ker(\boldsymbol{M})$, compares to the mapping from parameters to the set of regularized solutions $(M,N)\mapsto\phi_{\boldsymbol{L}}+\ker(\boldsymbol{M})\cap\ker(\boldsymbol{L})$. The difference is simply a matter of dimension reduction. It is important to note that neither mapping is generally one-to-one and solution sets associated with different parameters may intersect.

\begin{thm}\label{thm:tykhonov}
If $(M,N)$ is an LREM, $\det(M(z))\neq0$ for all $z\in\mathbb{T}$, the partial indices of $M$ are non-negative, and $\boldsymbol{L}:H^m\rightarrow H^l$ is a bounded linear operator, then for every covariance stationary process $\xi$ there exists a solution $X$ minimizing $\|\boldsymbol{L}\phi\|_{H^l}$ given by
\begin{align*}
X_t=\int z^t (\boldsymbol{I}-(\boldsymbol{L}|_{\ker(\boldsymbol{M})})^\dag \boldsymbol{L})\boldsymbol{M}^\dag\boldsymbol{N}I_n d\Phi,\quad t\in\mathbb{Z}.
\end{align*}
The solution is unique if and only if $\ker(\boldsymbol{M})\cap\ker(\boldsymbol{L})=\{0\}$.
\end{thm}
\begin{proof}
Follows from Lemma \ref{lem:reg} and the spectral representation theorem.
\end{proof}

The expression for regularized solutions in Theorem \ref{thm:tykhonov} is primarily of theoretical interest. It will allow us to study continuity and differentiability with respect to underlying parameters. For estimation and inference, on the other hand, \cite{regular} provides a numerical algorithm for computing regularized solutions in the \cite{sims} framework, which leads to equivalent algebraic rather than geometric criteria for existence and uniqueness.

We have shown that, with an appropriate choice of $\boldsymbol{L}$, the regularized solution can overcome the non-uniqueness problem. Our next result shows that regularized solutions also overcome the discontinuity problem. The continuity guarantee that we require for mainstream econometric methodology is not with respect to the $H^m$ norm but with respect to the $\mu$--essential supremum norm; see e.g.\ the continuity results for bounded spectral densities and the Gaussian likelihood functions in \citep{hannan,dp,anderson}. If we parametrize $M$ and $N$ as $M(z,\theta)$ and $N(z,\theta)$, then it is clear that we need $M(z,\theta)$ and $N(z,\theta)$ to be jointly continuous in $z$ and $\theta$. However, \cite{ga} note that this is not sufficient to ensure continuity of the Wiener-Hopf factors in the $\mu$--essential supremum norm. Thus, we use Green and Anderson's idea of imposing control over $\frac{d}{dz}M(z,\theta)$ and $\frac{d}{dz}N(z,\theta)$.

\begin{thm}\label{thm:tykcont}
Let $M:\mathbb{T}\times\Theta\rightarrow\mathbb{C}^{m\times m}$ and $N:\mathbb{T}\times\Theta\rightarrow\mathbb{C}^{m\times n}$. Under the conditions
\begin{enumerate}
\item $F=\mu I_n$.
\item $\Theta\subset\mathbb{R}^d$ is an open set.
\item $M(\,\cdot\,,\theta)$ and $N(\,\cdot\,,\theta)$ are analytic in a neighborhood of $\mathbb{T}$ for every $\theta\in\Theta$.
\item $M(z,\theta)$, $N(z,\theta)$, $\frac{d}{dz}M(z,\theta)$, and $\frac{d}{dz}N(z,\theta)$ are jointly continuous at every $(z,\theta)\in\mathbb{T}\times\Theta$.
\item $\ker(\boldsymbol{L})\cap\ker(\boldsymbol{M}(\theta))=\{0\}$ for all $\theta\in\Theta$.
\item $\det(M(z,\theta))\neq0$ for all $z\in\mathbb{T}$, and the partial indices of $M(z,\theta)$ are all non-negative for all $\theta\in\Theta$.
\end{enumerate}
Then $\phi_{\boldsymbol{L}}(\theta)=(\boldsymbol{I}-(\boldsymbol{L}|_{\ker(\boldsymbol{M}(\theta))})^\dag \boldsymbol{L})\boldsymbol{M}(\theta)^\dag\boldsymbol{N}(\theta)I_n
$, is continuous in the $\mu$--essential supremum norm.
\end{thm}
\begin{proof}
See Appendix \ref{sec:proofs}.
\end{proof}

The assumptions of Theorem \ref{thm:tykcont} are quite strong relative to the discussion so far. However, the relevant case for most macroeconometric applications is the case where 
$F=\mu I_n$, while $M(z,\theta)$ and $N(z,\theta)$ are Laurent matrix polynomials of uniformly bounded degree. In this case, the continuity of the coefficients of $M(z,\theta)$ and $N(z,\theta)$ is sufficient to ensure conditions (iii) and (iv) of Theorem \ref{thm:tykcont}. The New-Keynesian model \eqref{eq:nk} satisfies these conditions, as do all of the LREMs in \cite{canova}, \cite{dd}, or \cite{hs}. 

For the purpose of establishing asymptotic normality, we typically need not just continuity in the essential supremum norm but also differentiability in the essential supremum norm (i.e. the finite differential in $\theta$ converges to the infinitesimal differential in the essential supremum norm over $z\in\mathbb{T}$). This stronger form of differentiability allows us to differentiate under integrals, which appear in the asymptotics of maximum likelihood and generalized method of moments estimators. The following result provides exactly what we need.

\begin{thm}\label{thm:tykdiff}
Let $p$ be a positive integer. Under assumptions (i) -- (vi) of Theorem \ref{thm:tykcont} and,
\begin{enumerate}
\item[(vii)] $M(z,\theta)$, $N(z,\theta)$, $\frac{d}{dz}M(z,\theta)$, and $\frac{d}{dz}N(z,\theta)$ are jointly continuously differentiable of all orders up to $p$ for all $(z,\theta)\in\mathbb{T}\times\Theta$.
\end{enumerate}
Then $\phi_{\boldsymbol{L}}(\theta)=(\boldsymbol{I}-(\boldsymbol{L}|_{\ker(\boldsymbol{M}(\theta))})^\dag \boldsymbol{L})\boldsymbol{M}(\theta)^\dag\boldsymbol{N}(\theta)I_n$, is continuously differentiable of order $p$ with respect to $\theta$ in the $\mu$--essential supremum norm.
\end{thm}
\begin{proof}
See Appendix \ref{sec:proofs}.
\end{proof}

Condition (vii) of Theorem \ref{thm:tykdiff} is a direct strengthening of condition (iv) of Theorem \ref{thm:tykcont}. When $M(z,\theta)$ and $N(z,\theta)$ are Laurent matrix polynomials of uniformly bounded degree, as is usually the case in macroeconomic models, the $p$-th order continuous differentiability of the coefficients of $M(z,\theta)$ and $N(z,\theta)$ is sufficient to ensure this condition is satisfied.

To see Theorem \ref{thm:tykdiff} in action, consider the regularized solution to the example from the previous subsection. This can be obtained analytically in  just a handful of steps.
\begin{align*}
\phi_{\boldsymbol{I}}(\theta)&=\boldsymbol{M}(\theta)^\ast(\boldsymbol{M}(\theta)\boldsymbol{M}(\theta)^\ast)^{-1}I_2\\
&=\left[\begin{array}{cc}
\boldsymbol{V}^{(-2)} & \theta \boldsymbol{V}^{(-1)} \\
\boldsymbol{0} & \boldsymbol{I}
\end{array}\right]\left(
\left[\begin{array}{cc}
\boldsymbol{V}^2 & \boldsymbol{0} \\
\theta \boldsymbol{V} & \boldsymbol{I}
\end{array}\right]\left[\begin{array}{cc}
\boldsymbol{V}^{(-2)} & \theta \boldsymbol{V}^{(-1)} \\
\boldsymbol{0} & \boldsymbol{I}
\end{array}\right]
\right)^{-1}\left[\begin{array}{cc}
1 & 0 \\ 0 & 1
\end{array}\right]\\
&=\left[\begin{array}{cc}
\boldsymbol{V}^{(-2)} & \theta \boldsymbol{V}^{(-1)} \\
\boldsymbol{0} & \boldsymbol{I}
\end{array}\right]\left[\begin{array}{cc}
\boldsymbol{I} & \theta \boldsymbol{V} \\
\theta \boldsymbol{V}^{(-1)} & (\theta^2+1)\boldsymbol{I}
\end{array}\right]^{-1}\left[\begin{array}{cc}
1 & 0 \\ 0 & 1
\end{array}\right].
\end{align*}
To invert the operator in the expression above, we note that it is of the form defined in Definition \ref{defn:MV}, with underlying $\mathcal{W}^{2\times 2}$ element, $\left[\begin{smallmatrix} 1 & \theta z \\ \theta z^{-1} & (\theta^2+1)\end{smallmatrix}\right]$. This matrix has the Wiener-Hopf factorization $\left[\begin{smallmatrix} 1 & \frac{\theta}{\theta^2+1}z \\ 0 & 1\end{smallmatrix}\right]\left[\begin{smallmatrix} 1 & 0\vspace{0.1cm} \\ 0 & 1\end{smallmatrix}\right]\left[\begin{smallmatrix} \frac{1}{\theta^2+1} & 0 \\ \theta z^{-1} & \theta^2+1\end{smallmatrix}\right]$. Thus,
\begin{align*}
\phi_{\boldsymbol{I}}(\theta)&=\left[\begin{array}{cc}
\boldsymbol{V}^{(-2)} & \theta \boldsymbol{V}^{(-1)} \\
\boldsymbol{0} & \boldsymbol{I}
\end{array}\right]\left(\left[\begin{array}{cc}
\boldsymbol{I} & \frac{\theta}{\theta^2+1} \boldsymbol{V} \\
\boldsymbol{0} & \boldsymbol{I}
\end{array}\right]\left[\begin{array}{cc}
\frac{1}{\theta^2+1}\boldsymbol{I} & \boldsymbol{0} \\
\theta \boldsymbol{V}^{(-1)} & (\theta^2+1)\boldsymbol{I}
\end{array}\right]\right)^{-1}\left[\begin{array}{cc}
1 & 0 \\ 0 & 1
\end{array}\right]\\
&=\left[\begin{array}{cc}
\boldsymbol{V}^{(-2)} & \theta \boldsymbol{V}^{(-1)} \\
\boldsymbol{0} & \boldsymbol{I}
\end{array}\right]\left(\left[\begin{array}{cc}
(\theta^2+1)\boldsymbol{I} & \boldsymbol{0} \\
-\theta \boldsymbol{V}^{(-1)} & \frac{1}{\theta^2+1}\boldsymbol{I}
\end{array}\right]\left[\begin{array}{cc}
\boldsymbol{I} & -\frac{\theta}{\theta^2+1} \boldsymbol{V} \\
\boldsymbol{0} & \boldsymbol{I}
\end{array}\right]\right)\left[\begin{array}{cc}
1 & 0 \\ 0 & 1
\end{array}\right]\\
&=\left[\begin{array}{cc}
\boldsymbol{V}^{(-2)} & -\frac{\theta}{1+\theta^2}\boldsymbol{V}^{(-2)}\boldsymbol{V}+\frac{\theta}{1+\theta^2}\boldsymbol{V}^{(-1)}\\
-\theta\boldsymbol{V}^{(-1)} & \frac{1}{1+\theta^2}\boldsymbol{I}+\frac{\theta^2}{1+\theta^2}\boldsymbol{V}^{(-1)}\boldsymbol{V}
\end{array}\right]\left[\begin{array}{cc}
1 & 0 \\ 0 & 1
\end{array}\right]\\
&=\left[\begin{array}{cc}
z^{-2} & \frac{\theta}{1+\theta^2}z^{-1} \\ -\theta z^{-1} & \frac{1}{1+\theta^2}
\end{array}\right],
\end{align*}
where the last equality follows from the fact that $\ker(\boldsymbol{V})=\mathbb{C}^{1\times 2}$ when $F=\mu I_2$. As guaranteed by Theorems \ref{thm:tykcont} and \ref{thm:tykdiff}, this solution is not just continuous as a function of $\theta$ but also smooth. Note that the implicit right inverse of $\boldsymbol{M}_0$ in $\boldsymbol{M}^\dag$ is
\begin{align*}
\boldsymbol{M}_-(\theta)\boldsymbol{M}(\theta)^\dag\boldsymbol{M}_+(\theta)=\begin{cases}
\left[\begin{array}{cc}
\boldsymbol{V}^{(-2)} & 0 \\ 0 & 1
\end{array}\right],& \theta=0,\vspace{0.5cm}\\
\left[\begin{array}{cc}
\boldsymbol{V}^{(-1)} & \frac{1}{1+\theta^2}\left(\boldsymbol{V}^{(-1)}\boldsymbol{V}-\boldsymbol{I}\right) \\ \boldsymbol{0} & \boldsymbol{V}^{(-1)}\end{array}\right],& \theta\neq0.
\end{cases}
\end{align*}
Thus, it is only by using a special right inverse of $\boldsymbol{M}_0(\theta)$ that we can absorb the effect of discontinuity in the partial indices at $\theta=0$ on the solution.

To summarize, the LREM problem in macroeconometrics is ill-posed. However, regularization produces solutions that are unique, continuous, and even smooth under very general regularity conditions.

\section{Application}\label{sec:application}
As an application of the spectral approach to LREMs, we apply mainstream methodology to estimate and draw inference on the non-generic model of Section \ref{sec:regular}. We will see that the Gaussian likelihood function displays very irregular behavior, invalidating underlying assumptions of mainstream frequentist and Bayesian analysis. In turn, regularized solutions can avoid some of these anomalies.

Consider the non-generic system of Section \ref{sec:regular} in the time domain,
\begin{align*}
\begin{split}
\boldsymbol{P}(X_{1t+2}|\mathscr{H}_t)&=\xi_{1t}\\
\theta_0 \boldsymbol{P}(X_{1t+1}|\mathscr{H}_t)+X_{2t}&=\xi_{2t}
\end{split}& &t\in\mathbb{Z},
\end{align*}
with $\xi$ an i.i.d.\ Gaussian sequence of mean zero and covariance matrix $I_2$. The objective is to estimate and draw inference on $\theta_0\in\mathbb{R}$ from an observed dataset $x_T=(X_1^\ast,\ldots,X_T^\ast)^\ast$. We will follow mainstream methodology in disregarding uniqueness, identifiability, and invertibility in parametrizing the model (see \cite{ident} for a careful parametrization that addresses all of these considerations). The solution we computed in the previous section is
\begin{align*}
X_t=
\begin{cases}
\vspace{0.5cm}\left[\begin{array}{c}
\xi_{1t-2}\\
\xi_{2t}
\end{array}\right],& \theta_0=0,\\
\left[\begin{array}{c} \xi_{1t-2}+\theta_0^{-1}\xi_{2t-1}\\
-\theta_0\xi_{1t-1}\end{array}\right],& \theta_0\neq0,
\end{cases}& &t\in\mathbb{Z}.
\end{align*}
Let $\Sigma_T(\theta)$ be the covariance of $x_T$ according to the model parametrized by $\theta$. Then the likelihood function is
\begin{align*}
p(x_T|\theta)=(2\pi)^{-T}\det(\Sigma_T(\theta))^{-\frac{1}{2}}\exp\left\{-\frac{1}{2}x_T^\ast\Sigma_T(\theta)^{-1}x_T\right\}.
\end{align*}
It is more convenient to work with
\begin{align*}
\ell_T(\theta)=\frac{1}{T}\log\det\Sigma_T(\theta)+\frac{1}{T}x_T^\ast\Sigma_T(\theta)^{-1}x_T.
\end{align*}
We will refer to $\ell_T(\theta)$ as the log-likelihood function. Observe that
\begin{align*}
\Sigma_T(\theta)=\begin{cases}
I_{2T},& \theta=0\\
\left[\begin{array}{cccccccccc}
G_0 & G_1 & 0 & 0 & \cdots & 0\\
G_1^\ast & G_0 & G_1 & 0 & \cdots & 0\\
0 & G_1^\ast & G_0 & G_1 & \ddots & \vdots \\
\vdots & \ddots & \ddots & \ddots  & \ddots & 0\\
0 & \cdots & 0 & G_1^\ast & G_0 & G_1\\
0 & \cdots & 0 & 0 & G_1^\ast & G_0
\end{array}\right],& \theta\neq0,
\end{cases}
\end{align*}
where
\begin{align*}
G_0=\left[\begin{array}{cc}
1+\theta^{-2} & 0 \\ 0 & \theta^2
\end{array}\right],& &G_1=\left[\begin{array}{cc}
0 & 0 \\ -\theta & 0
\end{array}\right].
\end{align*}
Then $\Sigma_T(0)^{-1}=I_{2T}$ and it is easily checked that for $\theta\neq0$,
\begin{align*}
\Sigma_1(\theta)^{-1}&=\left[\begin{array}{cc} \frac{\theta^2}{1+\theta^2} & 0 \\ 0 & \theta^{-2} \end{array}\right]\\
\Sigma_T(\theta)^{-1}&=\left[\begin{array}{cccccccccc}
A & B & 0 & 0 & \cdots & 0\\
B^\ast & C & B & 0 & \cdots & 0\\
0 & B^\ast & C & B & \ddots & \vdots \\
\vdots & \ddots & \ddots & \ddots  & \ddots & 0\\
0 & \cdots & 0 & B^\ast & C & B\\
0 & \cdots & 0 & 0 & B^\ast & D
\end{array}\right],& &T>1,
\end{align*}
where
\begin{align*}
A=\left[\begin{array}{cc}
\frac{\theta^2}{1+\theta^2} & 0 \\ 0 & \frac{\theta^2+1}{\theta^2}
\end{array}\right],& &B=\left[\begin{array}{cc}
0 & 0 \\ \theta & 0
\end{array}\right],& &C=\left[\begin{array}{cc}
\theta^2 & 0 \\ 0 & \frac{\theta^2+1}{\theta^2}
\end{array}\right],& & D=\left[\begin{array}{cc}
\theta^2 & 0 \\ 0 & \theta^{-2}
\end{array}\right].
\end{align*}
It is also easily checked that $\det(\Sigma_T(\theta))=1+\theta^2$ for $\theta\neq0$ for all $T\geq1$. Substituting into the log-likelihood function we find that
\begin{align*}
\ell_T(0)=\frac{1}{T}\sum_{t=1}^TX_t^\ast X_t
\end{align*}
and, for $\theta\neq0$,
\begin{align*}
\ell_T(\theta)=\begin{cases}
\log(1+\theta^2)+\frac{\theta^2}{1+\theta^2}X_{11}^2+\theta^{-2}X_{21}^2,& T=1\\
\frac{1}{T}\log(1+\theta^2)+\frac{1}{T}X_1^\ast AX_1+\frac{1}{T}\sum_{t=2}^{T-1}X_t^\ast CX_t+\frac{2}{T}\sum_{t=1}^{T-1}X_t^\ast BX_{t+1}+\frac{1}{T}X_T^\ast DX_T,& T>1.
\end{cases}
\end{align*}
It is clear that
\begin{align*}
\lim_{0\neq|\theta|\rightarrow0}\ell_T(\theta)=
\lim_{|\theta|\rightarrow\infty}\ell_T(\theta)=\infty,& &T\geq1,
\end{align*}
almost surely because $x_T$ has a continuous distribution regardless of the value of $\theta_0$. Since $\ell_T(\theta)$ is continuous on $\mathbb{R}\backslash\{0\}$ and bounded from below, $\ell_T(\theta)$ must attain local minima somewhere to the left and to the right of $\theta=0$ with probability 1. Thus, almost surely, $\ell_T(\theta)$ is $W$ shaped, diverging as $0\neq|\theta|\rightarrow0$ and as $|\theta|\rightarrow\infty$, although it has a finite value at $\theta=0$. See Figure \ref{fig:ng1} for an illustration with simulated data. It follows that the likelihood function $p(x_T|\theta)$ almost surely has at least three local maxima, one of which is $\theta=0$. Moreover, almost surely $\lim_{0\neq|\theta|\rightarrow0}p(x_T|\theta)=0$ and $p(x_T|\theta)|_{\theta=0}>0$ so the likelihood function almost surely has a simple discontinuity at $\theta=0$.

\begin{figure}[t!]
\caption{Log-likelihood of Non-regularized Solutions to the Non-generic Model.}
\label{fig:ng1}
\smallskip
\center
\includegraphics[width=8.5cm]{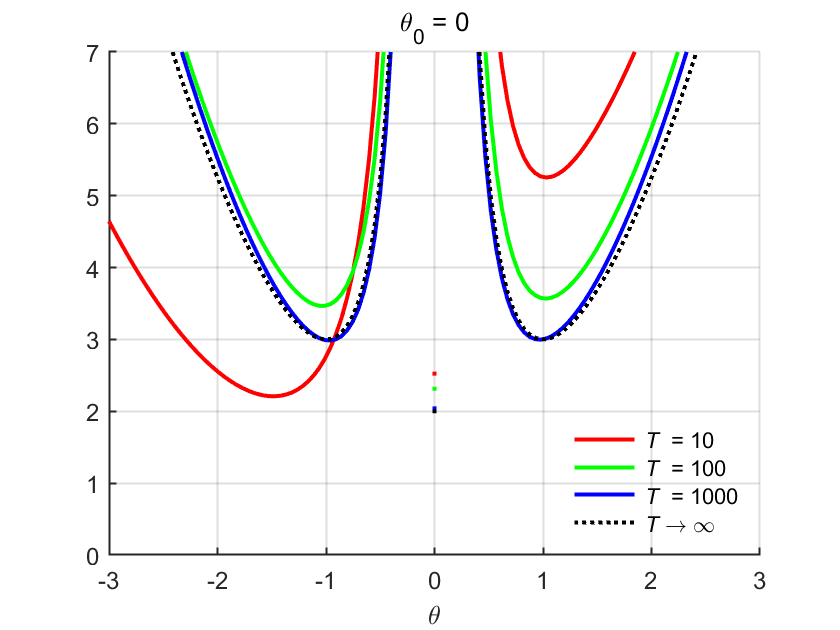}
\includegraphics[width=8.5cm]{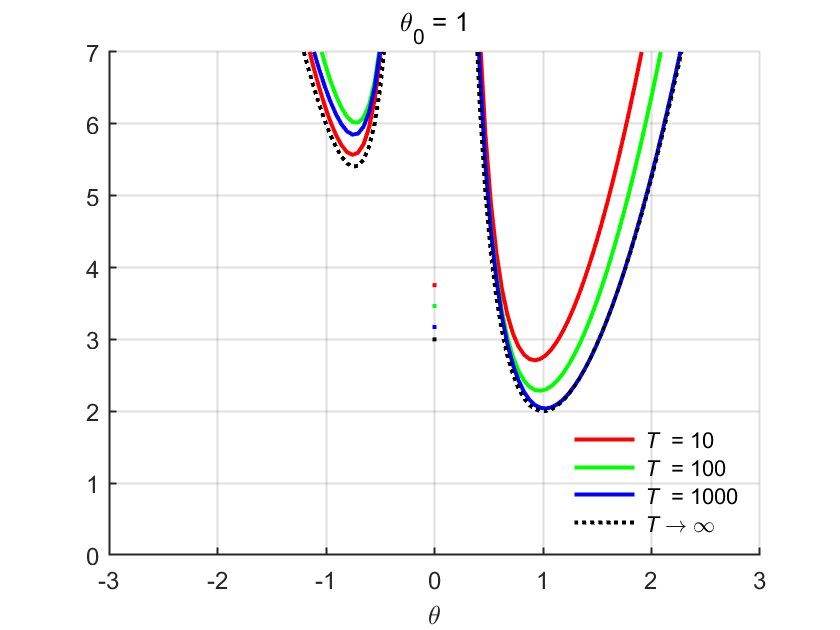}
\end{figure}

Taking the limit $T\rightarrow\infty$ is straightforward as most terms vanish. By the strong law of large numbers, almost surely,
\begin{align*}
\lim_{T\rightarrow\infty}\ell_T(\theta)=\begin{cases}
\mathrm{tr}(\mathbb{E}(X_0X_0^\ast)),& \theta=0\\
\mathrm{tr}(C\mathbb{E}(X_0X_0^\ast )+2B\mathbb{E}(X_1 X_0^\ast)),& \theta\neq0.
\end{cases}
\end{align*}
The moments above can be read from $G_0$ and $G_1$ evaluated at $\theta=\theta_0$ so the right hand side is
\begin{align*}
\ell(\theta,\theta_0)=\begin{cases}
2,& \theta=0,\,\theta_0=0\\
\theta_0^2+1+\theta_0^{-2},& \theta=0,\,\theta_0\neq0\\
\theta^2+1+\theta^{-2},& \theta\neq0,\,\theta_0=0\\
\theta^2(1+\theta_0^{-2})+(1+\theta^{-2})\theta_0^2-2\theta\theta_0,& \theta\neq0,\,\theta_0\neq0.
\end{cases}
\end{align*}
See Figure \ref{fig:ng1} for plots of $\ell(\theta,0)$ and $\ell(\theta,1)$. It is easy to check that $\ell(\theta,0)$ has local minima at $0$, $+1$, and $-1$ and if $\theta_0\neq0$, $\ell(\theta,\theta_0)$ has local minima at $0$, $\theta_0$, and a third point of the opposite sign to $\theta_0$ (the expression is complicated so we omit it). Note that for $\theta\in\mathbb{R}\backslash\{0\}$, $\ell_T(\theta)-\ell(\theta,\theta_0)$ is expressible as a sum of eight terms, each of which is a product of a continuous function in $\theta$ and a quantity that converges almost surely to zero. Thus, $\lim_{T\rightarrow\infty}\sup_{\theta\in\Theta}|\ell_T(\theta)-\ell(\theta,\theta_0)|=0$ for any compact subset $\Theta\subset\mathbb{R}\backslash\{0\}$.

Consider the case $\theta_0=0$. It is clear from the discussion above that the maximum likelihood estimator is consistent, although it is not asymptotically normal, and any prior that has an atom at the point $\theta=0$ will yield a posterior distribution that concentrates at the true value. However, this analytical approach cannot be carried over to a general methodology because analytical descriptions of solutions are generally infeasible for all but the simplest LREMs; the mapping from $(M,N)$ to any solution is highly non-linear. The practitioner who seeks to estimate this model on a computer may suspect a discontinuity at $\theta=0$ because they will be able to see the solution behaving erratically near that point but being short of an analytical description of the solution: (i) they will not have any theoretical guarantees as to what happens when $\theta\rightarrow0$ because none exist and (ii) they will not be able to compute the solution at $\theta=0$ reliably due to the numerical instability of computing Wiener-Hopf factorizations of non-generic elements of $\mathcal{W}^{m\times m}$.

Now consider the likely outcome of applying mainstream econometric methodology to the data when $\theta_0=0$.

Let us begin with the most basic of empirical analyses: plotting the likelihood function. Each evaluation of the likelihood function is obtained by solving the model numerically then using the Kalman filter \citep{canova,dd,hs}. Without knowing a priori the importance of the likelihood function at $\theta=0$, the researcher will not know to include that point specifically in the list plot of the likelihood function. Even if they did include it, the plot is not guaranteed to reveal the value of the likelihood function at this point because of the numerical instability of the computation at this point. Thus, the plot of the likelihood function will appear to have modes near $\pm1$, while the population parameter, $\theta_0=0$, may (and most likely will) appear to be the least likely point of the parameter space because $\lim_{0\neq|\theta|\rightarrow0}p(x_T|\theta)=0$.

Mainstream frequentist analysis utilizes variations of the Newton-Raphson algorithm to minimize $\ell_T(\theta)$, then obtains confidence intervals from derivatives of $\ell_T(\theta)$ \citep{canova,dd}. This exercise is almost certain to fail in our setting because, again, without knowing a priori that the point $\theta=0$ is crucial to the analysis, the algorithm will be initialized in $\mathbb{R}\backslash\{0\}$ and it will converge almost surely to (approximately) either $+1$ or $-1$ depending on where it is initialized. The standard error of the estimate, computed using the second derivative of $\ell_T(\theta)$ at the estimated $\theta$, will be of order $O(T^{-1/2})$ almost surely so as the sample size gets larger it will appear to the researcher that the estimate is precise when it is in fact completely wrong.

Mainstream Bayesian analysis utilizes variations on the random walk Metropolis-Hastings algorithm to sample from the posterior distribution function,
\begin{align*}
p(\theta|x_T)\propto p(x_T|\theta)p(\theta),
\end{align*}
where $p(\theta)$ is a prior distribution function, before reporting estimates of the mean or median of the posterior distribution as well as credibility regions \citep{canova,dd,hs}. Here, again, estimation is certain to fail because, being unaware of the significance of the point $\theta=0$, the researcher is likely to use a continuous prior that places probability zero at the set $\{0\}$. The result is that the posterior distribution will concentrate around $\pm1$. The only way to allow the prior to be informative about $\theta$ is to allow for an atom in the prior at $\theta=0$ but none of the aforementioned textbooks present algorithms that can accommodate posterior distributions with atoms.

In contrast to the above, the regularized solution,
\begin{align*}
\tilde{X}_{1t}&=\xi_{1t-2}+\frac{\theta_0}{1+\theta_0^2}\xi_{2t-1}\\
\tilde{X}_{2t}&=-\theta_0\xi_{1t-1}+\frac{1}{1+\theta_0^2}\xi_{2t}
\end{align*}
displays no such anomalies. If $\tilde{x}_T=(\tilde{X}_1^\ast,\ldots,\tilde{X}_T^\ast)^\ast$ and $\tilde{\Sigma}_T(\theta)$ is the covariance of $\tilde x_T$ according to the model parametrized by $\theta$. Then the $ij$-th block of $\tilde{\Sigma}_T(\theta)$ is of the form $\int z^{i-j}\phi_{\boldsymbol{I}}(\theta)\phi_{\boldsymbol{I}}(\theta)^\ast d\mu$. Since $\phi_{\boldsymbol{I}}(\theta)$ is continuously differentiable of any order with respect to $\theta$ in the $\mu$--essential supremum norm, we can exchange the integration with differentiation with respect to $\theta$ so $\tilde{\Sigma}_T(\theta)$ is continuously differentiable of any order with respect to $\theta$. This implies that the same is true of the likelihood function at every point $\theta\in\mathbb{R}$ where $\tilde{\Sigma}_T(\theta)^{-1}$ exists. It is easy to check that $\tilde{\Sigma}_T(\theta)$ and $\tilde{\Sigma}_T(\theta)^{-1}$ have the same block structure as $\Sigma_T(\theta)$ and $\Sigma_T(\theta)^{-1}$ respectively. The only difference is that now
\begin{align*}
G_0&=\left[
\begin{array}{cc}
 \frac{\theta ^2}{\left(\theta ^2+1\right)^2}+1 & 0 \\
 0 & \theta ^2+\frac{1}{\left(\theta ^2+1\right)^2} \\
\end{array}
\right],& G_1&=\left[
\begin{array}{cc}
 0 & 0 \\
 \theta  \left(\frac{1}{\left(\theta ^2+1\right)^2}-1\right) & 0 \\
\end{array}
\right],\\
A&=\left[
\begin{array}{cc}
 \frac{(\theta^2+1)^2}{\theta^4+3\theta^2+1} & 0 \\
 0 & \frac{\theta ^4+3 \theta ^2+1}{\left(\theta ^2+1\right)^2} \\
\end{array}
\right],& B&=\left[
\begin{array}{cc}
 0 & 0 \\
 \frac{\theta ^3 \left(\theta ^2+2\right)}{\left(\theta ^2+1\right)^2} & 0 \\
\end{array}
\right],\\
C&=\left[
\begin{array}{cc}
 \frac{\theta ^6+2 \theta ^4+\theta ^2+1}{\left(\theta ^2+1\right)^2} & 0 \\
 0 & \frac{\theta ^4+3 \theta ^2+1}{\left(\theta ^2+1\right)^2} \\
\end{array}
\right],& D&=\left[
\begin{array}{cc}
 \frac{\theta ^6+2 \theta ^4+\theta ^2+1}{\left(\theta ^2+1\right)^2} & 0 \\
 0 &  \frac{(\theta^2+1)^2}{\theta^4+3\theta^2+1}
\end{array}
\right].
\end{align*}
Figure \ref{fig:ng2} provides a plot of the log-likelihood function of the regularized solution with simulated data along with the almost sure limit in $T$. It is clear now that frequentist and Bayesian analyses can be carried out using the regularized solution. Note that the first, second, and third derivatives of the limiting log-likelihood function vanish at $\theta=0$ when $\theta_0=0$. Thus a slightly more delicate frequentist analysis is necessary (see e.g.\ \cite{rotnitzky}). This is, again, made possible by the existence of derivatives of all orders, thanks to regularization.

\begin{figure}[t!]
\caption{Log-likelihood of Regularized Solutions to the Non-generic Model.}
\label{fig:ng2}
\smallskip
\center
\includegraphics[width=8.5cm]{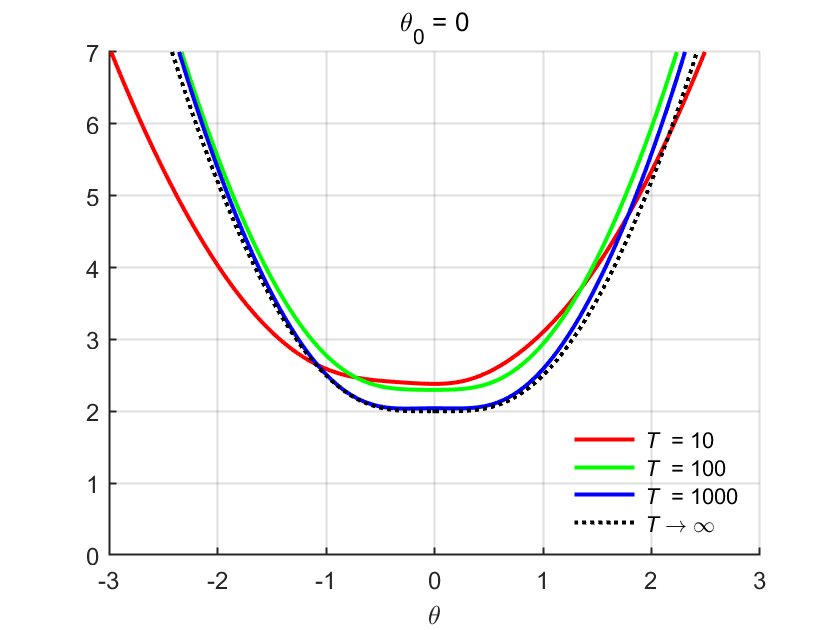}
\includegraphics[width=8.5cm]{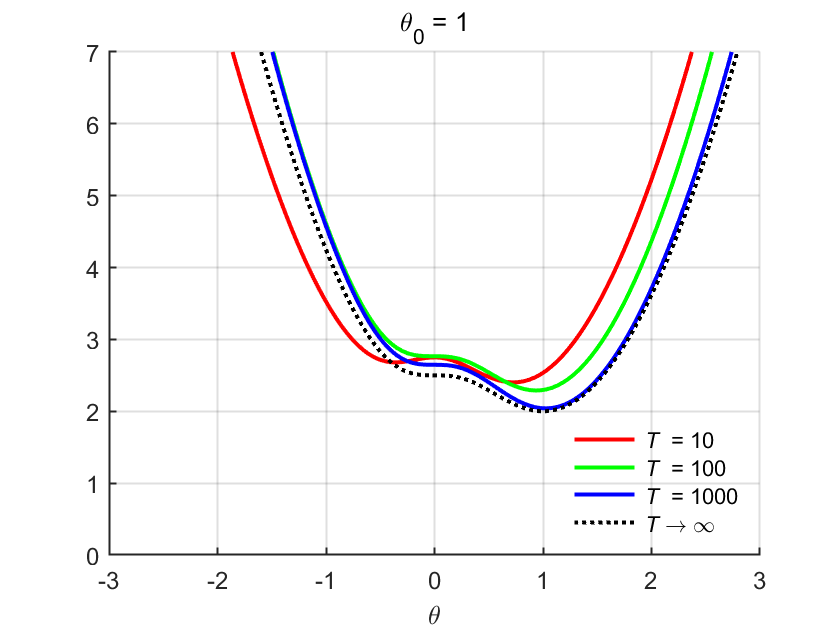}
\end{figure}

\section{Conclusion}\label{sec:conclusion}
This paper has extended the LREM literature in the direction of spectral analysis. It has done so by relaxing common assumptions and developing a new regularized solution. The spectral approach has allowed us to study examples of limiting Gaussian likelihood functions of simple LREMs, which demonstrate the advantages of the new regularized solution as well as highlighting weaknesses in mainstream methodology. For the remainder, we consider some implications for future work.

The regularized solution proposed in this paper is the natural one to consider for the frequency domain. However, its motivation has been entirely econometric in nature and this begs the question of whether it can be derived from decision-theoretic foundations as proposed by \cite{taylor}. Regularization has already made inroads into decision theory \citep{gabaix}. This line of inquiry may yield other forms of regularization which may have more interesting dynamic or statistical properties.

The parameter space in LREMs can be disconnected and it can matter a great deal where one initializes their optimization routine to find the maximum likelihood estimator or their exploration routine for sampling from the posterior distribution. Therefore, it would be useful to develop simple preliminary estimators of LREMs analogous to the results for VARMA (e.g.\ Sections 8.4 and 11.5 of \cite{bd}) that can provide good initial conditions for frequentist and Bayesian algorithms.

Wiener-Hopf factorization theory has been demonstrated here and in previous work (\cite{onatski}, \cite{linsys}, \cite{ident}) to be the appropriate mathematical framework for analysing LREMs. This begs the question of what is the appropriate framework for non-linear rational expectations models. The hope is that the mathematical insights from the theory of linear models will allow for important advances in non-linear modeling and inference.

Finally, researchers often rely on high-level assumptions as tentative placeholders when a result seems plausible but a proof from first principles is not apparent. Continuity of solutions to LREMs with respect to parameters has for a long time been one such high-level assumption in the LREM literature. The fact that it is generally false, should give us pause to reflect on the prevalence of this technique. At the same time, the author hopes to have conveyed a sense of optimism that theoretical progress from first principles is possible.

\section*{Appendix}

\appendix

\section{Parametrizing the Set of Solutions}\label{sec:setsol}
By Lemma \ref{lem:existencephi}, the set of solutions to \eqref{eq:lremfreqcompact} when $\kappa_m\geq0$ is the affine space $\boldsymbol{M}^{(-1)}\boldsymbol{N}I_n+\ker(\boldsymbol{M})$, where
\begin{align*}
\ker(\boldsymbol{M})=\boldsymbol{M}_-^{-1}\ker(\boldsymbol{M}_0).
\end{align*}
It suffices to parametrize
\begin{align*}
\ker(\boldsymbol{M}_0)=\left\{\left[\begin{smallmatrix}
\psi_1\\
\vdots\\
\psi_m
\end{smallmatrix}\right]: \psi_i\in \ker(\boldsymbol{V}^{\kappa_i}),\; i=1,\ldots, m\right\}.
\end{align*}
The problem then reduces to the parametrization of $\ker(\boldsymbol{V}^\kappa)$ for $\kappa\geq0$ and, by Lemma \ref{lem:V} (viii), the problem reduces even further to the parametrization of $\ker(\boldsymbol{V})$. Now if $\ker(\boldsymbol{V})=\{0\}$, then $\ker(\boldsymbol{V}^\kappa)=\{0\}$ for all $\kappa\geq0$ and so $\ker(\boldsymbol{M}_0)=\{0\}$. If $\ker(\boldsymbol{V})\neq\{0\}$, we may find an orthonormal basis for it,
\begin{align*}
\Upsilon_1,\ldots, \Upsilon_r\in\ker(\boldsymbol{V}).
\end{align*}
Note that $r\leq n$ by Lemma \ref{lem:V} (vi). Lemma \ref{lem:V} (viii) then implies that
\begin{align*}
\Upsilon_1,\ldots, \Upsilon_r, \boldsymbol{V}^{(-1)}\Upsilon_1,\ldots, \boldsymbol{V}^{(-1)}\Upsilon_r,\ldots, \boldsymbol{V}^{(1-\kappa)}\Upsilon_1,\ldots, \boldsymbol{V}^{(1-\kappa)}\Upsilon_r
\end{align*}
is an orthonormal basis for $\ker(\boldsymbol{V}^\kappa)$ for $\kappa\geq0$. Thus,
\begin{align*}
\ker(\boldsymbol{M}_0)=\left\{K\Upsilon: K_{ij}(z)=\sum_{s=0}^{\kappa_i-1}K_{sij}z^{-s}, K_{sij}\in\mathbb{C},\; i=1,\ldots,m,\; j=1,\ldots,r,\; s=0,\ldots,\kappa_i-1\right\},
\end{align*}
where
\begin{align*}
\Upsilon=\left[\begin{array}{c}
\Upsilon_1\\
\vdots\\
\Upsilon_r
\end{array}\right].
\end{align*}
Therefore, $\ker(\boldsymbol{M}_0)$ is the space obtained by pre-multiplying $\Upsilon$ by the space of $m\times n$ complex matrix polynomials in $z^{-1}$ with $i$-th row degrees bounded by $\kappa_i-1$. Notice that the rows of $K$ associated with partial indices at zero are identically zero. The dimension of $\ker(\boldsymbol{M})$ can be obtained by counting the coefficients $K_{sij}$ above, $\dim\ker(\boldsymbol{M})=r(\kappa_1+\cdots+\kappa_m)$.

\section{Relation to Previous Literature}\label{sec:previous}
This section provides a review of the approaches of \cite{whiteman}, \cite{onatski}, \cite{tanwalker}, \cite{tan}, and \cite{meyer} (henceforth, the previous literature). We will see that the previous literature makes significantly stronger assumptions about $\xi$ that also make it cumbersome to address the topic of zeros of $\det(M)$. On the other hand, the stronger assumptions do allow for closed form expressions of solutions.

The previous literature takes as its starting point the existence of a Wold decomposition for $\xi$ with no purely deterministic part. Thus, it imposes that $F$ be absolutely continuous with respect to $\mu$ and that the spectral density matrix has an analytic spectral factorization with spectral factors of fixed rank $\mu$--a.e.\ \citep[Theorem II.4.1]{rozanov}. The conditions are
\begin{gather}
dF=\Gamma\Gamma^\ast d\mu,\nonumber\\
\sum_{s=0}^\infty\Gamma_sz^{-s}\text{ converges for }|z|>1,\text{ where }\Gamma_s=\int z^s\Gamma d\mu\label{eq:wold}\\
\mathrm{rank}(\Gamma(z))=r, \quad\mu-\text{a.e. }z\in\mathbb{T}.\nonumber
\end{gather}
Theorem II.8.1 of \cite{rozanov} provides equivalent analytical conditions. This paper has demonstrated that a complete spectral theory of LREMs is possible without imposing any such restrictions.

Clearly, the spectral factorization in \eqref{eq:wold} is not unique. However, there always exists a spectral factor $\Gamma$ unique up to right multiplication by a unitary matrix such that there is an $\Upsilon\in H_0^r$ satisfying 
\begin{align*}
\Upsilon(z)\Gamma(z)=I_r,\qquad\mu-\text{a.e. }z\in\mathbb{T},
\end{align*}
\cite[Section II.4]{rozanov}. This choice of spectral factor yields a Wold representation. To see this, let
\begin{align*}
\zeta_t=\int z^t\Upsilon d\Phi,\qquad t \in\mathbb{Z}.
\end{align*}
Then 
\begin{align*}
\boldsymbol{E}\zeta_t\zeta_s^\ast=\int z^{t-s}\Upsilon dF \Upsilon^\ast=\int z^{t-s}\Upsilon\Gamma\Gamma^\ast \Upsilon^\ast d\mu=\int z^{t-s}I_rd\mu=\begin{cases}
I_r,& t=s,\\
0,& t\neq s.
\end{cases}
\end{align*}
Thus, $\zeta$ is an $r$-dimensional standardized white noise process causal in $\xi$. It follows easily that $\ker(\boldsymbol{V})=\{x^\ast\Upsilon: x\in\mathbb{C}^r\}$ so that $\dim\ker(\boldsymbol{V})=r$. Since $\Upsilon$ is a left inverse of $\Gamma$,
\begin{align*}
\|\Gamma \Upsilon-I_n\|^2_{H^n}=\mathrm{tr}\left(\int(\Gamma \Upsilon-I_n)\Gamma\Gamma^\ast(\Gamma \Upsilon-I_n)^\ast d\mu\right)=0
\end{align*}
and so
\begin{align*}
\xi_t=\int z^t d\Phi=\int z^t\Gamma\Upsilon d\Phi=\int \Gamma z^t\Upsilon d\Phi=\sum_{s=0}^\infty\Gamma_s\zeta_{t-s},\quad t\in\mathbb{Z},
\end{align*}
is a Wold representation with
\begin{align*}
\Gamma_s=\int z^s\Gamma d\mu,\qquad s\geq0.
\end{align*}

The previous literature imposes a priori the representation
\begin{align*}
X_t=\sum_{s=0}^\infty \Xi_s\zeta_{t-s},\qquad t\in\mathbb{Z}
\end{align*}
before attempting to solve for the coefficients $\Xi_s$ by complex analytic methods \citep{whiteman,tanwalker,tan,meyer} or by Wiener-Hopf factorization \citep{onatski}. This method cannot yield the correct solution if $\xi$ has a non-trivial purely deterministic part. In fact, the representation above need not be assumed a priori and can be derived as a consequence of Theorem \ref{thm:onatski1} under conditions \eqref{eq:wold}.
\begin{align*}
X_t&=\int z^t \left(\boldsymbol{M}^{(-1)}\boldsymbol{N}I_n+\boldsymbol{M}_-^{-1}\psi\right)d\Phi\\
&=\int z^t \left(\boldsymbol{M}^{(-1)}\boldsymbol{N}I_n+\boldsymbol{M}_-^{-1}\psi\right)\Gamma\Upsilon d\Phi\\
&=\int \Xi z^t \Upsilon d\Phi,\qquad t\in\mathbb{Z},
\end{align*}
where
\begin{align*}
\Xi=\left(\boldsymbol{M}^{(-1)}\boldsymbol{N}I_n+\boldsymbol{M}_-^{-1}\psi\right)\Gamma.
\end{align*}
Thus we have obtained the spectral characteristic of $X$ relative to the random measure associated with $\zeta$. It follows that $X$ is indeed representable as a moving average in $\zeta$ with coefficients,
\begin{align*}
\Xi_s=\int z^s\Xi d\mu,\qquad s\geq0.
\end{align*}

It is important to note, however, that the stronger assumptions of the previous literature lead to a more explicit expression for spectral characteristics of solutions. In particular,
\begin{align*}
\boldsymbol{N}I_n&=\boldsymbol{P}(N|H_0^m)\\
&=\boldsymbol{P}(N\Gamma\Upsilon|H_0^m)\\
&=[N\Gamma]_-\Upsilon.
\end{align*}
This follows from the fact that $\{z^t\Upsilon:t>0\}$ is orthogonal to $H_0^m$. It then follows that
\begin{align*}
\boldsymbol{M}^{(-1)}\boldsymbol{N}I_n+\boldsymbol{M}_-^{-1}\psi&=\boldsymbol{M}_-^{-1}\boldsymbol{M}_0^{(-1)}\boldsymbol{M}_+^{-1}\left([N\Gamma]_-\Upsilon\right)+\boldsymbol{M}_-^{-1}\psi\\
&=M_-^{-1}M_0^{-1}\boldsymbol{P}(M_+^{-1}[N\Gamma]_-\Upsilon|H_0^m)+M_-^{-1}\psi\\
&=M_-^{-1}M_0^{-1}[M_+^{-1}[N\Gamma]_-]_-\Upsilon+M_-^{-1}\psi\\
&=M_-^{-1}M_0^{-1}[M_+^{-1}N\Gamma]_-\Upsilon+M_-^{-1}\psi,
\end{align*}
because $[M_+^{-1}[N\Gamma]_-]_-=[M_+^{-1}N\Gamma]_--[M_+^{-1}[N\Gamma]_+]_-=[M_+^{-1}N\Gamma]_--0$. Finally, using the results of Appendix \ref{sec:setsol}, $\psi=K\Upsilon$, where $K$ is a matrix polynomial in $z^{-1}$. Thus,
\begin{align*}
\phi=M_-^{-1}M_0^{-1}[M_+^{-1}N\Gamma]_-\Upsilon+M_-^{-1}K\Upsilon
\end{align*}
and so
\begin{align*}
\Xi=M_-^{-1}M_0^{-1}[M_+^{-1}N\Gamma]_-+M_-^{-1}K.
\end{align*}

An interesting special case is when $M$, $N$, and $\Gamma$ are rational. If $M$ is rational, any Wiener-Hopf factors $M_\pm$ are rational as well \cite[Theorem I.2.1]{cg}. If $\Gamma$ is rational, the associated $\Upsilon$ can also be chosen to be rational \cite[Theorem 1]{baggioferrante}. Thus, $\phi$ and $\Xi$ are also rational.

Finally, we note that the previous literature has avoided any mention of zeros of $\det(M)$. Indeed, it is substantially more difficult to deal with zeros of $\det(M)$ without the general theory of this paper because one no longer has access to degenerate spectral measures (e.g.\ the Dirac measure) that can straightforwardly excite the instability of the system.

\section{Solving the Non-Generic System}\label{sec:nongeneric}
Consider solving the system \eqref{eq:nongeneric} by hand when $\theta\neq0$. In the time domain, this is given by
\begin{align}
\boldsymbol{P}(X_{1t+2}|\mathscr{H}_t)&=\xi_{1t}\label{eq:nongenerictime1}\\
\theta \boldsymbol{P}(X_{1t+1}|\mathscr{H}_t)+X_{2t}&=\xi_{2t},\label{eq:nongenerictime2}
\end{align}
where $\xi$ is a standard white noise process. Macroeconomic textbooks (e.g.\ \cite{sargent79}) implicitly assume the admissibility of the following elementary operations for solving LREMs.
\begin{enumerate}
\item Multiply both sides of equation $i$ by a non-zero constant.
\item Apply the mapping, $h\mapsto\boldsymbol{P}(\boldsymbol{U}^sh,|\mathscr{H}_t)$, to both sides of equation $j\neq i$ and add the resultant to equation $i$.
\item Permute equations $i$ and $j$.
\end{enumerate}
These elementary operations are a direct generalization of the familiar row-reduction elementary operations in linear algebra ($s=0$) as well as the elementary operations for VARMA manipulation ($s\leq0$). See p.\ 39 of \cite{hd}.

We seek a process $X$ causal in $\xi$ that solves \eqref{eq:nongenerictime1} and \eqref{eq:nongenerictime2}. The most immediate solution does not require any of the elementary operations above. It can be obtained by noting that, since $\boldsymbol{P}(\xi_{1t}|\mathscr{H}_t)=\xi_{1t}$, we can set
\begin{align*}
X_{1t}&=\xi_{1t-2}
\end{align*}
then substitute into equation \eqref{eq:nongenerictime2} and, noting that $\boldsymbol{P}(\xi_{1t-1}|\mathscr{H}_t)=\xi_{1t-1}$, we obtain
\begin{align*}
\theta \xi_{1t-1}+X_{2t}&=\xi_{2t}.
\end{align*}
We therefore obtain the solution
\begin{align*}
X_{1t}&=\xi_{1t-2}\\
X_{2t}&=-\theta\xi_{1t-1}+\xi_{2t},
\end{align*}
which is clearly causal in $\xi$. However, this is not the only way of solving the problem. If we multiply both sides of equation \eqref{eq:nongenerictime1} by $-\theta$ (operation (i)), we obtain
\begin{align*}
-\theta\boldsymbol{P}(X_{1t+2}|\mathscr{H}_t)=-\theta\xi_{1t}.
\end{align*}
Now apply $h\mapsto\boldsymbol{P}(\boldsymbol{U}h|\mathscr{H}_t)$ to both sides of equation \eqref{eq:nongenerictime2} and add the resultant to equation \eqref{eq:nongenerictime1} (operation (ii)) to obtain
\begin{align*}
\boldsymbol{P}(X_{2t+1}|\mathscr{H}_t)&=-\theta\xi_{1t}
\end{align*}
Thus, we can set
\begin{align*}
X_{2t}=-\theta\xi_{1t-1}
\end{align*}
Plugging this into \eqref{eq:nongenerictime2}, we obtain
\begin{align*}
\theta \boldsymbol{P}(X_{1t+1}|\mathscr{H}_t)-\theta\xi_{1t-1}&=\xi_{2t},
\end{align*}
which allows us to set (using operation (i)),
\begin{align*}
X_{1t}&=\xi_{1t-2}+\theta^{-1}\xi_{2t-1}.
\end{align*}
Finally, permuting the two equations we have obtained (operation (iii)), we obtain the alternative solution
\begin{align*}
X_{1t}&=\xi_{1t-2}+\theta^{-1}\xi_{2t-1}\\
X_{2t}&=-\theta\xi_{1t-1}.
\end{align*}
This is precisely the solution obtained in Section \ref{sec:regular}.

\section{Proofs}\label{sec:proofs}
Theorems \ref{thm:tykcont} and \ref{thm:tykdiff} require some additional technical results that we need to develop first.

Define for $\phi\in H^m$ and $\omega\in\mathbb{R}$, $\boldsymbol{S}_\omega \phi(z)=\phi(ze^{\mathrm{i}\omega})$ and, for $\omega\neq0$, $\boldsymbol{\Delta}_\omega \phi=\frac{1}{\omega}(\boldsymbol{S}_\omega\phi-\phi)$. Set $\boldsymbol{d}\phi$ to be the $H^m$ limit of $\boldsymbol{\Delta}_\omega \phi$ as $\omega\rightarrow\infty$ whenever it exists. Clearly $\boldsymbol{S}_\omega$ is a unitary operator on $H_0^m$, $\boldsymbol{\Delta}_\omega$ is a bounded operator on $H_0^m$, and $\boldsymbol{d}$ is not generally bounded on $H_0^m$.

We will need the following inequality, inspired by similar inequalities of \cite{anderson} and \cite{ga}, to prove our main results.

\begin{lem}\label{lem:ga}
If $F=\mu I_n$ and $\phi,\boldsymbol{d}\phi\in H^m$, then
\begin{align*}
\|\phi\|_\infty\leq \|\phi\|_{H^m}+\|\boldsymbol{d}\phi\|_{H^m}.
\end{align*}
\end{lem}
\begin{proof}
Let $\tau(v,w)$ be the counter-clockwise segment of $\mathbb{T}$ from $v\in\mathbb{T}$ to $w\in\mathbb{T}$ and let $1_{\tau(v,w)}$ be the indicator function for that segment. By a change of variables
\begin{align*}
\int 1_{\tau(v,w)}\boldsymbol{\Delta}_\omega \phi d\mu=\frac{1}{\omega}\int_{\tau(w,we^{\mathrm{i}\omega})}\phi d\mu-\frac{1}{\omega}\int_{\tau(v,ve^{\mathrm{i}\omega})}\phi d\mu.
\end{align*}
Since $\phi\in H^m$, $\phi$ is $\mu$--integrable so the right hand side converges to $\phi(w)-\phi(v)$ for $\mu$ -- a.e.\ $w$ and $v$, the Lebesgue points of $\phi$ \cite[Theorem 7.10]{xrudin}. On the other hand, since $\boldsymbol{\Delta}_\omega \phi$ converges in $H^m$, the continuity of the inner product implies that the left hand side converges to $\int 1_{\tau(v,w)}\boldsymbol{d}\phi d\mu$. Therefore,
\begin{align*}
\int 1_{\tau(v,w)}\boldsymbol{d}\phi d\mu=\phi(w)-\phi(v),\qquad \mu - \text{a.e. } w,v\in\mathbb{T}.
\end{align*}
It follows by the triangle and Jensen's inequality that
\begin{align*}
\|\phi(w)\|_{\mathbb{C}^{m\times n}}\leq\|\phi(v)\|_{\mathbb{C}^{m\times n}}+\int 1_{\tau(v,w)}\|\boldsymbol{d}\phi\|_{\mathbb{C}^{m\times n}} d\mu,\qquad \mu - \text{a.e. } w,v\in\mathbb{T}.
\end{align*}
Now among all Lebesgue points of $\phi$, choose a $v$ such that $\|\phi(v)\|_{\mathbb{C}^{m\times n}}^2\leq\|\phi\|_{H^m}^2$. Thus,
\begin{align*}
\|\phi(w)\|_{\mathbb{C}^{m\times n}}\leq\|\phi\|_{H^m}+\int \|\boldsymbol{d}\phi\|_{\mathbb{C}^{m\times n}} d\mu,\qquad \mu - \text{a.e. } w\in\mathbb{T}.
\end{align*}
Since $\int \|\boldsymbol{d}\phi\|_{\mathbb{C}^{m\times n}} d\mu\leq \left(\int \|\boldsymbol{d}\phi\|_{\mathbb{C}^{m\times n}}^2 d\mu\right)^{1/2}=\|\boldsymbol{d}\phi\|_{H^m}$ we have
\begin{align*}
\|\phi(w)\|_{\mathbb{C}^{m\times n}}\leq\|\phi\|_{H^m}+\|\boldsymbol{d}\phi\|_{H^m},\qquad \mu - \text{a.e. } w\in\mathbb{T}.
\end{align*}
Now simply take the $\mu$--essential supremum on the left hand side \cite[p.\ 66]{xrudin}.
\end{proof}

We will also need a notion of differentiation of operators that interacts well with $\boldsymbol{d}$. For an arbitrary $\boldsymbol{A}:H^m\rightarrow H^l$, define $\boldsymbol{dA}$ to be the operator limit of $\frac{1}{\omega}(\boldsymbol{S}_\omega\boldsymbol{A}\boldsymbol{S}_{-\omega}-\boldsymbol{A})$ as $\omega\rightarrow0$ if it exists. It is easily shown that for $\phi\in H^m$ and $ \boldsymbol{B}:H^n\rightarrow H^m$,
\begin{align*}
\boldsymbol{d}(\boldsymbol{A}\phi)=\boldsymbol{dA}\phi+\boldsymbol{A}\boldsymbol{d}\phi, & &\boldsymbol{d}(\boldsymbol{AB})=\boldsymbol{dAB}+\boldsymbol{AdB},
\end{align*}
whenever $\boldsymbol{d}\phi$, $\boldsymbol{dA}$, and $\boldsymbol{dB}$ exist.

If $M:\mathbb{C}\rightarrow\mathbb{C}^{m\times n}$ is analytic in a neighborhood of $\mathbb{T}$, $\boldsymbol{dM}$ takes a particularly simple form.

\begin{lem}\label{lem:dM}
Let $M:\mathbb{C}\rightarrow\mathbb{C}^{m\times n}$ be analytic in neighborhood of $\mathbb{T}$ and let $M'$ be its derivative, then $\boldsymbol{dM}:\phi\mapsto \boldsymbol{P}\left(\mathrm{i}zM'\phi\big|H_0^m\right)$.
\end{lem}
\begin{proof}
Since $\mathrm{i}zM'$ is also analytic in a neighborhood of $\mathbb{T}$ \citep[Corollary 10.16]{xrudin}, its restriction to $\mathbb{T}$ is in $\mathcal{W}^{m\times n}$ and $\mathrm{i}\boldsymbol{zM'}$ is well-defined on $H_0^n$. Next, since $\boldsymbol{S}_\omega$ is a unitary operator on $H_0^m$,
\begin{align*}
\frac{1}{\omega}\left(\boldsymbol{S}_\omega \boldsymbol{M}\boldsymbol{S}_{-\omega}-\boldsymbol{M}\right)\phi&=\frac{1}{\omega}\left(
\boldsymbol{P}\left(\boldsymbol{S}_\omega (M\boldsymbol{S}_{-\omega}\phi)|\boldsymbol{S}_\omega H_0^m\right)-\boldsymbol{P}\left(M\phi|H_0^m\right)\right)\\
&=\frac{1}{\omega}\left(
\boldsymbol{P}\left((\boldsymbol{S}_\omega M)\phi|H_0^m\right)-\boldsymbol{P}\left(M\phi|H_0^m\right)\right)\\
&=
\boldsymbol{P}\left((\boldsymbol{\Delta}_\omega M)\phi|H_0^m\right),
\end{align*}
where we have used the fact that $\boldsymbol{S}_\omega H_0^m=H_0^m$ and the fact that $\boldsymbol{S}_\omega(M\boldsymbol{S}_{-\omega}\phi)=(\boldsymbol{S}_\omega M)\phi$ for all $\phi\in H^n$. Since $\boldsymbol{\Delta}_\omega M-\mathrm{i}zM'$ restricted to $\mathbb{T}$ is in $\mathcal{W}^{m\times n}$, the discussion following Definition \ref{defn:MV} implies that $\frac{1}{\omega}\left(\boldsymbol{S}_\omega \boldsymbol{M}\boldsymbol{S}_{-\omega}-\boldsymbol{M}\right)-\mathrm{i}\boldsymbol{zM'}$ is bounded in the operator norm by
\begin{align*}
\left\|
\boldsymbol{\Delta}_\omega M-\mathrm{i}zM'
\right\|_\infty&=\sup_{z\in\mathbb{T}} \left\|\frac{1}{\omega}\int_0^\omega\mathrm{i}ze^{\mathrm{i}\lambda}M'(ze^{\mathrm{i}\lambda})-\mathrm{i}zM'(z)d\lambda\right\|_{\mathbb{C}^{m\times n}}\\
&\leq\sup_{z\in\mathbb{T}}\sup_{0\leq\lambda\leq\omega}\left\|\mathrm{i}ze^{\mathrm{i}\lambda}M'(ze^{\mathrm{i}\lambda})-\mathrm{i}zM'(z)\right\|_{\mathbb{C}^{m\times n}}.
\end{align*}
This converges to zero as $\omega\rightarrow0$ by the uniform continuity of $\mathrm{i}zM'$ on $\mathbb{T}$.
\end{proof}

The final lemma consists of technical results, more general versions of which are due to \cite{locker} and \cite{callongroetsch}. Our proofs are specialized and modified so that they follow from first principles.

\begin{lem}\label{lem:altH}
Let $M\in\mathcal{W}^{m\times m}$, let $\det(M(z))\neq0$ for all $z\in\mathbb{T}$, and suppose the partial indices of $M$ are non-negative. Let $\boldsymbol{L}:H_0^m\rightarrow H^l$ be a bounded linear operator, let $\ker(\boldsymbol{M})\cap\ker(\boldsymbol{L})=\{0\}$, and let $\boldsymbol{W}=\boldsymbol{M}^\ast\boldsymbol{M}+\boldsymbol{L}^\ast\boldsymbol{L}$. Then the following holds:
\begin{enumerate}
\item For $\phi,\psi\in H_0^m$, the inner product
\begin{align*}
[[\phi,\psi]]=((\boldsymbol{M}\phi,\boldsymbol{M}\psi))+((\boldsymbol{L}\phi,\boldsymbol{L}\psi))
\end{align*}
defines a Hilbert space $\mathsf{H}_0^m$ and we write $\|\phi\|_{\mathsf{H}_0^m}^2=[[\phi,\phi]]$.
\item $H_0^m$ and $\mathsf{H}_0^m$ have equivalent norms.
\item If $\boldsymbol{A}:H_0^m\rightarrow H_0^m$, then its $\mathsf{H}_0^m$ adjoint is given by
\begin{align*}
\boldsymbol{A}^\times=\boldsymbol{W}^{-1}\boldsymbol{A}^\ast\boldsymbol{W},
\end{align*}
\item The $\mathsf{H}_0^m$ Moore-Penrose inverse of $\boldsymbol{M}$ is,
\begin{align*}
\boldsymbol{M}^-=(\boldsymbol{I}-(\boldsymbol{L}|_{\ker(\boldsymbol{M})})^\dag\boldsymbol{L})\boldsymbol{M}^\dag
\end{align*}
\end{enumerate}
\end{lem}
\begin{proof}
(i) Clearly, $[[\,\cdot\,,\,\cdot\,]]$ is an inner product and it remains to show that $\mathsf{H}_0^m$ is complete. Let $\{\phi_i: i=1,2,\ldots\,\}\subset H_0^m$ be an $\mathsf{H}_0^m$ Cauchy sequence. Then both $\boldsymbol{M}\phi_i$ and $\boldsymbol{L}\phi_i$ are Cauchy in $H^m$ and $H^l$ respectively. They must therefore have limits $\varphi_{\boldsymbol{M}}\in H^m$ and $\varphi_{\boldsymbol{L}}\in H^l$ respectively. Now write 
\begin{align*}
\phi_i=\phi_{i,\ker(\boldsymbol{M})}+\phi_{i,\ker(\boldsymbol{M})^\perp},\qquad i=1,2,\ldots,
\end{align*}
where $\phi_{i,\ker(\boldsymbol{M})}\in \ker(\boldsymbol{M})$ and $\phi_{i,\ker(\boldsymbol{M})^\perp}\in \ker(\boldsymbol{M})^\perp$, the orthogonal complement to $\ker(\boldsymbol{M})$ in $H_0^m$. Then
\begin{align*}
\boldsymbol{M}\phi_i=\boldsymbol{M}\phi_{i,\ker(\boldsymbol{M})^\perp}\rightarrow \varphi_{\boldsymbol{M}}\text{ in }H_0^m.
\end{align*}
We have already established in Lemma \ref{lem:mp} that $\boldsymbol{M}^\dag$ is a bounded linear operator. Since $\boldsymbol{M}^\dag\boldsymbol{M}$ is the orthogonal projection onto $\ker(\boldsymbol{M})^\perp$ in $H_0^m$ \citep[p.\ 47]{groetsch},
\begin{align*}
\phi_{i,\ker(\boldsymbol{M})^\perp}=\boldsymbol{M}^\dag\boldsymbol{M}\phi_{i,\ker(\boldsymbol{M})^\perp}\rightarrow \boldsymbol{M}^\dag\varphi_{\boldsymbol{M}}\text{ in }H_0^m.
\end{align*}
It follows, since $\boldsymbol{L}$ is a bounded linear operator, that
\begin{align*}
\boldsymbol{L}\phi_{i,\ker(\boldsymbol{M})^\perp}\rightarrow\boldsymbol{L}\boldsymbol{M}^\dag\varphi_{\boldsymbol{M}}\text{ in }H^l.
\end{align*}
Therefore,
\begin{align*}
\boldsymbol{L}|_{\ker(\boldsymbol{M})}\phi_{i,\ker(\boldsymbol{M})}=\boldsymbol{L}\phi_{i,\ker(\boldsymbol{M})}\rightarrow\varphi_{\boldsymbol{L}}-\boldsymbol{L}\boldsymbol{M}^\dag\varphi_{\boldsymbol{M}}\text{ in }H^l.
\end{align*}
It has already been established in Lemma \ref{lem:existencephi} that $\dim(\ker(\boldsymbol{M}))<\infty$, thus $\boldsymbol{L}|_{\ker(\boldsymbol{M})}$ is of finite rank and its image is closed. It follows that $(\boldsymbol{L}|_{\ker(\boldsymbol{M})})^\dag$ exists and is a bounded linear operator \cite[Corollary 2.1.3]{groetsch}. Moreover,
$\ker(\boldsymbol{M})\cap\ker(\boldsymbol{L})=\{0\}$ implies that $\boldsymbol{L}|_{\ker(\boldsymbol{M})}$ is injective, therefore the image of $(\boldsymbol{L}|_{\ker(\boldsymbol{M})})^\ast$ is $\ker(\boldsymbol{M})$. By Theorem 2.1.2 of \cite{groetsch}, the image of $(\boldsymbol{L}|_{\ker(\boldsymbol{M})})^\ast$ is the image of $(\boldsymbol{L}|_{\ker(\boldsymbol{M})})^\dag$. Thus, $(\boldsymbol{L}|_{\ker(\boldsymbol{M})})^\dag (\boldsymbol{L}|_{\ker(\boldsymbol{M})})$ is the orthogonal projection onto $\ker(\boldsymbol{M})$ in $H_0^m$ and so,
\begin{align*}
\phi_{i,\ker(\boldsymbol{M})}&=(\boldsymbol{L}|_{\ker(\boldsymbol{M})})^\dag (\boldsymbol{L}|_{\ker(\boldsymbol{M})})\phi_{i,\ker(\boldsymbol{M})}\rightarrow (\boldsymbol{L}|_{\ker(\boldsymbol{M})})^\dag\left(\varphi_{\boldsymbol{L}}-\boldsymbol{L}\boldsymbol{M}^\dag\varphi_{\boldsymbol{M}}\right)\text{ in }H_0^m.
\end{align*}
Therefore $\phi_i$ converges in $H_0^m$ to a point, call it $\phi_0$. Finally,
\begin{align*}
\|\phi_i-\phi_0\|_{\mathsf{H}_0^m}^2=\|\boldsymbol{M}(\phi_i-\phi_0)\|_{H^m}^2+\|\boldsymbol{L}(\phi_i-\phi_0)\|_{H^l}^2
\end{align*}
and the boundedness of $\boldsymbol{M}$ and $\boldsymbol{L}$ imply that $\phi_i$ converges to $\phi_0$ in $\mathsf{H}_0^m$ as well.

(ii) Since $\boldsymbol{M}$ and $\boldsymbol{L}$ are bounded linear operators, there exists an upper bound $c>0$ on their operator norms and
\begin{align*}
\|\phi\|_{\mathsf{H}_0^m}=\left(\|\boldsymbol{M}\phi\|_{H^m}^2+\|\boldsymbol{L}\phi\|_{H^l}^2\right)^{1/2}\leq \sqrt{2}c\|\phi\|_{H^m},\qquad \phi\in \mathsf{H}_0^m.
\end{align*}
The equivalence of $\|\,\cdot\,\|_{H^m}$ and $\|\,\cdot\,\|_{\mathsf{H}_0^m}$ on $H_0^m$ then follows from Corollary XII.4.2 of \cite{basic}.

(iii) For $\phi,\psi\in H_0^m$,
\begin{align*}
[[\phi,\psi]]=((\boldsymbol{W}\phi,\psi)).
\end{align*}
This implies that
\begin{align*}
((\boldsymbol{W}\boldsymbol{A}^\times\phi,\psi))=[[\boldsymbol{A}^\times\phi,\psi]]=[[\phi,\boldsymbol{A}\psi]]=((\boldsymbol{W}\phi,\boldsymbol{A}\psi)).
\end{align*}
Since the last term is equal to $((\boldsymbol{A}^\ast\boldsymbol{W}\phi,\psi))$ and $\phi$ and $\psi$ are arbitrary,
\begin{align*}
\boldsymbol{W}\boldsymbol{A}^\times=\boldsymbol{A}^\ast\boldsymbol{W}.
\end{align*}
If $\boldsymbol{W}\phi=0$, then $0=((\boldsymbol{W}\phi,\phi))=\|\phi\|_{\mathsf{H}_0^m}^2$. This implies that $\boldsymbol{W}$ is injective. Next, if $\phi$ is orthogonal to the image of $\boldsymbol{W}$, then $((\phi,\boldsymbol{W}\varphi))=0$ for all $\varphi\in H_0^m$. Since $\boldsymbol{W}$ is self-adjoint as an operator on $H_0^m$, we have that $((\boldsymbol{W}\phi,\varphi))=0$ for all $\varphi\in H_0^m$, in particular $0=((\boldsymbol{W}\phi,\phi))=\|\phi\|_{\mathsf{H}_0^m}^2$. Thus, $\boldsymbol{W}$ is surjective. It follows that $\boldsymbol{W}$ is invertible \citep[p.\ 283]{basic} and the expression for $\boldsymbol{A}^\times$ follows.

(iv) By the arguments used in Lemma \ref{lem:mp}, we have the following representation
\begin{align*}
\boldsymbol{M}^-&=\boldsymbol{M}^\times(\boldsymbol{M}\boldsymbol{M}^\times)^{-1}\\
&=\boldsymbol{W}^{-1}\boldsymbol{M}^\ast\left(\boldsymbol{M}\boldsymbol{W}^{-1}\boldsymbol{M}^\ast\right)^{-1}\\
&=\boldsymbol{W}^{-1}\boldsymbol{M}^\ast\left(\boldsymbol{M}\boldsymbol{W}^{-1}\boldsymbol{M}^\ast\right)^{-1}\boldsymbol{M}\boldsymbol{M}^\dag\\
&=\boldsymbol{W}^{-1/2}\left(\boldsymbol{W}^{-1/2}\boldsymbol{M}^\ast\left(\boldsymbol{M}\boldsymbol{W}^{-1}\boldsymbol{M}^\ast\right)^{-1}\boldsymbol{M}\boldsymbol{W}^{-1/2}\right)\boldsymbol{W}^{1/2}\boldsymbol{M}^\dag.
\end{align*}
The second equality follows from (iii), the third follows from Lemma \ref{lem:mp}, and the fourth follows from Theorem V.6.1 of \cite{classes1}. Next, notice that the operator 
\begin{align*}
\boldsymbol{W}^{-1/2}\boldsymbol{M}^\ast\left(\boldsymbol{M}\boldsymbol{W}^{-1}\boldsymbol{M}^\ast\right)^{-1}\boldsymbol{M}\boldsymbol{W}^{-1/2}
\end{align*}
is a self-adjoint projection acting on $H_0^m$. By Theorem II.13.1 of \cite{basic}, it is the $H_0^m$ orthogonal projection onto its image and it is easily seen that this is the image of $\boldsymbol{W}^{-1/2}\boldsymbol{M}^\ast$. Let
\begin{align*}
\boldsymbol{\Pi}=\boldsymbol{I}-\boldsymbol{W}^{-1/2}\boldsymbol{M}^\ast\left(\boldsymbol{M}\boldsymbol{W}^{-1}\boldsymbol{M}^\ast\right)^{-1}\boldsymbol{M}\boldsymbol{W}^{-1/2}.
\end{align*}
Then $\boldsymbol{\Pi}$ is the $H_0^m$ orthogonal projection onto $\boldsymbol{W}^{1/2}\ker(\boldsymbol{M})$, the orthogonal complement to the image of $\boldsymbol{W}^{-1/2}\boldsymbol{M}^\ast$ in $H_0^m$. Since $\ker(\boldsymbol{M})$ is the image of $(\boldsymbol{L}|_{\ker(\boldsymbol{M})})^\dag$, $\boldsymbol{\Pi}$ is the $H_0^m$ orthogonal projection onto the image of $\boldsymbol{W}^{1/2}(\boldsymbol{L}|_{\ker(\boldsymbol{M})})^\dag$. We then have that
\begin{align*}
\boldsymbol{\Pi}&=\left(\boldsymbol{W}^{1/2}(\boldsymbol{L}|_{\ker(\boldsymbol{M})})^\dag\right)\left(\boldsymbol{W}^{1/2}(\boldsymbol{L}|_{\ker(\boldsymbol{M})})^\dag\right)^\dag\\
&=\left(\boldsymbol{W}^{1/2}(\boldsymbol{L}|_{\ker(\boldsymbol{M})})^\dag\right)\left\{\left(\boldsymbol{W}^{1/2}(\boldsymbol{L}|_{\ker(\boldsymbol{M})})^\dag\right)^\ast\left(\boldsymbol{W}^{1/2}(\boldsymbol{L}|_{\ker(\boldsymbol{M})})^\dag\right)\right\}^\dag \left(\boldsymbol{W}^{1/2}(\boldsymbol{L}|_{\ker(\boldsymbol{M})})^\dag\right)^\ast\\
&=\boldsymbol{W}^{1/2}(\boldsymbol{L}|_{\ker(\boldsymbol{M})})^\dag\left\{(\boldsymbol{L}|_{\ker(\boldsymbol{M})}^\ast)^\dag\boldsymbol{W}(\boldsymbol{L}|_{\ker(\boldsymbol{M})})^\dag\right\}^\dag (\boldsymbol{L}|_{\ker(\boldsymbol{M})}^\ast)^\dag\boldsymbol{W}^{1/2}\\
&=\boldsymbol{W}^{1/2}(\boldsymbol{L}|_{\ker(\boldsymbol{M})})^\dag\left\{(\boldsymbol{L}|_{\ker(\boldsymbol{M})}^\ast)^\dag\boldsymbol{L}^\ast\boldsymbol{L}(\boldsymbol{L}|_{\ker(\boldsymbol{M})})^\dag\right\}^\dag (\boldsymbol{L}|_{\ker(\boldsymbol{M})}^\ast)^\dag\boldsymbol{W}^{1/2}\\
&=\boldsymbol{W}^{1/2}(\boldsymbol{L}|_{\ker(\boldsymbol{M})})^\dag\left\{(\boldsymbol{L}|_{\ker(\boldsymbol{M})}^\ast)^\dag(\boldsymbol{L}|_{\ker(\boldsymbol{M})}^\ast)(\boldsymbol{L}|_{\ker(\boldsymbol{M})})(\boldsymbol{L}|_{\ker(\boldsymbol{M})})^\dag\right\}^\dag (\boldsymbol{L}|_{\ker(\boldsymbol{M})}^\ast)^\dag\boldsymbol{W}^{1/2}\\
&=\boldsymbol{W}^{1/2}(\boldsymbol{L}|_{\ker(\boldsymbol{M})})^\dag\left\{\left((\boldsymbol{L}|_{\ker(\boldsymbol{M})})(\boldsymbol{L}|_{\ker(\boldsymbol{M})})^\dag\right)^\ast\left((\boldsymbol{L}|_{\ker(\boldsymbol{M})})(\boldsymbol{L}|_{\ker(\boldsymbol{M})})^\dag\right)\right\}^\dag (\boldsymbol{L}|_{\ker(\boldsymbol{M})}^\ast)^\dag\boldsymbol{W}^{1/2}\\
&=\boldsymbol{W}^{1/2}(\boldsymbol{L}|_{\ker(\boldsymbol{M})})^\dag\left\{\left((\boldsymbol{L}|_{\ker(\boldsymbol{M})})(\boldsymbol{L}|_{\ker(\boldsymbol{M})})^\dag\right)\left((\boldsymbol{L}|_{\ker(\boldsymbol{M})})(\boldsymbol{L}|_{\ker(\boldsymbol{M})})^\dag\right)\right\}^\dag (\boldsymbol{L}|_{\ker(\boldsymbol{M})}^\ast)^\dag\boldsymbol{W}^{1/2}\\
&=\boldsymbol{W}^{1/2}(\boldsymbol{L}|_{\ker(\boldsymbol{M})})^\dag\left\{(\boldsymbol{L}|_{\ker(\boldsymbol{M})})(\boldsymbol{L}|_{\ker(\boldsymbol{M})})^\dag\right\}^\dag (\boldsymbol{L}|_{\ker(\boldsymbol{M})}^\ast)^\dag\boldsymbol{W}^{1/2}\\
&=\boldsymbol{W}^{1/2}(\boldsymbol{L}|_{\ker(\boldsymbol{M})})^\dag(\boldsymbol{L}|_{\ker(\boldsymbol{M})})(\boldsymbol{L}|_{\ker(\boldsymbol{M})})^\dag (\boldsymbol{L}|_{\ker(\boldsymbol{M})}^\ast)^\dag\boldsymbol{W}^{1/2}\\
&=\boldsymbol{W}^{1/2}(\boldsymbol{L}|_{\ker(\boldsymbol{M})})^\dag (\boldsymbol{L}|_{\ker(\boldsymbol{M})}^\ast)^\dag\boldsymbol{W}^{1/2}\\
&=\boldsymbol{W}^{1/2}(\boldsymbol{L}|_{\ker(\boldsymbol{M})})^\dag (\boldsymbol{L}|_{\ker(\boldsymbol{M})}^\ast)^\dag\boldsymbol{W}\boldsymbol{W}^{-1/2}\\
&=\boldsymbol{W}^{1/2}(\boldsymbol{L}|_{\ker(\boldsymbol{M})})^\dag (\boldsymbol{L}|_{\ker(\boldsymbol{M})}^\ast)^\dag\boldsymbol{L}|_{\ker(\boldsymbol{M})}^\ast\boldsymbol{L}\boldsymbol{W}^{-1/2}\\
&=\boldsymbol{W}^{1/2}(\boldsymbol{L}|_{\ker(\boldsymbol{M})})^\dag \left(\boldsymbol{L}|_{\ker(\boldsymbol{M})}(\boldsymbol{L}|_{\ker(\boldsymbol{M})})^\dag\right)^\ast\boldsymbol{L}\boldsymbol{W}^{-1/2}\\
&=\boldsymbol{W}^{1/2}(\boldsymbol{L}|_{\ker(\boldsymbol{M})})^\dag \left(\boldsymbol{L}|_{\ker(\boldsymbol{M})}(\boldsymbol{L}|_{\ker(\boldsymbol{M})})^\dag\right)\boldsymbol{L}\boldsymbol{W}^{-1/2}\\
&=\boldsymbol{W}^{1/2}(\boldsymbol{L}|_{\ker(\boldsymbol{M})})^\dag\boldsymbol{L}\boldsymbol{W}^{-1/2},
\end{align*}
where we have used basic properties of the Moore-Penrose inverse \citep[Sections 2.1-2.2]{groetsch} as well as the fact that $(\boldsymbol{L}|_{\ker(\boldsymbol{M})})^\dag$ and $(\boldsymbol{L}|_{\ker(\boldsymbol{M})})^\ast$ map into $\ker(\boldsymbol{M})$. It follows that
\begin{align*}
\boldsymbol{M}^-&=\boldsymbol{W}^{-1/2}\left(\boldsymbol{I}-\boldsymbol{\Pi}\right)\boldsymbol{W}^{1/2}\boldsymbol{M}^\dag\\
&=(\boldsymbol{I}-(\boldsymbol{L}|_{\ker(\boldsymbol{M})})^\dag\boldsymbol{L})\boldsymbol{M}^\dag.\qedhere
\end{align*}
\end{proof}

Lemma \ref{lem:altH} (i) introduces a Hilbert space $\mathsf{H}_0^m$ that plays an important role in our regularization theory; $\mathsf{H}_0^m$ and $H_0^m$ have the same elements but different inner products. Lemma \ref{lem:altH} (ii) implies that convergence of points in (operators on) $H_0^m$ is equivalent to convergence of points in (operators on) $\mathsf{H}_0^m$. In particular, $\boldsymbol{d}$, whether acting on points or operators, takes the same value in both spaces. On the other hand, adjoints and orthogonal projections are not the same in both spaces due to the different inner products. This is what is proven in Lemma \ref{lem:altH} (iii) and (iv).

\begin{proof}[Proof of Lemma \ref{lem:reg}]
The proof that $\phi_{\boldsymbol{L}}$ solves the regularization problem is due to \cite{callongroetsch}. We provide an alternative direct derivation. By Lemmas \ref{lem:existencephi} and \ref{lem:mp},
\begin{align*}
\min\left\{\|\boldsymbol{L}\phi\|^2_{H^l}: \boldsymbol{M}\phi=\boldsymbol{N}I_n\right\}&=\min\left\{\|\boldsymbol{L}\phi\|^2_{H^l}: \phi\in\boldsymbol{M}^\dag\boldsymbol{N}I_n+\ker(\boldsymbol{M})\right\}\\
&=\min\left\{\|\boldsymbol{L}(\boldsymbol{M}^\dag\boldsymbol{N}I_n+\chi)\|^2_{H^l}: \chi\in\ker(\boldsymbol{M})\right\}\\
&=\min\left\{\|\boldsymbol{L}\boldsymbol{M}^\dag\boldsymbol{N}I_n+\boldsymbol{L}|_{\ker(\boldsymbol{M})}\chi\|^2_{H^l}: \chi\in\ker(\boldsymbol{M})\right\}.
\end{align*}
By Lemma \ref{lem:existencephi}, $\ker(\boldsymbol{M})$ is of finite dimension, therefore the image of $\boldsymbol{L}|_{\ker(\boldsymbol{M})}$ is finite dimensional and closed. Thus $(\boldsymbol{L}|_{\ker(\boldsymbol{M})})^\dag$ exists and is a bounded linear operator \cite[Corollary 2.1.3]{groetsch}. The minimum above is therefore attained for $\chi=-(\boldsymbol{L}|_{\ker(\boldsymbol{M})})^\dag \boldsymbol{L}\boldsymbol{M}^\dag \boldsymbol{N}I_n+\ker(\boldsymbol{L}|_{\ker(\boldsymbol{M})})$ \cite[p.\ 41]{groetsch}. The uniqueness result then follows from the fact that $\ker(\boldsymbol{L}|_{\ker(\boldsymbol{M})})=\ker(\boldsymbol{L})\cap\ker(\boldsymbol{M})$.
\end{proof}

\begin{proof}[Proof of Theorem \ref{thm:tykcont}]
Fix $\theta_0\in\Theta$. Condition (i) together with Lemma \ref{lem:ga} imply that
\begin{align*}
\|\phi_{\boldsymbol{L}}(\theta)-\phi_{\boldsymbol{L}}(\theta_0)\|_\infty\leq \|\phi_{\boldsymbol{L}}(\theta)-\phi_{\boldsymbol{L}}(\theta_0)\|_{H^m}+\|\boldsymbol{d}\phi_{\boldsymbol{L}}(\theta)-\boldsymbol{d}\phi_{\boldsymbol{L}}(\theta_0)\|_{H^m},
\end{align*}
provided $\phi_{\boldsymbol{L}}(\theta)$, $\phi_{\boldsymbol{L}}(\theta_0)$, $\boldsymbol{d}\phi_{\boldsymbol{L}}(\theta)$, and $\boldsymbol{d}\phi_{\boldsymbol{L}}(\theta_0)$ are in $H_0^m$. We will prove that this is indeed the case and the right hand side converges to zero as $\theta\rightarrow\theta_0$. In this proof, all operator norms and orthogonal projections are understood to be with respect to $H_0^m$ and not with respect to $\mathsf{H}_0^m$. Adjoints and Moore-Penrose inverses with respect to $\mathsf{H}_0^m$ will be denoted by the separate notation used in Lemma \ref{lem:altH}.

\bigskip

\noindent\emph{STEP 1.} $\phi_{\boldsymbol{L}}(\theta)$ is well-defined for every $\theta\in\Theta$ and $\lim_{\theta\rightarrow\theta_0}\|\phi_{\boldsymbol{L}}(\theta)-\phi_{\boldsymbol{L}}(\theta_0)\|_{H^m}=0$.

\bigskip

By Lemma \ref{lem:reg}, conditions (iii), (v), and (vi) imply that $\phi_{\boldsymbol{L}}(\theta)$ is uniquely defined for any $\theta\in\Theta$, including $\theta_0$.

By the discussion following Definition \ref{defn:w} and condition (iii), the operators $\boldsymbol{M}(\theta)$ and $\boldsymbol{M}(\theta_0)$ are well-defined. By the discussion following Definition \ref{defn:MV}, the operator norm of $\boldsymbol{M}(\theta)-\boldsymbol{M}(\theta_0)$ is bounded above by $\|M(\,\cdot\,,\theta)-M(\,\cdot\,,\theta_0)\|_\infty$. By condition (iv), $M(z,\theta)$ is jointly continuous and so $\|M(\,\cdot\,,\theta)-M(\,\cdot\,,\theta_0)\|_\infty$ is continuous at $\theta_0$ \cite[Theorem 9.14]{sundaram}. It follows that $\lim_{\theta\rightarrow\theta_0}\|M(\,\cdot\,,\theta)-M(\,\cdot\,,\theta_0)\|_\infty=0$ and so $\boldsymbol{M}(\theta)$ converges to $\boldsymbol{M}(\theta_0)$ in the operator norm as $\theta\rightarrow\theta_0$. The same argument applied to $\boldsymbol{N}(\theta)$ proves that $\boldsymbol{N}(\theta)$ converges to $\boldsymbol{N}(\theta_0)$ in the operator norm as $\theta\rightarrow\theta_0$.

Lemma \ref{lem:existencephi} and condition (vi) imply that $\boldsymbol{M}(\theta)$ is onto for any $\theta\in\Theta$ so the orthogonal projection operator onto the image of $\boldsymbol{M}(\theta)$ is the identity mapping on $H_0^m$ for any $\theta\in\Theta$. Thus, $\boldsymbol{M}(\theta)\boldsymbol{M}(\theta)^\dag=\boldsymbol{M}(\theta_0)\boldsymbol{M}(\theta_0)^\dag$ and Theorem 1.6 of \cite{koliha} implies that $\boldsymbol{M}(\theta)^\dag$ converges in the operator norm to $\boldsymbol{M}(\theta_0)^\dag$ as $\theta\rightarrow\theta_0$.

Next, for $\theta\in\Theta$, let $\boldsymbol{Q}(\theta)$ be the orthogonal projection onto $\ker(\boldsymbol{M}(\theta))$. By verifying the four conditions that determine the Moore-Penrose inverse \cite[p.\ 48]{groetsch}, it is easily seen that $(\boldsymbol{L}|_{\ker(\boldsymbol{M}(\theta)})^\dag=(\boldsymbol{L}\boldsymbol{Q}(\theta))^\dag$. By Lemma \ref{lem:existencephi} and condition (vi), $\boldsymbol{L}\boldsymbol{Q}(\theta)$ is an operator of finite rank. Condition (v) now implies that the rank of $\boldsymbol{L}\boldsymbol{Q}(\theta)$ is equal to $\dim(\ker(\boldsymbol{M}(\theta)))$. Since $\boldsymbol{Q}(\theta)=\boldsymbol{I}-\boldsymbol{M}(\theta)^\dag\boldsymbol{M}(\theta)$ converges in the operator norm to $\boldsymbol{Q}(\theta_0)=\boldsymbol{I}-\boldsymbol{M}(\theta_0)^\dag\boldsymbol{M}(\theta_0)$ as $\theta\rightarrow\theta_0$, $\boldsymbol{L}\boldsymbol{Q}(\theta)$ converges to $\boldsymbol{L}\boldsymbol{Q}(\theta_0)$ in operator norm and the smallest non-zero singular value of $\boldsymbol{L}\boldsymbol{Q}(\theta)$ also converges to that of $\boldsymbol{L}\boldsymbol{Q}(\theta_0)$ \cite[Corollary VI.1.6]{classes1}. It follows that the operator norm of $(\boldsymbol{L}\boldsymbol{Q}(\theta))^\dag$ remains bounded as $\theta\rightarrow\theta_0$. By Theorem 1.6 of \cite{koliha} again, $(\boldsymbol{L}\boldsymbol{Q}(\theta))^\dag$ converges in operator norm to $(\boldsymbol{L}\boldsymbol{Q}(\theta_0))^\dag$.

It follows from the above and Lemma \ref{lem:altH} (iv) that $\boldsymbol{M}(\theta)^-=(\boldsymbol{I}-(\boldsymbol{L}|_{\ker(\boldsymbol{M}(\theta)})^\dag\boldsymbol{L})\boldsymbol{M}(\theta)^\dag$ converges in operator norm to $\boldsymbol{M}(\theta_0)^-=(\boldsymbol{I}-(\boldsymbol{L}|_{\ker(\boldsymbol{M}(\theta_0)})^\dag\boldsymbol{L})\boldsymbol{M}(\theta_0)^\dag$ as $\theta\rightarrow\theta_0$. Therefore, $\phi_{\boldsymbol{L}}(\theta)=\boldsymbol{M}(\theta)^-\boldsymbol{N}(\theta)I_n$ converges to $\phi_{\boldsymbol{L}}(\theta_0)=\boldsymbol{M}(\theta_0)^-\boldsymbol{N}(\theta_0)I_n$ in $H_0^m$ as $\theta\rightarrow\theta_0$.

\bigskip

\noindent\emph{STEP 2.} $\boldsymbol{d}\phi_{\boldsymbol{L}}(\theta)$ exists and $\lim_{\theta\rightarrow\theta_0}\|\boldsymbol{d}\phi_{\boldsymbol{L}}(\theta)-\boldsymbol{d}\phi_{\boldsymbol{L}}(\theta_0)\|_{H^m}=0$.

\bigskip

For $\theta\in\Theta$, $\boldsymbol{dM}(\theta)$ exists by Lemma \ref{lem:dM} and condition (iii). It follows that $\boldsymbol{S}_\omega\boldsymbol{M}(\theta)\boldsymbol{S}_{-\omega}$ converges to $\boldsymbol{M}(\theta)$ in the operator norm as $\omega\rightarrow0$. For $\omega\in\mathbb{R}$, $\boldsymbol{S}_\omega\boldsymbol{M}(\theta)\boldsymbol{S}_{-\omega}:\phi\mapsto P(M(ze^{\mathrm{i}\omega},\theta)\phi|H_0^m)$ and $M(ze^{\mathrm{i}\omega},\theta)\in\mathcal{W}^{m\times m}$ has the same partial indices as $M(z,\theta)$. Therefore, condition (vi) implies that $\boldsymbol{S}_\omega\boldsymbol{M}(\theta)\boldsymbol{S}_{-\omega}$ is onto and by the same argument as in step 1, $(\boldsymbol{S}_\omega\boldsymbol{M}(\theta)\boldsymbol{S}_{-\omega})^-$ converges to $\boldsymbol{M}(\theta)^-$ in the operator norm as $\omega\rightarrow0$. Theorem 2.1 of \cite{koliha} applied to $\boldsymbol{S}_\omega\boldsymbol{M}(\theta)\boldsymbol{S}_{-\omega}$, as $\mathsf{H}_0^m\rightarrow\mathsf{H}_0^m$ mappings, gives
\begin{align*}
\boldsymbol{d}(\boldsymbol{M}(\theta)^-)=-\boldsymbol{M}(\theta)^-\boldsymbol{dM}(\theta)\boldsymbol{M}(\theta)^-+(\boldsymbol{I}-\boldsymbol{M}(\theta)^-\boldsymbol{M}(\theta))\boldsymbol{dM}(\theta)^\times(\boldsymbol{M}(\theta)^-)^\times\boldsymbol{M}(\theta)^-,
\end{align*}
By Lemma \ref{lem:dM} and condition (iii) again,
\begin{align*}
\boldsymbol{d}(\boldsymbol{N}(\theta)I_n)=\boldsymbol{dN}(\theta)I_n.
\end{align*}
This implies that for every $\theta\in\Theta$,
\begin{align*}
\boldsymbol{d}\phi_{\boldsymbol{L}}(\theta)&=\boldsymbol{d}(\boldsymbol{M}(\theta)^-)\boldsymbol{N}(\theta)I_n+\boldsymbol{M}(\theta)^-\boldsymbol{d}\boldsymbol{N}(\theta)I_n\\
&=\left(-\boldsymbol{M}(\theta)^-\boldsymbol{dM}(\theta)\boldsymbol{M}(\theta)^-+(\boldsymbol{I}-\boldsymbol{M}(\theta)^-\boldsymbol{M}(\theta))\boldsymbol{dM}(\theta)^\times(\boldsymbol{M}(\theta)^-)^\times\boldsymbol{M}(\theta)^-\right)\boldsymbol{N}(\theta)I_n\\
&\qquad+\boldsymbol{M}(\theta)^-\boldsymbol{d}\boldsymbol{N}(\theta)I_n\\
&=-\boldsymbol{M}(\theta)^-\boldsymbol{dM}(\theta)\boldsymbol{M}(\theta)^-\boldsymbol{N}(\theta)I_n\\
&\qquad+(\boldsymbol{I}-\boldsymbol{M}(\theta)^-\boldsymbol{M}(\theta))\boldsymbol{W}(\theta)^{-1}\boldsymbol{dM}(\theta)^\ast(\boldsymbol{M}(\theta)^-)^\ast\boldsymbol{W}(\theta)\boldsymbol{M}(\theta)^-\boldsymbol{N}(\theta)I_n\\
&\qquad+\boldsymbol{M}(\theta)^-\boldsymbol{d}\boldsymbol{N}(\theta)I_n,
\end{align*}
where $\boldsymbol{W}(\theta)=\boldsymbol{M}(\theta)^\ast\boldsymbol{M}(\theta)+\boldsymbol{L}^\ast\boldsymbol{L}$. Clearly, $\boldsymbol{d}\phi_{\boldsymbol{L}}(\theta)$ converges to $\boldsymbol{d}\phi_{\boldsymbol{L}}(\theta_0)$ in $H_0^m$ as $\theta\rightarrow\theta_0$ if all of the operators that appear in the expression above are continuous in the operator norm at $\theta_0$. The continuity of $\boldsymbol{M}(\theta)$, $\boldsymbol{N}(\theta)$, and $\boldsymbol{M}(\theta)^-$ at $\theta_0$ has already been established in step 1. It remains to establish the operator norm continuity of $\boldsymbol{W}(\theta)^{-1}$, $\boldsymbol{dM}(\theta)$, and $\boldsymbol{dN}(\theta)$ at $\theta_0$. The continuity of inversion at invertible operators ensures that $\boldsymbol{W}(\theta)^{-1}$ is continuous in the operator norm at $\theta_0$ \citep[Corollary II.8.2]{basic}. The operator norm of $\boldsymbol{dM}(\theta)-\boldsymbol{dM}(\theta_0)$ is bounded above by $\sup_{z\in\mathbb{T}}\left\|\mathrm{i}z\frac{d}{dz}M(z,\theta)-\mathrm{i}z\frac{d}{dz}M(z,\theta_0)\right\|_{\mathbb{C}^{m\times m}}=\sup_{z\in\mathbb{T}}\left\|\frac{d}{dz}M(z,\theta)-\frac{d}{dz}M(z,\theta_0)\right\|_{\mathbb{C}^{m\times m}}$, which converges to zero as $\theta\rightarrow\theta_0$ by the joint continuity of $\frac{d}{dz}M(z,\theta)$ (condition (iv)). A similar argument yields the continuity of $\boldsymbol{dN}(\theta)$.
\end{proof}

\begin{proof}[Proof of Theorem \ref{thm:tykdiff}]
For $\phi:\Theta\rightarrow H^m$, $x\in\mathbb{R}^d$, and $\epsilon\neq0$ define
\begin{align*}
\boldsymbol{\nabla}_{\epsilon,x}\phi(\theta)=\frac{\phi(\theta+\epsilon x)-\phi(\theta)}{\epsilon}.
\end{align*}
The claim of the theorem is that for every $\theta\in\Theta$, there exists a symmetric multilinear mapping $D_\theta^p\phi_{\boldsymbol{L}}(\theta):\prod_{i=1}^p\mathbb{R}^d\rightarrow H_0^m$ such that for every $\{x_1,\ldots,x_p\}\subset\mathbb{R}^d$,
\begin{align*}
\lim_{(\epsilon_1,\ldots,\epsilon_p)\rightarrow0}\left\|\boldsymbol{\nabla}_{\epsilon_p,x_p}\cdots\boldsymbol{\nabla}_{\epsilon_1,x_1}\phi_{\boldsymbol{L}}(\theta)-D_\theta^p\phi_{\boldsymbol{L}}(\theta)(x_1,\ldots,x_p)\right\|_\infty=0
\end{align*}
and $D_\theta^p\phi_{\boldsymbol{L}}(\theta)(x_1,\ldots,x_p)$ is continuous in $\mu$--essential supremum norm with respect to $\theta$. As in the proof of Theorem \ref{thm:tykcont}, we will prove the existence of $\boldsymbol{\nabla}_{\epsilon_p,x_p}\cdots\boldsymbol{\nabla}_{\epsilon_1,x_1}\phi_{\boldsymbol{L}}(\theta)$ and $\boldsymbol{d}\left(\boldsymbol{\nabla}_{\epsilon_p,x_p}\cdots\boldsymbol{\nabla}_{\epsilon_1,x_1}\phi_{\boldsymbol{L}}(\theta)\right)$ and their convergence as elements of $H^m$ to $D_\theta^p\phi_{\boldsymbol{L}}(\theta)(x_1,\ldots,x_p)$ and $\boldsymbol{d}(D_\theta^p\phi_{\boldsymbol{L}}(\theta)(x_1,\ldots,x_p))$ respectively as $(\epsilon_1,\ldots,\epsilon_p)\rightarrow0$. The continuity in $H^m$ of\linebreak $D_\theta^p\phi_{\boldsymbol{L}}(\theta)(x_1,\ldots,x_p)$ and $\boldsymbol{d}(D_\theta^p\phi_{\boldsymbol{L}}(\theta)(x_1,\ldots,x_p))$ with respect to $\theta$ then proves the continuity of $D_\theta^p\phi_{\boldsymbol{L}}(\theta)(x_1,\ldots,x_p)$ in $\mu$--essential supremum norm with respect to $\theta$ by Lemma \ref{lem:ga}.

\bigskip

\noindent\emph{STEP 1.} The result holds for $p=1$.

\bigskip

For $x\in\mathbb{R}^d$ and $\phi\in H_0^m$, define
\begin{align*}
\boldsymbol{D}_\theta \boldsymbol{M}(\theta)x:\phi\mapsto \boldsymbol{P}((D_\theta M(z,\theta)x)\phi|H_0^m),
\end{align*}
where $D_\theta M(z,\theta)$ is the Jacobian of $M(z,\theta)$ with respect to $\theta$. Then, by arguments that are by now familiar,
\begin{align*}
\left\|\boldsymbol{\nabla}_{\epsilon,x}\boldsymbol{M}(\theta)\phi-(\boldsymbol{D}_\theta \boldsymbol{M}(\theta)x)\phi\right\|_{H^m}&=\left\|\boldsymbol{P}((\boldsymbol{\nabla}_{\epsilon,x}M(z,\theta)-D_\theta M(z,\theta)x)\phi|H_0^m)\right\|_{H^m}\\
&\leq\left\|(\boldsymbol{\nabla}_{\epsilon,x}M(z,\theta)-D_\theta M(z,\theta)x)\phi\right\|_{H^m}\\
&\leq \left\|\boldsymbol{\nabla}_{\epsilon,x}M(z,\theta)-D_\theta M(z,\theta)x\right\|_\infty\|\phi\|_{H^m}\\
&\leq \left\|\frac{1}{\epsilon}\int_0^\epsilon\left( D_\theta M(z,\theta+\rho x)x-D_\theta M(z,\theta)x\right)d\rho\right\|_\infty\|\phi\|_{H^m}\\
&\leq \sup_{0\leq\rho\leq\epsilon}\left\|D_\theta M(z,\theta+\rho x)x-D_\theta M(z,\theta)x\right\|_\infty\|\phi\|_{H^m},
\end{align*}
which converges to zero as $\epsilon\rightarrow0$ by the uniform continuity of $D_\theta M(z,\theta+\rho x)x$ with respect to $(z,\rho)\in\mathbb{T}\times[0,\bar\epsilon]$, where $\bar\epsilon$ is chosen so that the segment between $\theta$ and $\theta+\bar\epsilon x$ is inside $\Theta$. Thus, the mapping $\phi\mapsto (\boldsymbol{D}_\theta \boldsymbol{M}(\theta)x)\phi$ is linear and bounded on $H_0^m$ \cite[Corollary XII.4.4]{basic}. By the same arguments, the mapping $\boldsymbol{D}_\theta \boldsymbol{N}(\theta)x:\phi\mapsto \boldsymbol{P}((D_\theta N(z,\theta)x)\phi|H_0^m)$ is a bounded linear operator from $H_0^n$ to $H_0^m$ as well as the operator limit of $\boldsymbol{\nabla}_{\epsilon,x}\boldsymbol{N}(\theta)$.

Since $\boldsymbol{D}_\theta \boldsymbol{M}(\theta)x$ exists and, by the same arguments as used above, $\boldsymbol{M}(\theta+\epsilon x)^-$ converges to $\boldsymbol{M}(\theta)^-$ as $\epsilon\rightarrow0$, Theorem 2.1 of \cite{koliha}, applied to $\boldsymbol{M}(\theta+\epsilon x)$ as $\mathsf{H}_0^m\rightarrow\mathsf{H}_0^m$ mappings, implies that $\boldsymbol{\nabla}_{\epsilon,x}
(\boldsymbol{M}(\theta)^-)$ converges in operator norm to
\begin{align*}
\boldsymbol{D}_\theta \boldsymbol{M}(\theta)^- x&=-\boldsymbol{M}(\theta)^- (\boldsymbol{D}_\theta \boldsymbol{M}(\theta)x) \boldsymbol{M}(\theta)^-+(\boldsymbol{I}-\boldsymbol{M}(\theta)^- \boldsymbol{M}(\theta))(\boldsymbol{D}_\theta \boldsymbol{M}(\theta)x)^\times(\boldsymbol{M}(\theta)^-)^\times\boldsymbol{M}(\theta)^-\\
&=-\boldsymbol{M}(\theta)^- (\boldsymbol{D}_\theta \boldsymbol{M}(\theta)x) \boldsymbol{M}(\theta)^-\\
&\qquad +(\boldsymbol{I}-\boldsymbol{M}(\theta)^- \boldsymbol{M}(\theta))\boldsymbol{W}(\theta)^{-1}(\boldsymbol{D}_\theta \boldsymbol{M}(\theta)x)^\ast(\boldsymbol{M}(\theta)^-)^\ast\boldsymbol{W}(\theta)\boldsymbol{M}(\theta)^-,
\end{align*}
as $\epsilon\rightarrow0$, where $\boldsymbol{W}(\theta)=\boldsymbol{M}(\theta)^\ast\boldsymbol{M}(\theta)+\boldsymbol{L}^\ast\boldsymbol{L}$. It follows that the mapping $\phi\mapsto (\boldsymbol{D}_\theta \boldsymbol{M}(\theta)^- x)\phi$ is a bounded linear operator on $H_0^m$. Thus, $\boldsymbol{\nabla}_{\epsilon,x}\phi_{\boldsymbol{L}}(\theta)$ converges in $H^m$ to 
\begin{align*}
\boldsymbol{D}_\theta\phi_{\boldsymbol{L}}(\theta)x=(\boldsymbol{D}_\theta \boldsymbol{M}(\theta)^- x)\boldsymbol{N}(\theta)I_n+\boldsymbol{M}(\theta)^-(\boldsymbol{D}_\theta \boldsymbol{N}(\theta)x)I_n.
\end{align*}

Next, following the same techniques as used above, the joint continuity of $\frac{d}{dz}(D_\theta M(z,\theta)x)=D_\theta\left(\frac{d}{dz}M(z,\theta)\right)x$ implies that the operator $\boldsymbol{d}(\boldsymbol{D}_\theta \boldsymbol{M}(\theta)x)$ exists as a cross operator derivative, is equal to the cross derivative $\boldsymbol{D}_\theta(\boldsymbol{dM}(\theta)x)$, is given by $\phi\mapsto \boldsymbol{P}\left(\mathrm{i}z\frac{d}{dz}(D_\theta M(z,\theta)x)\phi\bigg|H_0^m\right)$, and is a bounded linear operator on $H_0^m$. The same is true of $\boldsymbol{d}(\boldsymbol{D}_\theta \boldsymbol{N}(\theta)x)$. This implies that $\boldsymbol{d}(\boldsymbol{\nabla}_{\epsilon,x}\phi_{\boldsymbol{L}}(\theta))=\boldsymbol{\nabla}_{\epsilon,x}\boldsymbol{d}\phi_{\boldsymbol{L}}(\theta)$ converges in $H^m$ to $\boldsymbol{D}_\theta(\boldsymbol{d}\phi_{\boldsymbol{L}}(\theta))x=\boldsymbol{d}(\boldsymbol{D}_\theta\phi_{\boldsymbol{L}}(\theta)x)$, given by
\begin{align*}
\boldsymbol{d}(\boldsymbol{D}_\theta\phi_{\boldsymbol{L}}(\theta)x)&=\boldsymbol{d}(\boldsymbol{D}_\theta \boldsymbol{M}(\theta)^- x)\boldsymbol{N}(\theta)I_n+(\boldsymbol{D}_\theta \boldsymbol{M}(\theta)^- x)\boldsymbol{d}\boldsymbol{N}(\theta)I_n\\
&\qquad+\boldsymbol{d}(\boldsymbol{M}(\theta)^-)(\boldsymbol{D}_\theta \boldsymbol{N}(\theta)x)I_n+\boldsymbol{M}(\theta)^-\boldsymbol{d}(\boldsymbol{D}_\theta \boldsymbol{N}(\theta)x)I_n\\
&=\big\{-\boldsymbol{d}(\boldsymbol{M}(\theta)^-) (\boldsymbol{D}_\theta \boldsymbol{M}(\theta)x) \boldsymbol{M}(\theta)^--\boldsymbol{M}(\theta)^- \boldsymbol{d}(\boldsymbol{D}_\theta \boldsymbol{M}(\theta)x) \boldsymbol{M}(\theta)^-\\
&\qquad-\boldsymbol{M}(\theta)^- (\boldsymbol{D}_\theta \boldsymbol{M}(\theta)x) \boldsymbol{d}(\boldsymbol{M}(\theta)^-)+\\
&\qquad+(\boldsymbol{I}-\boldsymbol{d}(\boldsymbol{M}(\theta)^-) \boldsymbol{M}(\theta))(\boldsymbol{D}_\theta \boldsymbol{M}(\theta)x)^\times(\boldsymbol{M}(\theta)^-)^\times \boldsymbol{M}(\theta)^-\\
&\qquad+(\boldsymbol{I}-\boldsymbol{M}(\theta)^- d\boldsymbol{M}(\theta))(\boldsymbol{D}_\theta \boldsymbol{M}(\theta)x)^\times(\boldsymbol{M}(\theta)^-)^\times \boldsymbol{M}(\theta)^-\\
&\qquad+(\boldsymbol{I}-\boldsymbol{M}(\theta)^- \boldsymbol{M}(\theta))\boldsymbol{d}(\boldsymbol{D}_\theta \boldsymbol{M}(\theta)x)^\times(\boldsymbol{M}(\theta)^-)^\times \boldsymbol{M}(\theta)^-\\
&\qquad+(\boldsymbol{I}-\boldsymbol{M}(\theta)^- \boldsymbol{M}(\theta))(\boldsymbol{D}_\theta \boldsymbol{M}(\theta)x)^\times \boldsymbol{d}(\boldsymbol{M}(\theta)^-)^\times \boldsymbol{M}(\theta)^-\\
&\qquad+(\boldsymbol{I}-\boldsymbol{M}(\theta)^- \boldsymbol{M}(\theta))(\boldsymbol{D}_\theta \boldsymbol{M}(\theta)x)^\times(\boldsymbol{M}(\theta)^-)^\times \boldsymbol{d}(\boldsymbol{M}(\theta)^-)\big\}\boldsymbol{N}(\theta)I_n+\\
&\qquad +(\boldsymbol{D}_\theta \boldsymbol{M}(\theta)^- x)\boldsymbol{d}\boldsymbol{N}(\theta)I_n+\boldsymbol{d}(\boldsymbol{M}(\theta)^-)(\boldsymbol{D}_\theta \boldsymbol{N}(\theta)x)I_n\\
&\qquad+\boldsymbol{M}(\theta)^-\boldsymbol{d}(\boldsymbol{D}_\theta \boldsymbol{N}(\theta)x)I_n.
\end{align*}
Given the expressions for $\boldsymbol{D}_\theta \boldsymbol{M}(\theta)^- x$, $\boldsymbol{d}(\boldsymbol{D}_\theta\boldsymbol{M}(\theta)x)$, and $\boldsymbol{d}(\boldsymbol{D}_\theta \boldsymbol{N}(\theta)x)$, as well as results obtained in the proof of Theorem \ref{thm:tykcont}, it is now clear that all of the operators that appear in the expressions for $\boldsymbol{D}_\theta\phi_{\boldsymbol{L}}(\theta)x$ and $\boldsymbol{d}(\boldsymbol{D}_\theta\phi_{\boldsymbol{L}}(\theta)x)$ are continuous in operator norm with respect to $\theta$. It follows that $\boldsymbol{D}_\theta\phi_{\boldsymbol{L}}(\theta)x$ and $\boldsymbol{d}(\boldsymbol{D}_\theta\phi_{\boldsymbol{L}}(\theta)x)$ are continuous in $H^m$ norm with respect to $\theta$. By Lemma \ref{lem:ga} then, $\boldsymbol{D}_\theta\phi_{\boldsymbol{L}}(\theta)x$ is continuous in $\mu$--essential supremum norm with respect to $\theta$.

\bigskip

\noindent\emph{STEP 2.} The result holds for all $p>1$.

\bigskip

For $\{x_1,\ldots,x_p\}\subset\mathbb{R}^d$, $\phi\in H_0^m$, and $i=1,\ldots,p$, define
\begin{align*}
\boldsymbol{D}_\theta^i\boldsymbol{M}(\theta)(x_1,\ldots,x_i):\phi\mapsto \boldsymbol{P}(D_\theta^iM(z,\theta)(x_1,\ldots,x_i)\phi|H_0^m),\end{align*}
where $D_\theta^i M(z,\theta)(x_1,\ldots,x_i)=\lim_{\epsilon_i\rightarrow0,\ldots,\epsilon_1\rightarrow0}\boldsymbol{\nabla}_{\epsilon_i,x_i}\cdots\boldsymbol{\nabla}_{\epsilon_1,x_1}M(z,\theta)$ is the $i$-th order derivative of $M(z,\theta)$ with respect to $\theta$ in the directions $x_1,\ldots,x_i$. Just as in step 1, it is easily established that for $\{x_1,\ldots,x_p\}\subset\mathbb{R}^d$ and $i=1,\ldots,p$, $\boldsymbol{D}_\theta^i \boldsymbol{M}(\theta)(x_1,\ldots, x_i)$ is the operator limit of $\boldsymbol{\nabla}_{\epsilon_i,x_i}\cdots\boldsymbol{\nabla}_{\epsilon_1,x_1}\boldsymbol{M}(\theta)$ as $(\epsilon_1,\ldots,\epsilon_i)\rightarrow0$, therefore it is linear and bounded on $H_0^m$. The same result holds for $\boldsymbol{D}_\theta^i \boldsymbol{N}(\theta)(x_1,\ldots, x_i)$.

The expression of $\boldsymbol{D}_\theta \boldsymbol{M}(\theta)^- x$ in step 1 consists of further differentiable operators. Therefore, for $\{x_1,\ldots,x_p\}\subset\mathbb{R}^d$ and $i=1,\ldots,p$, the operator limit of $\boldsymbol{\nabla}_{\epsilon_i,x_i}\cdots\boldsymbol{\nabla}_{\epsilon_1,x_1}\boldsymbol{M}(\theta)^-$ exists, call it $\boldsymbol{D}_\theta^i \boldsymbol{M}(\theta)^-(x_1,\ldots, x_i)$, and is a finite sum of composites of the mappings
\begin{gather*}
\boldsymbol{M}(\theta), \boldsymbol{M}(\theta)^-, \boldsymbol{D}_\theta \boldsymbol{M}(\theta)(x_1),\ldots,\boldsymbol{D}_\theta \boldsymbol{M}(\theta)(x_p), \boldsymbol{D}_\theta^2 \boldsymbol{M}(\theta)(x_1,x_2),\ldots,\boldsymbol{D}_\theta^2 \boldsymbol{M}(\theta)(x_{i-1},x_i),\\
\boldsymbol{D}_\theta^3 \boldsymbol{M}(\theta)(x_1,x_2,x_3),\ldots,\boldsymbol{D}_\theta^3 \boldsymbol{M}(\theta)(x_{i-2},x_{i-1},x_i), \ldots, \boldsymbol{D}_\theta^i\boldsymbol{M}(\theta)(x_1,\ldots, x_i),
\end{gather*}
and their $\mathsf{H}_0^m$ adjoints. As all of these mappings are continuous in the operator norm with respect to $\theta$, it follows that $\boldsymbol{D}_\theta^i\phi_{\boldsymbol{L}}(z,\theta)(x_1,\ldots,x_i)$ is continuous in $H^m$ with respect to $\theta$.

Next, the joint continuity of $\frac{d}{dz}D_\theta^iM(z,\theta)(x_1,\ldots,x_i)=D_\theta^i\left(\frac{d}{dz}M(z,\theta)\right)(x_1,\ldots,x_i)$ for $i=1,\ldots,p$ implies the existence of the operator derivative
\begin{align*}
\boldsymbol{d}(\boldsymbol{D}_\theta^i \boldsymbol{M}(\theta)(x_1,\ldots,x_i)):\phi\mapsto\boldsymbol{P}\left(\mathrm{i}z\frac{d}{dz}D_\theta^i M(z,\theta)(x_1,\ldots,x_i)\phi\bigg|H_0^m\right),
\end{align*}
for $i=1,\ldots,p$, a similar result also holding for $\boldsymbol{d}(\boldsymbol{D}_\theta^i \boldsymbol{N}(\theta)(x_1,\ldots,x_i))$. This implies that
\begin{align*}
\boldsymbol{d}(\boldsymbol{\nabla}_{\epsilon_p,x_p}\cdots \boldsymbol{\nabla}_{\epsilon_1,x_1}\phi_{\boldsymbol{L}}(\theta))=\boldsymbol{\nabla}_{\epsilon_p,x_p}\cdots \boldsymbol{\nabla}_{\epsilon_1,x_1}\boldsymbol{d}\phi_{\boldsymbol{L}}(\theta)
\end{align*}
converges in $H^m$ to 
\begin{align*}
\boldsymbol{d}(\boldsymbol{D}_\theta^p\phi_{\boldsymbol{L}}(\theta)(x_1,\ldots,x_p))=\boldsymbol{D}_\theta^p(\boldsymbol{d}\phi_{\boldsymbol{L}}(\theta))(x_1,\ldots,x_p).
\end{align*}
In particular, $\boldsymbol{d}(\boldsymbol{D}_\theta^p\phi_{\boldsymbol{L}}(z,\theta_0)(x_1,\ldots,x_p))$ is a finite sum of composites of the mappings 
\begin{gather*}
\boldsymbol{M}(\theta), \boldsymbol{M}(\theta)^-, \boldsymbol{D}_\theta \boldsymbol{M}(\theta)(x_1),\ldots,\boldsymbol{D}_\theta \boldsymbol{M}(\theta)(x_p), \boldsymbol{D}_\theta^2 \boldsymbol{M}(\theta)(x_1,x_2),\ldots,\boldsymbol{D}_\theta^2 \boldsymbol{M}(\theta)(x_{p-1},x_p),\\
\boldsymbol{D}_\theta^3 \boldsymbol{M}(\theta)(x_1,x_2,x_3),\ldots,\boldsymbol{D}_\theta^3 \boldsymbol{M}(\theta)(x_{p-2},x_{p-1},x_p),\ldots, \boldsymbol{D}_\theta^p\boldsymbol{M}(\theta)(x_1,\ldots, x_p),\\
\boldsymbol{d}\boldsymbol{M}(\theta), \boldsymbol{d}(\boldsymbol{D}_\theta \boldsymbol{M}(\theta)(x_1)),\ldots, \boldsymbol{d}(\boldsymbol{D}_\theta \boldsymbol{M}(\theta)(x_p)), \boldsymbol{d}(\boldsymbol{D}_\theta^2 \boldsymbol{M}(\theta)(x_1,x_2)), \ldots, \boldsymbol{d}(\boldsymbol{D}_\theta^2 \boldsymbol{M}(\theta)(x_{p-1},x_p)),\\
\boldsymbol{d}(\boldsymbol{D}_\theta^3 \boldsymbol{M}(\theta)(x_1,x_2,x_3)),\ldots,\boldsymbol{d}(\boldsymbol{D}_\theta^3 \boldsymbol{M}(\theta)(x_{p-2},x_{p-1},x_p)),\ldots, \boldsymbol{d}(\boldsymbol{D}_\theta^p\boldsymbol{M}(\theta)(x_1,\ldots, x_p)),\\
\boldsymbol{N}(\theta), \boldsymbol{D}_\theta \boldsymbol{N}(\theta)(x_1), \ldots, \boldsymbol{D}_\theta \boldsymbol{N}(\theta)(x_p), \boldsymbol{D}_\theta^2 \boldsymbol{N}(\theta)(x_1,x_2), \ldots, \boldsymbol{D}_\theta^2 \boldsymbol{N}(\theta)(x_{p-1},x_p),\\ 
\boldsymbol{d}(\boldsymbol{D}_\theta^3 \boldsymbol{N}(\theta)(x_1,x_2,x_3)),\ldots,\boldsymbol{d}(\boldsymbol{D}_\theta^3 \boldsymbol{M}(\theta)(x_{p-2},x_{p-1},x_i)),\ldots, \boldsymbol{D}_\theta^p\boldsymbol{N}(\theta)(x_1,\ldots, x_p),\\
\boldsymbol{d}\boldsymbol{N}(\theta), \boldsymbol{d}(\boldsymbol{D}_\theta \boldsymbol{N}(\theta)(x_1)), \ldots, \boldsymbol{d}(\boldsymbol{D}_\theta \boldsymbol{N}(\theta)(x_p)), \boldsymbol{d}(\boldsymbol{D}_\theta^2 \boldsymbol{N}(\theta)(x_1,x_2)),\ldots, \boldsymbol{d}(\boldsymbol{D}_\theta^2 \boldsymbol{N}(\theta)(x_{p-1},x_p)),\\
\boldsymbol{d}(\boldsymbol{D}_\theta^3 \boldsymbol{N}(\theta)(x_1,x_2,x_3)),\ldots,\boldsymbol{d}(\boldsymbol{D}_\theta^3 \boldsymbol{M}(\theta)(x_{p-2},x_{p-1},x_p)),\ldots, \boldsymbol{d}(\boldsymbol{D}_\theta^p\boldsymbol{N}(\theta)(x_1,\ldots, x_p))
\end{gather*}
and their $\mathsf{H}_0^m$ adjoints applied to $I_n$. All of these operators are bounded on $H_0^m$ and continuous in the operator norm at $\theta_0$ by the joint continuity of $D_\theta^iM(z,\theta)$, $\frac{d}{dz}D_\theta^iM(z,\theta)$, $D_\theta^iN(z,\theta)$, and $\frac{d}{dz}D_\theta^iN(z,\theta)$ for $i=0,\ldots, p$. Thus, $\boldsymbol{d}(\boldsymbol{D}_\theta^p\phi_{\boldsymbol{L}}(\theta))$ is continuous at $\theta_0$ in $H^m$ and the result follows by Lemma \ref{lem:ga}.
\end{proof}

\bibliographystyle{newapa}
\bibliography{spectral_ET3}

\end{document}